\documentclass[sigconf,nonacm, authorversion,screen]{acmart}
\usepackage{todonotes}
\usepackage[commandnameprefix=always]{changes}
\usepackage{paco-notations}
\usepackage{supertech-llncs-light}
\usepackage{flushend}
\usepackage{balance}
\usepackage{cleveref}
\usepackage{gsn}
\usepackage{booktabs}
\usepackage{dblfloatfix}
\usepackage{mathtools}
\usepackage{tabularx}
\usepackage{threeparttablex}

\setcopyright{rightsretained}
\punt{
\acmConference[SPAA '26] {ACM SPAA 2026}{July 6-10, 2026}{Royal Holloway, University of London, in London, UK}
}

\begin{document}

\title{PACO: A Fully Cache-Oblivious Parallel FFT with One Global Redistribution}
\punt{
\titlenote{This research is supported in part by Huawei Coorp under the contract no. TC20240528049}
}

\author{Shina Guo}
\affiliation{%
    \institution{College of Computer Science and Artificial Intelligence, Fudan University}
    \city{Shanghai}
    \country{P. R. China}
}
\email{25213050013@m.fudan.edu.cn}

\author{Weiguo Gao}
\affiliation{%
    \institution{School of Mathematical Sciences, School of Data Science, Fudan University}
    \city{Shanghai}
    \country{P. R. China}
}
\email{wggao@fudan.edu.cn}

\author{Yuan Tang}
\authornote{Corresponding Author; Also affiliated with Shanghai Key Lab. of Intelligent Information Processing.}
\affiliation{%
    \institution{College of Computer Science and Artificial Intelligence, Fudan University}
    \city{Shanghai}
    \country{P. R. China}
}
\email{yuantang@fudan.edu.cn}

\begin{abstract}
Fast Fourier transforms (FFTs) on parallel machines incur two distinct
forms of data movement: processor-local transfers through the memory
hierarchy and global redistribution between processors. Existing
four-step FFT organizations can switch the active transform dimension
with one global transpose-like exchange, but this property alone does
not yield cache-efficient local computation. Conversely,
cache-oblivious FFTs obtain asymptotically optimal local memory traffic
through recursively applied layout transformations, whose direct
materialization can require additional data-rearrangement passes and
global exchanges.

We present \algName{}, a fully cache-oblivious parallel FFT framework
that reconciles these objectives. PACO executes
\[
    \textsc{LocalFFT}
    \;\longrightarrow\;
    \textsc{OneGlobalPermutation}
    \;\longrightarrow\;
    \textsc{LocalFFT}.
\]
Its local stages recursively partition both transform and batch
dimensions without knowledge of the cache parameters. Rather than
materializing the transpose-like layouts induced by this recursion,
PACO defers them. We show that the deferred layouts compose into a
base-$b$ digit-reversal permutation, which PACO fuses with the
redistribution already required to change the local transform
dimension. The resulting middle stage is a perfectly balanced parallel
cache-oblivious digit-reversal permutation: each source--destination
processor pair exchanges exactly \(N/p^2\) elements (the \(p\) self-pairs incur no network traffic).

Under an exact base-$b$ slab decomposition in a hybrid ideal-cache/BSP
model, PACO computes an \(N\)-point DFT exactly with maximum per-processor work
\[
    \Theta\!\left(\frac{N}{p}\log N\right)
\]
and maximum per-processor cache complexity
\[
    \Theta\!\left(
        \frac{N}{pB}
        \left(1+\log_M N\right)
    \right).
\]
It uses exactly one global redistribution round. PACO returns
canonical logical DFT coefficients under a factor-swapped target slab
ownership. The single redistribution is necessary under the stated
no-replication ownership model, while its communication volume is
optimal for the prescribed fused permutation and source--target slab
distributions.
\end{abstract}

\maketitle

\secput{intro}{Introduction}

The fast Fourier transform (FFT) is a fundamental primitive in
scientific computing, signal processing, and data analysis. The
Cooley--Tukey algorithm reduces the arithmetic cost of an \(N\)-point
discrete Fourier transform (DFT) from \(\Theta(N^2)\) to
\(\Theta(N\log N)\) operations~\cite{CooleyTu65}. On modern parallel
machines, however, arithmetic is only one component of the execution
cost. An FFT must move data through each processor's local memory
hierarchy and redistribute data between processors when the active
transform dimension changes. Thus, an efficient parallel FFT should
simultaneously control processor-local cache traffic and
interprocessor communication.

Existing FFT organizations typically optimize one of these forms of
data movement more directly than the other. Four-step and six-step FFT
organizations reshape a one-dimensional transform into matrix views
and use transpose-like operations to switch the active transform
dimension~\cite{Bailey90,VanLoan92}. With compatible local kernels and
data distributions, a two-factor four-step organization can use one
global slab redistribution while producing the canonical DFT
coefficient order in a factor-swapped output view. This algebraic
one-exchange structure, however, does not by itself guarantee fully
cache-oblivious local execution. Conversely, cache-oblivious FFTs
obtain asymptotically optimal ideal-cache complexity through
recursively applied transpose-like layout transformations
~\cite{FrigoLePr99}. Directly materializing these
transformations in a distributed setting may introduce extra local
data passes and, when an affected digit field is distributed,
additional global redistributions. Other low-communication FFTs reduce
the number of global all-to-all operations through alternative
factorizations, approximation, oversampling, or additional
neighbor-level communication~\cite{EdelmanMcTo99,TangPaKi12}; their
numerical and distributional contracts differ from the exact FFT
setting considered here.

\paragrf{Problem setting: }
We consider an exact bounded-radix Cooley--Tukey FFT under an explicit
source and target ownership contract. Let
\[
    N=N_1N_2.
\]
The input is viewed as an \(N_1\times N_2\) row-major matrix
distributed by \(N_2\)-column slabs. The output is viewed as an
\(N_2\times N_1\) row-major matrix distributed by \(N_1\)-column
slabs. Its coordinate \((k_2,k_1)\) represents
\[
    y[k_1+k_2N_1].
\]
Hence the output has canonical logical DFT coefficient indices, but
its physical ownership is factor-swapped relative to the input
distribution. Restoring the initial slab ownership is not part of the
problem and may require a further redistribution.

This paper presents PACO, a Processor-Aware Cache-Oblivious FFT framework
\punt{
\footnote{PACO names our algorithm family; here
``Communication-Optimal'' is meant in the scoped sense made precise
in the \emph{Scope of optimality} paragraph following \thmref{main-result}. }
}
that combines fully cache-oblivious
local computation with exactly one global redistribution. PACO follows the fixed structure
\[
    \textsc{LocalFFT}
    \;\longrightarrow\;
    \textsc{OneGlobalPermutation}
    \;\longrightarrow\;
    \textsc{LocalFFT}.
\]
Its local stages use a cache-oblivious batched FFT scheduler that
recursively partitions both transform and batch dimensions without
knowledge of the cache capacity \(M\) or block size \(B\).

The central observation is that a transpose-free cache-oblivious local
FFT does not leave its result in arbitrary disorder. Instead, the
layout transformations deferred during local recursion compose into a
structured base-\(b\) digit-reversal permutation. PACO materializes
this complete permutation exactly once, at the redistribution already
required between the two top-level Cooley--Tukey transform groups.
Thus, the middle stage simultaneously corrects the first local stage's
deferred layout, swaps the two factor fields, prepares the layout
consumed by the second local stage, and transfers each item to its
factor-swapped target owner. PACO does not introduce a separate global
communication phase merely to repair digit order.

\paragrf{Contributions.}
This paper makes the following contributions:
\begin{itemize}
    \item
    We design a fully cache-oblivious batched Cooley--Tukey FFT
    scheduler that recursively partitions transform and batch
    dimensions, avoids recursion-level physical transposes, and
    achieves asymptotically optimal local work and cache complexity
    under the stated memory model.

    \item
    We show that the lazy layouts induced by the local recursion
    compose into \(\proc{PARA-CO-DRP}_b\), a base-\(b\) PARAllel Cache-Oblivious Digit-Reversal Permutation, and we fuse
    that permutation with the ownership change required between the
    two top-level transform groups. The resulting
    \(\proc{PARA-CO-DRP}_b\) middle stage uses exactly one global
    redistribution and routes exactly \(N/p^2\) items from every source
    processor to every destination processor.

    \item
    Under an exact base-\(b\) slab decomposition in a hybrid
    ideal-cache/BSP model, we prove maximum per-processor work
    \[
        \Theta\!\left(
            \frac{N}{p}\log N
        \right)
    \]
    and maximum per-processor cache complexity
    \[
        \Theta\!\left(
            \frac{N}{pB}
            \left(
                1+\log_M N
            \right)
        \right).
    \]
    We further prove that one global redistribution round is necessary
    under the no-replication ownership model, and that PACO's
    communication volume is optimal for the prescribed fused
    slab-to-slab permutation.
\end{itemize}

\punt{
Our migration-optimality claim applies to the prescribed fused
permutation and fixed source--target slab ownerships; it does not
assert communication-volume optimality for all distributed FFT
algorithms with different intermediate ownership schedules.
}
\secput{model}{Model and Main Result}

We study the canonical length-\(N\) discrete Fourier transform
\[
    y
    =
    \Omega_N x,
    \qquad
    (\Omega_N)_{k,j}
    =
    \omega_N^{-kj},
    \qquad
    \omega_N
    =
    \exp\!\left(
        \frac{2\pi\imath}{N}
    \right).
\]
Let \([n]\coloneqq\{0,1,\ldots,n-1\}\). One complex scalar is one
item in both the cache and communication models.

\paragrf{Problem and ownership contract.}
Fix a constant radix
\[
    b
    =
    2^t
    =
    O(1).
\]
We assume
\[
    N=N_1N_2,\qquad
    N_1=b^{m_1},\qquad
    N_2=b^{m_2},
\]
and write
\[
    m:=m_1+m_2=\log_b N.
\]
There are \(p\ge2\) processors, indexed by \([p]\). For the exact
balance results, \(p\), \(N_1\), and \(N_2\) are powers of \(b\), and
\[
    p\mid N_1,
    \qquad
    p\mid N_2.
    \label{eq:exact-divisibility}
\]

The input vector is viewed as an \(N_1\times N_2\) row-major matrix
\[
    X[j_1,j_2]
    =
    x[j_1 N_2+j_2],
    \qquad
    j_1\in[N_1],
    \quad
    j_2\in[N_2],
\]
distributed by equal column slabs:
\begin{equation}
    \operatorname{owner}_0(j_1,j_2)
    =
    \left\lfloor
        \frac{j_2 p}{N_2}
    \right\rfloor.
    \label{eq:input-slab-owner}
\end{equation}
Thus, processor \(q\) initially owns the columns
\[
    q\frac{N_2}{p}
    \le j_2 <
    (q+1)\frac{N_2}{p}.
\]

We factor the canonical frequency index as
\[
    k
    =
    k_1+k_2N_1,
    \qquad
    k_1\in[N_1],
    \quad
    k_2\in[N_2],
\]
and view the output as an \(N_2\times N_1\) row-major matrix
\begin{equation}
    Y[k_2,k_1]
    =
    y[k_1+k_2N_1].
    \label{eq:factor-swapped-output-view}
\end{equation}
The target ownership is the column-slab distribution
\begin{equation}
    \operatorname{owner}_1(k_2,k_1)
    =
    \left\lfloor
        \frac{k_1p}{N_1}
    \right\rfloor.
    \label{eq:output-slab-owner}
\end{equation}
Hence, PACO returns canonical DFT coefficients \(y[k]\), but under a
factor-swapped target ownership. It does not restore the initial slab
distribution.

\paragrf{Machine and execution model.}
Each processor has a private fully associative ideal cache of \(M\)
items and block size \(B\), with optimal replacement and initially
empty cache. We assume the tall-cache condition
\begin{equation}
    M
    =
    \Omega(B^2)
    \label{eq:tall-cache}
\end{equation}
and the slab-width conditions
\begin{equation}
    \frac{N_1}{p}
    \ge B,
    \qquad
    \frac{N_2}{p}
    \ge B.
    \label{eq:slab-width-cache-block}
\end{equation}

Interprocessor communication follows the BSP model. In one BSP
\(h\)-relation, every processor sends and receives at most \(h\)
remote items; its cost is \(gh+L\), where \(g\) is the per-item gap
and \(L\) is the synchronization cost. A \emph{global redistribution
round} consists of local packing, one BSP \(h\)-relation, and local
unpacking. Packing and unpacking are charged to local work and cache
misses, not to \(h\).

We use an owner-computes model: input values are initially
unreplicated, no processor initially owns the entire input, and a
processor obtains a remotely owned value only through an explicit BSP
communication superstep. Arithmetic vertices of the fixed-radix
Cooley--Tukey DAG are not recomputed.

Let \(W_{\max}\) and \(Q_{\max}\) denote,
respectively, the maximum per-processor local work and cache misses.
Let \(R\) be the number of global redistribution rounds, \(h\) the
maximum per-processor remote volume in a round, and \(\mu\) the
aggregate number of remotely transmitted items.

\begin{theorem}[Main result]
\label{thm:main-result}
Under assumptions
\eqref{input-slab-owner}--\eqref{slab-width-cache-block}, PACO
computes the canonical length-\(N\) DFT under the input ownership
\eqref{input-slab-owner} and factor-swapped target ownership
\eqref{output-slab-owner}. In particular,
\[
    Y[k_2,k_1]
    =
    (\Omega_Nx)[k_1+k_2N_1]
\]
for every \(k_1\in[N_1]\) and \(k_2\in[N_2]\).

Its maximum per-processor work and cache complexity are
\[
    W_{\max}
    =
    \Theta\!\left(
        \frac{N}{p}\log N
    \right)
\]
and
\[
    Q_{\max}
    =
    \Theta\!\left(
        \frac{N}{pB}
        \left(
            1+\log_M N
        \right)
    \right).
\]
PACO performs exactly one global redistribution round,
\[
    R
    =
    1,
\]
whose BSP volume and aggregate remote migration volume are
\[
    h
    =
    \frac{N}{p}
    \left(
        1-\frac{1}{p}
    \right),
    \qquad
    \mu
    =
    N
    \left(
        1-\frac{1}{p}
    \right).
\]
\end{theorem}

\paragrf{Scope of optimality.}
The optimality claims in Theorem~\ref{thm:main-result} have different
scopes. The one-round bound is optimal for the complete DFT under the
stated owner-computes and no-replication assumptions. The work bound
is optimal for the fixed-radix Cooley--Tukey computation DAG without
arithmetic recomputation. The cache bound matches the lower bound for
the local FFT DAGs assigned by the PACO decomposition. Finally, the
migration bound is optimal for the prescribed fused digit-reversal
permutation between the source and target slab ownerships; it is not
a communication-volume lower bound for all distributed FFT schedules
with different intermediate ownerships.

The next section develops the two-factor decomposition and the layout
algebra that allows PACO to fuse the complete digit reversal with its
sole global redistribution.

\secput{overview}{PACO: One Redistribution from Lazy Layouts}

PACO combines two transpose-free cache-oblivious local FFT phases with
one global redistribution:
\[
    \textsc{LocalFFT}
    \;\longrightarrow\;
    \textsc{OneGlobalPermutation}
    \;\longrightarrow\;
    \textsc{LocalFFT}.
\]
This section explains why the layout changes deferred by the local
recursions can be materialized exactly at the ownership change already
required between the two top-level Cooley--Tukey transform groups.
\punt{
The fully cache-oblivious local scheduler is developed in
\secref{local-co-fft}, and the cache-oblivious implementation and
communication analysis of the global permutation are given in
\secref{global-permutation}.
}

\subsection{Two-Factor Cooley--Tukey Decomposition}
\label{sec:two-factor-ct}

Recall that
\[
    N
    =
    N_1N_2.
\]
Writing
\[
    j
    =
    j_1N_2+j_2,
    \qquad
    k
    =
    k_1+k_2N_1,
\]
where
\[
    j_1,k_1\in[N_1],
    \qquad
    j_2,k_2\in[N_2],
\]
gives
\begin{align}
    y[k_1+k_2N_1]
    ={}&
    \sum_{j_2=0}^{N_2-1}
    \left(
        \sum_{j_1=0}^{N_1-1}
        x[j_1N_2+j_2]
        \omega_{N_1}^{-k_1j_1}
    \right)
    \notag\\
    &\qquad\qquad\cdot
    \omega_N^{-k_1j_2}
    \omega_{N_2}^{-k_2j_2}.
    \label{eq:two-factor-ct}
\end{align}

Thus, Phase~I performs one length-\(N_1\) FFT for every
\(j_2\), Phase~II applies the corresponding twiddle factor and
redistributes the intermediate values, and Phase~III performs one
length-\(N_2\) FFT for every \(k_1\). Under the input ownership
\eqref{input-slab-owner}, every Phase-I transform is local to one
processor. Under the target ownership
\eqref{output-slab-owner}, every Phase-III transform is local to one
processor. The purpose of the middle stage is therefore not to create
a new transform factorization, but to change ownership and layout
between these two local transform groups.

\subsection{Lazy Layouts of the Local FFT}
\label{sec:lazy-layout-fft}

For a permutation
\[
    \pi:[n]\longrightarrow[n],
\]
let \(P_\pi\) denote the physical-to-logical layout operator
\[
    (P_\pi z)[j]
    =
    z[\pi(j)].
\]
For \(0\le\ell\le m\), let \(\pi_\ell\) denote base-\(b\) digit
reversal on \(\ell\)-digit strings. In particular, if
\[
    u
    =
    d_{\ell-1}
    \pplus
    \cdots
    \pplus
    d_0,
\]
then
\[
    \pi_\ell(u)
    =
    d_0
    \pplus
    \cdots
    \pplus
    d_{\ell-1}.
\]
Digit reversal is an involution.

\begin{lemma}[Lazy-layout FFT]
\label{lem:lazy-layout-fft}
Consider a transpose-free local radix-\(b\) FFT of length \(b^\ell\)
with natural-order constant-size leaves and twiddle factors evaluated
from logical coordinates. Its output-layout form satisfies
\begin{equation}
    \mathcal A_\ell
    =
    P_{\pi_\ell}\Omega_{b^\ell}.
    \label{eq:output-layout-local-operator}
\end{equation}
Conversely, the same recursion can consume the digit-reversed input
layout and produce a natural-order DFT output:
\begin{equation}
    \mathcal B_\ell
    =
    \Omega_{b^\ell}P_{\pi_\ell}.
    \label{eq:input-layout-local-operator}
\end{equation}
\end{lemma}

\begin{proof}[Proof sketch]
A recursive Cooley--Tukey step consists of two child-transform groups,
a diagonal twiddle operator, and a final swap of two transform digit
fields. PACO performs the child transforms and twiddle multiplication,
but does not physically materialize the final digit-field swap.
Recursively deferred swaps compose into base-\(b\) digit reversal:
for every split
\[
    \ell
    =
    \ell_1+\ell_2,
\]
the physical digit fields \(U\pplus V\) represent the logical fields
\[
    \pi_{\ell_2}(V)
    \pplus
    \pi_{\ell_1}(U).
\]
Twiddle factors remain correct because they are evaluated using the
logical coordinates represented by each physical item (see the layout-aware twiddle operator, \appref{twiddle}); equivalently,
the diagonal twiddle operator is conjugated by the current layout
permutation. Since digit reversal is an involution, conjugating the
output-layout form yields the input-layout form. The complete
induction and digit-field view invariant are given in the appendix.
\end{proof}

Lemma~\ref{lem:lazy-layout-fft} shows that a local PACO FFT does not
leave its result in arbitrary disorder. It produces a precise and
composable layout: base-\(b\) digit reversal. PACO exploits this
structure rather than materializing a separate local transpose or
digit-reversal pass after every recursive factorization.

\subsection{One Fused Middle Permutation}
\label{sec:full-middle-permutation}
Phase I applies \(\mathcal A_{m_1}\) to the locally owned length-\(N_1\)
transforms indexed by \(j_2\).  After Phase I, let
\((a,c)\) denote a physical array coordinate, where \(a\in[N_1]\)
is the physical row coordinate and \(c\in[N_2]\) is the unchanged
second-stage input coordinate.  Thus \(c=j_2\), whereas \(a\) is not
the original coordinate \(j_1\): it represents the logical first-stage
frequency
\[
    k_1=\pi_{m_1}(a).
\]

When the item at physical coordinate \((a,c)\) is packed for the
global redistribution, \algName{} applies the top-level twiddle factor
\[
    \omega_N^{-k_1 c}
\]
using its represented logical coordinates \(k_1=\pi_{m_1}(a)\) and
\(c=j_2\).

\algName{} then materializes the complete \(m\)-digit reversal, rather than
only the \(m_1\)-digit layout exposed by Phase~I. For digit fields
\[
    a\pplus c,
    \qquad
    |a|
    =
    m_1,
    \qquad
    |c|
    =
    m_2,
\]
the middle permutation is
\begin{equation}
    \pi_m(a\pplus c)
    =
    \pi_{m_2}(c)
    \pplus
    \pi_{m_1}(a).
    \label{eq:full-middle-digit-reversal}
\end{equation}

Equation~\eqref{full-middle-digit-reversal} performs four operations
simultaneously:
\begin{enumerate}
    \item
    it cancels the \(\pi_{m_1}\) output layout exposed by Phase~I;

    \item
    it swaps the \(N_1\)- and \(N_2\)-digit fields;

    \item
    it establishes the \(\pi_{m_2}\) input layout consumed by
    Phase~III; and

    \item
    it changes ownership from \(\operatorname{owner}_0\) to
    \(\operatorname{owner}_1\), making every complete
    length-\(N_2\) Phase-III transform processor local.
\end{enumerate}

More concretely, an item at physical Phase-I coordinate
\((a,c)\) is routed to the factor-swapped coordinate
\[
    \left(
        \pi_{m_2}(c),
        \pi_{m_1}(a)
    \right).
\]
Its destination processor is therefore determined by the second
coordinate,
\[
    \left\lfloor
        \frac{\pi_{m_1}(a)p}{N_1}
    \right\rfloor.
\]
The detailed packing, routing, local transposition, and unpacking
procedure that realizes this mapping cache-obliviously is given in
\secref{fused-redistribution}. Its exact source--destination balance is
analyzed there as well.

The ownership change is required by the data dependencies of the
second transform group: for a fixed \(k_1\), a length-\(N_2\) FFT
requires values from all \(j_2\in[N_2]\), whereas the input slabs
distribute those coordinates among the processors. PACO therefore
uses the unavoidable communication boundary to materialize the full
deferred layout transformation, rather than introducing a separate
redistribution solely to repair digit order.

\subsection{End-to-End Correctness}
\label{sec:three-phase-paco}

PACO applies the output-layout local operator of
Lemma~\ref{lem:lazy-layout-fft} in Phase~I, the fused permutation
\eqref{full-middle-digit-reversal} together with the top-level
twiddle diagonal in Phase~II, and the input-layout local operator in
Phase~III. 
Identifying the input array with $\mathbb{C}^{N_1}\otimes\mathbb{C}^{N_2}$
(physical coordinate $(a,c)$, $a\in[N_1]$, $c\in[N_2]$) and the output array
with $\mathbb{C}^{N_2}\otimes\mathbb{C}^{N_1}$ (physical coordinate $(u,v)$,
$u\in[N_2]$, $v\in[N_1]$), PACO's operator composition is
\begin{equation}
  \mathcal{F}
  \;=\;
  \bigl(\mathcal{B}_{m_2}\otimes I_{N_1}\bigr)\,
  P_{\pi_m}\,
  D_{\mathrm{tw}}\,
  \bigl(\mathcal{A}_{m_1}\otimes I_{N_2}\bigr), \label{eq:paco-composite-op}
\end{equation}
where $\mathcal{A}_{m_1}\otimes I_{N_2}$ applies the output-layout length-$N_1$
transform $\mathcal{A}_{m_1}=P_{\pi_{m_1}}\Omega_{N_1}$ to the row factor
(batched over the $N_2$ columns); $D_{\mathrm{tw}}$ is the diagonal top-level
twiddle $(D_{\mathrm{tw}})_{(a,c)}=\omega_N^{-\pi_{m_1}(a)\,c}$, evaluated from
the logical coordinate $k_1=\pi_{m_1}(a)$ carried at physical row $a$;
$P_{\pi_m}\colon\mathbb{C}^{N_1}\otimes\mathbb{C}^{N_2}\to
\mathbb{C}^{N_2}\otimes\mathbb{C}^{N_1}$ is the base-$b$ digit reversal~\eqref{full-middle-digit-reversal},
which simultaneously cancels $P_{\pi_{m_1}}$, swaps the two factors, and installs
$P_{\pi_{m_2}}$; and $\mathcal{B}_{m_2}\otimes I_{N_1}$ applies the input-layout
length-$N_2$ transform $\mathcal{B}_{m_2}=\Omega_{N_2}P_{\pi_{m_2}}$ to the row
factor (batched over the $N_1$ columns).

\punt{
Schematically, its operator composition is
\begin{equation}
    \mathcal B_{m_2}
    \;\circ\;
    P_{\pi_m}
    \;\circ\;
    \mathcal D_{\mathrm{tw}}
    \;\circ\;
    \mathcal A_{m_1},
    \label{eq:paco-operator-composition}
\end{equation}
where the local operators act on their respective batched transform
fields and \(\mathcal D_{\mathrm{tw}}\) denotes multiplication by the
top-level Cooley--Tukey twiddle factors.
}

\begin{theorem}[End-to-end correctness]
\label{thm:paco-correctness}
Under the ownership assumptions of \secref{model}, PACO computes the
canonical DFT
\[
    y
    =
    \Omega_Nx.
\]
More precisely, its factor-swapped output view satisfies
\[
    Y[k_2,k_1]
    =
    y[k_1+k_2N_1]
\]
for every
\[
    k_1\in[N_1],
    \qquad
    k_2\in[N_2],
\]
and this coefficient is owned by
\[
    \operatorname{owner}_1(k_2,k_1)
    =
    \left\lfloor
        \frac{k_1p}{N_1}
    \right\rfloor.
\]
\end{theorem}

\begin{proof}
Write the input array as $\mathbb{C}^{N_1}\otimes\mathbb{C}^{N_2}$ and the
output array as $\mathbb{C}^{N_2}\otimes\mathbb{C}^{N_1}$, as in~\eqref{paco-composite-op}. By
\lemref{lazy-layout-fft} (in batched form), Phase~I realizes
$\mathcal{A}_{m_1}\otimes I_{N_2}=(P_{\pi_{m_1}}\otimes I_{N_2})(\Omega_{N_1}\otimes I_{N_2})$
and Phase~III realizes
$\mathcal{B}_{m_2}\otimes I_{N_1}=(\Omega_{N_2}\otimes I_{N_1})(P_{\pi_{m_2}}\otimes I_{N_1})$.

The middle stage factors as
\[
  P_{\pi_m}
  =(P_{\pi_{m_2}}\otimes I_{N_1})\,\Pi_{\mathrm{swap}}\,(P_{\pi_{m_1}}\otimes I_{N_2}),
\]
where $\Pi_{\mathrm{swap}}\colon\mathbb{C}^{N_1}\otimes\mathbb{C}^{N_2}\to
\mathbb{C}^{N_2}\otimes\mathbb{C}^{N_1}$ is the factor transpose
$(\Pi_{\mathrm{swap}}z)[u,v]=z[v,u]$. Indeed, both sides send physical $(u,v)$ to
$z[\pi_{m_1}(v),\pi_{m_2}(u)]$, since reversing the $m$-digit string
$u\!+\!\!+\!v$ yields $\pi_{m_1}(v)\!+\!\!+\!\pi_{m_2}(u)$.
Likewise the layout-aware twiddle diagonal satisfies
\[
  D_{\mathrm{tw}}
  =(P_{\pi_{m_1}}\otimes I_{N_2})\,\widetilde D\,(P_{\pi_{m_1}}\otimes I_{N_2}),
  \qquad \widetilde D_{(k_1,j_2)}=\omega_N^{-k_1 j_2},
\]
because conjugating a diagonal by the involution $P_{\pi_{m_1}}\otimes I_{N_2}$
merely relabels its entries to $\widetilde D_{(\pi_{m_1}(a),c)}=\omega_N^{-\pi_{m_1}(a)c}$.

Substituting these identities into~(10) and cancelling each adjacent pair
$(P_{\pi_{m_1}}\otimes I_{N_2})^2=I$ and $(P_{\pi_{m_2}}\otimes I_{N_1})^2=I$ gives
\[
  \mathcal{F}
  =(\Omega_{N_2}\otimes I_{N_1})\,\Pi_{\mathrm{swap}}\,\widetilde D\,
   (\Omega_{N_1}\otimes I_{N_2}).
\]
This is the two-factor four-step operator. Evaluating it on the input $x$ reproduces the
right-hand side of~\eqref{two-factor-ct}:
\[
  (\mathcal{F}x)[k_2,k_1]
  =\sum_{j_2=0}^{N_2-1}\!\Bigl(\sum_{j_1=0}^{N_1-1}
      x[j_1N_2+j_2]\,\omega_{N_1}^{-k_1 j_1}\Bigr)
      \omega_N^{-k_1 j_2}\,\omega_{N_2}^{-k_2 j_2}
  = y[k_1+k_2N_1].
\]
Hence $Y[k_2,k_1]=y[k_1+k_2N_1]$ for all $k_1\in[N_1]$, $k_2\in[N_2]$.

Finally, the column factor $v$ is untouched by Phase~III (which transforms only
the row factor) and equals $\pi_{m_1}(a)=k_1$.
Thus $Y[k_2,k_1]$ resides in physical column $v=k_1$ of the output
$N_2\times N_1$ slab, whose owner under~\eqref{output-slab-owner} is
$\lfloor k_1 p/N_1\rfloor=\mathrm{owner}_1(k_2,k_1)$.
\end{proof}

\punt{
\begin{proof}
Phase~I computes the length-\(N_1\) DFTs under the output layout
\(P_{\pi_{m_1}}\). During Phase~II, the top-level twiddle is applied
using the represented logical frequency coordinate, and
\eqref{full-middle-digit-reversal}, both cancels the Phase-I layout and
places each fixed-\(k_1\) intermediate vector under the
\(P_{\pi_{m_2}}\) input layout. Phase~III applies
\[
    \mathcal B_{m_2}
    =
    \Omega_{N_2}P_{\pi_{m_2}},
\]
thereby producing the natural-order length-\(N_2\) transforms.
Together with the two-factor identity
\eqref{two-factor-ct}, this is exactly the canonical DFT
\(\Omega_Nx\). The stated ownership follows from
\eqref{output-slab-owner}.
The operator identities are proved in \appref{layout}; the procedure-level realization in \appref{line-by-line}.
\end{proof}
}

\secput{local-co-fft}{Cache-Oblivious Local FFT and Fused Redistribution}

PACO is implemented by two cache-oblivious building blocks. The first
is a batched local FFT that carries recursive layout changes lazily,
without materializing recursion-level transposes. The second is a
fused redistribution that applies the top-level twiddles, materializes
the complete \punt{middle} digit reversal, and changes ownership in one BSP
communication round.

\subsection{Cache-Oblivious Batched Local FFT}
\label{sec:fully-cache-oblivious-scheduler}

A local PACO subproblem consists of \(s\) independent transforms of
length
\(
    n
    =
    b^\ell.
\)
The transforms are represented by recursive digit-field views rather
than copied or transposed subarrays. One batch field remains the least
significant physical address field; varying this field accesses a
contiguous memory run. This distinguished field is preserved whenever
the transform recursion creates new batch coordinates.

The scheduler uses the relative sizes of \(n\) and \(s\), but does not
inspect either \(M\) or \(B\). If
\[
    s
    \ge
    n,
\]
it recursively splits one batch digit, thereby partitioning the
transform instances into \(b\) data-disjoint subbatches. It first
consumes batch digits introduced by previous transform splits and
splits the distinguished contiguous batch field only when no such
outer batch digit remains. 
If
\[
    s
    <
    n,
\]
the scheduler performs a balanced digit-aligned Cooley--Tukey split.
Specifically, it writes
\[
    n
    =
    n_1n_2,
    \qquad
    n_1
    =
    b^{\lceil\ell/2\rceil},
    \qquad
    n_2
    =
    b^{\lfloor\ell/2\rfloor},
    \qquad
    \ell
    =
    \log_b n,
\]
so that
\[
    1
    \le
    \frac{n_1}{n_2}
    \le
    b.
\]
Writing a logical transform coordinate as
\[
    x
    =
    u n_2+v,
    \qquad
    u\in[n_1],
    \quad
    v\in[n_2],
\]
and denoting an existing batch coordinate by \(\beta\in[t]\), the
first child group consists of \(s n_2\) transforms of length \(n_1\).
Its transform coordinate is \(u\), while \(v\) becomes a new outer
batch coordinate: its batched transforms are indexed by
\[
    (v,\beta).
\]
Thus, the \(n_2\)-digit field precedes the pre-existing batch digits in
the child view, but the least-significant contiguous batch field
within \(\beta\) is unchanged.

Symmetrically, the second child group consists of \(s n_1\) transforms
of length \(n_2\), with transform coordinate \(v\) and batch coordinate
\[
    (u,\beta).
\]
The associated Cooley--Tukey twiddle
\[
    \omega_n^{-uv}
\]
is evaluated from these logical coordinates and is multiplied when
the second child group first accesses the corresponding intermediate
value. Hence, transform splitting changes only the recursive view:
it creates an outer batch field but materializes neither a transpose
nor a separate twiddle pass.

\punt{
If
\[
    t
    <
    n,
\]
it splits the transform digit field into two balanced parts. The
resulting Cooley--Tukey child groups are interpreted as larger batches
of shorter transforms: one transform digit field becomes an outer
batch field, but no data are moved and the contiguous batch field is
unchanged.
}

The top-level twiddle factors of each recursive Cooley--Tukey step are
evaluated from the logical digit coordinates represented by the view.
They are fused into the first accesses of the second child-transform
group. At constant-size leaves, PACO invokes a
natural-order local DFT kernel. In output-layout mode, the resulting
batched operator is
\[
    I_s
    \otimes
    \left(
        P_{\pi_\ell}\Omega_n
    \right);
\]
in input-layout mode, it is
\[
    I_s
    \otimes
    \left(
        \Omega_n P_{\pi_\ell}
    \right).
\]
Thus, the scheduler realizes the lazy-layout operators of
Lemma~\ref{lem:lazy-layout-fft} on every transform in the batch.

\begin{proposition}[Cache-oblivious batched local FFT]
\label{prop:local-batched-fft}
For \(s\) transforms of length \(n=b^\ell\), the local PACO scheduler
performs
\begin{equation}
    W_{\mathrm{local}}(n,s)
    =
    O\!\left(
        ns
        \left(
            1+\log n
        \right)
    \right)
    \label{eq:local-scheduler-work}
\end{equation}
local work and uses
\[
    O(\log(ns))
\]
recursion-stack words. Under the tall-cache assumption
\[
    M
    =
    \Omega(B^2),
\]
and for an initial contiguous batch width at least \(B\), it incurs
\begin{equation}
    Q_{\mathrm{local}}(n,s)
    =
    O\!\left(
        \frac{ns}{B}
        \left(
            1+\log_M n
        \right)
    \right)
    \label{eq:local-cache-upper-bound}
\end{equation}
cache misses.
\end{proposition}

\begin{proof}[Proof sketch]
Fix a sufficiently small constant \(\alpha>0\), let
\[
    M^\star
    =
    \alpha M,
\]
and consider the maximal recursive views satisfying
\[
    ns
    \le
    M^\star.
\]
This is only an analytical cache-fitting frontier: the algorithm itself
never tests this condition.

\paragrf{Frontier geometry.}
A transform split preserves the working-set size of each child group.
Indeed, if
\[
    n
    =
    n_1n_2,
\]
then its two child groups have parameters
\[
    (n_1,sn_2)
    \qquad\text{and}\qquad
    (n_2,sn_1),
\]
and both have volume
\[
    n_1(sn_2)
    =
    n_2(sn_1)
    =
    ns.
\]
Consequently, a non-cache-fitting view can first enter the
cache-fitting region only through a batch split. Every nonroot maximal
frontier view therefore satisfies
\begin{equation}
    \frac{M^\star}{b}
    <
    ns
    \le
    M^\star,
    \label{eq:local-frontier-volume}
\end{equation}
because its parent has the same transform length and \(b\) times as
many batch instances.

We next bound the transform length at a frontier reached after at
least one transform split. A transform split is executed only when
\(s<n\). For a balanced digit-aligned factorization
\[
    n
    =
    n_1n_2,
    \qquad
    1
    \le
    \frac{n_1}{n_2}
    \le
    b,
\]
the first child has parameters \((n_1,sn_2)\eqqcolon (n', s')\). Since \(s<n_1n_2\),
\[
    sn_2
    <
    n_1n_2^2
    \le
    b^2n_1^3. 
\]
Symmetrically, the second child satisfies
\[
    sn_1
    <
    b^2n_2^3.
\]
Thus, after every transform split, and hence after any later batch split
 (a batch split only decreases \(s\) without changing \(n\)), the reachable state satisfies the following invariant:
\begin{equation}
    s
    \le
    b^2 n^3.
    \label{eq:local-reachable-state}
\end{equation}
Combining \eqref{local-frontier-volume} and
\eqref{local-reachable-state} gives
\[
    \frac{M^\star}{b}
    <
    ns
    \le
    b^2n^4,
\]
and therefore every frontier reached after at least one transform split has
\begin{equation}
    n
    =
    \Omega(M^{1/4}).
    \label{eq:local-frontier-transform-lower-bound}
\end{equation}

\paragrf{Transform-induced frontier expansion.}
Let \(K\) be the maximum number of transform splits on a
root-to-frontier path. If the active transform length before one split
is \(n_j=b^{d_j}\), then the larger child has length at most
\[
    b^{\lceil d_j/2\rceil}
    \le
    \sqrt{b n_j}.
\]
Induction gives
\[
    n_K
    \le
    b^{1-1/2^K}n^{1/2^K}
    <
    b\,n^{1/2^K}.
\]
For \(K\ge1\), combining this inequality with
\eqref{local-frontier-transform-lower-bound} yields
\[
    c_bM^{1/4}
    \le
    b\,n^{1/2^K}
\]
for a constant \(c_b>0\) depending only on \(b\) and \(\alpha\).
Hence,
\[
    2^K
    =
    O\!\left(
        1+\log_M n
    \right).
\]
A batch split partitions working-set mass among its children, whereas
a transform split creates two child groups, each of the parent's mass.
Thus, the total mass represented at the cache-fitting frontier is at
most
\[
    O\!\left(
        ns\,2^K
    \right)
    =
    O\!\left(
        ns
        \left(
            1+\log_M n
        \right)
    \right).
\]

\paragrf{Spatial locality at the frontier.} A separate argument (\lemref{B-width} in \appref{B-width-at-frontier}) shows that whenever the distinguished contiguous field is split,
the tall-cache assumption forces the resulting run width to satisfy
$c=\Omega(\sqrt M)=\Omega(B)$; if it is never split it retains its initial width
$\ge B$ by~\eqref{slab-width-cache-block}. Hence a frontier view of mass $ns$ is covered by $ns/c$ runs of
length $\Omega(B)$ and touches $O(ns/B)$ blocks, which remain resident during its
depth-first descendant execution. Summing over the frontier yields
\[
Q_{\mathrm{local}}(n,s)=O\left(\tfrac{ns}{B}(1+\log_M n)\right).
\]

\punt{
\paragrf{Spatial locality at the frontier.}
It remains to show that a frontier view does not fragment into too
many short runs. If the distinguished contiguous field is never split
on the root-to-frontier path, then its width remains at least its
initial width, which is at least \(B\) \eqref{slab-width-cache-block}. Otherwise, consider the last
split of that field, and let \(c\) be its child width. Its parent width
is \(bc\). This split occurs only when no outer batch field remains,
so the parent batch size is \(bc\), and the batch-split condition gives
\[
    n
    \le
    bc.
\]
If this child is already cache fitting, then its parent is not, and
therefore
\[
    M^\star
    <
    nbc
    \le
    b^2c^2.
\]
If it is not yet cache fitting, then
\[
    M^\star
    <
    nc
    \le
    bc^2.
\]
In either case,
\[
    c
    =
    \Omega(\sqrt M)
    =
    \Omega(B)
\]
by the tall-cache assumption. Subsequent outer-batch and transform
splits do not change \(c\).

A frontier view of volume \(ns\) is therefore covered by
\(ns/c\) contiguous runs of length \(c=\Omega(B)\), and intersects
\[
    O\!\left(
        \frac{ns}{B}
    \right)
\]
cache blocks. Its complete descendant execution accesses only this
frontier working set, so these blocks remain resident during the
depth-first execution. Summing the block costs over the frontier gives
\[
    Q_{\mathrm{local}}(n,s)
    =
    O\!\left(
        \frac{ns}{B}
        \left(
            1+\log_M n
        \right)
    \right).
\]
}

The work and stack bounds follow from the fixed-radix
Cooley--Tukey recurrence and the digit-aligned recursion depth. More details in \appref{local-fft}.
\end{proof}

\punt{
\begin{proof}[Proof sketch]
For the cache analysis, consider the maximal recursive views whose
working-set size is at most a sufficiently small constant fraction of
\(M\). A transform split preserves the working-set size of each child
group: if the parent has parameters \((n,s)\), then its two child
groups have parameters \((n_1,sn_2)\) and \((n_2,sn_1)\), both of
volume \(ns\). It can therefore increase the total frontier mass, but
only by creating a second transform group.

After \(K\) balanced digit-aligned transform splits on one path, the
active transform length is at most
\[
    b\,n^{1/2^K}.
\]
Conversely, every nontrivial cache-frontier view reached after a
transform split has transform length
\(\Omega(M^{1/4})\): the scheduler can enter the cache-fitting region
only through a batch split, while every state reached after a transform
split satisfies \(s=O(n^3)\). Hence
\[
    2^K
    =
    O\!\left(
        1+\log_M n
    \right).
\]
Thus, the total mass represented at the cache-fitting frontier is at
most
\[
    O\!\left(
        ns(1+\log_M n)
    \right).
\]

Spatial locality is supplied by the distinguished contiguous batch
field. If this field is never split before the frontier, its initial
width is at least \(B\). Otherwise, consider its last split. The
scheduler splits that field only when no outer batch digits remain and
the batch size is at least the transform length. Combining this
condition with the fact that the parent is not cache fitting gives a
post-split contiguous width \(\Omega(\sqrt{M})\), hence
\(\Omega(B)\) by the tall-cache assumption. Every frontier view of
volume \(ns\) therefore intersects \(O(ns/B)\) blocks, and all of its
descendants reuse this same cache-resident footprint. Summing over the
frontier proves the bound.
\end{proof}
}

\punt{
\begin{proof}[Proof sketch]
Transform splits preserve the subproblem volume and create temporal
locality: once a recursive view fits in cache, all of its descendant
transforms and fused twiddle accesses reuse the same resident data.
The number of transform-induced expansions at the analytical
cache-fitting frontier is
\[
    O(1+\log_M n).
\]

Batch splits partition the data without duplicating work. Moreover,
the digit-field discipline preserves a contiguous batch field of width
\(\Omega(B)\) at the cache-fitting frontier. Hence, a frontier view of
volume \(nt\) intersects
\[
    O(nt/B)
\]
cache blocks. Summing over the frontier gives
\eqref{local-cache-upper-bound}. The complete view invariant and
cache-frontier argument are deferred to the appendix.
\end{proof}
}

For PACO Phase~I, Proposition~\ref{prop:local-batched-fft} applies
with
\[
    n
    =
    N_1,
    \qquad
    s
    =
    \frac{N_2}{p}.
\]
For Phase~III, it applies with
\[
    n
    =
    N_2,
    \qquad
    s
    =
    \frac{N_1}{p}.
\]
The two local phases therefore incur, per processor,
\begin{equation}
    O\!\left(
        \frac{N}{pB}
        \left(
            1+\log_M N
        \right)
    \right)
    \label{eq:two-local-phases-cache-bound}
\end{equation}
cache misses in total.

\subsection{Fused Redistribution}
\label{sec:fused-redistribution}

\def\cw{0.62}
\colorlet{cB}{blue!11}
\colorlet{cO}{orange!20}

\newcommand{\bin}[1]{\ifcase#1 00\or01\or10\or11\fi}

\newcommand{\cl}[5]{%
  \fill[#5] (#1,#2) rectangle ++(\cw,-\cw);
  \draw[thin] (#1,#2) rectangle ++(\cw,-\cw);
  \node[font=\tiny\bfseries,inner sep=0,scale=0.95]
       at (#1+\cw/2,#2-\cw/2){%
       {\color{blue!55!black}\bin{#3}}{\color{black},}{\color{red!70!black}\bin{#4}}};}

\newcommand{\blk}[3]{%
  \foreach \px/\py/\rr/\cc/\ff in {#3}{%
    \cl{#1+\px*\cw}{#2-\py*\cw}{\rr}{\cc}{\ff}}}

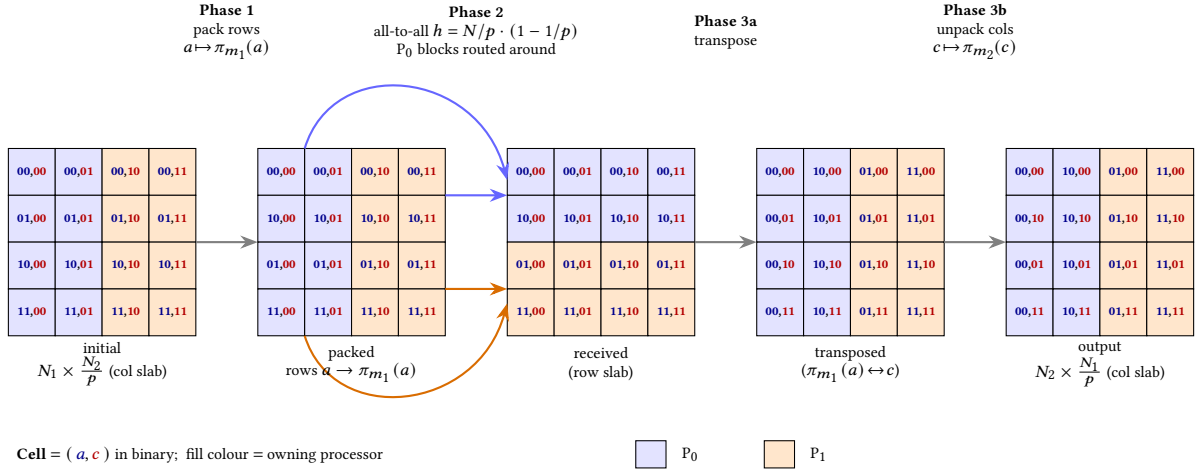
\begin{figure*}
    \centering
    \begin{tikzpicture}[font=\small]

\punt{
\node[font=\large\bfseries] at (7.84,2.75)
  {PARA-CO-DRP$_b$: tracking logical indices $(r,c)$ through one global exchange};

\node[font=\scriptsize] at (7.84,2.32)
  {$N_1\!=\!N_2\!=\!4,\ p\!=\!2,\ b\!=\!2$\quad(digit reversal on $2$ bits: $1\!\leftrightarrow\!2$,\ $0,3$ fixed)\quad
   \textbf{\color{blue!55!black}blue\,=\,P$_0$},\ \textbf{\color{orange!85!black}orange\,=\,P$_1$}};
}
\blk{0}{0}{0/0/0/0/cB,1/0/0/1/cB,2/0/0/2/cO,3/0/0/3/cO,
           0/1/1/0/cB,1/1/1/1/cB,2/1/1/2/cO,3/1/1/3/cO,
           0/2/2/0/cB,1/2/2/1/cB,2/2/2/2/cO,3/2/2/3/cO,
           0/3/3/0/cB,1/3/3/1/cB,2/3/3/2/cO,3/3/3/3/cO}

\blk{3.3}{0}{0/0/0/0/cB,1/0/0/1/cB,2/0/0/2/cO,3/0/0/3/cO,
             0/1/2/0/cB,1/1/2/1/cB,2/1/2/2/cO,3/1/2/3/cO,
             0/2/1/0/cB,1/2/1/1/cB,2/2/1/2/cO,3/2/1/3/cO,
             0/3/3/0/cB,1/3/3/1/cB,2/3/3/2/cO,3/3/3/3/cO}

\blk{6.6}{0}{0/0/0/0/cB,1/0/0/1/cB,2/0/0/2/cB,3/0/0/3/cB,
             0/1/2/0/cB,1/1/2/1/cB,2/1/2/2/cB,3/1/2/3/cB,
             0/2/1/0/cO,1/2/1/1/cO,2/2/1/2/cO,3/2/1/3/cO,
             0/3/3/0/cO,1/3/3/1/cO,2/3/3/2/cO,3/3/3/3/cO}

\blk{9.9}{0}{0/0/0/0/cB,1/0/2/0/cB,2/0/1/0/cO,3/0/3/0/cO,
             0/1/0/1/cB,1/1/2/1/cB,2/1/1/1/cO,3/1/3/1/cO,
             0/2/0/2/cB,1/2/2/2/cB,2/2/1/2/cO,3/2/3/2/cO,
             0/3/0/3/cB,1/3/2/3/cB,2/3/1/3/cO,3/3/3/3/cO}

\blk{13.2}{0}{0/0/0/0/cB,1/0/2/0/cB,2/0/1/0/cO,3/0/3/0/cO,
              0/1/0/2/cB,1/1/2/2/cB,2/1/1/2/cO,3/1/3/2/cO,
              0/2/0/1/cB,1/2/2/1/cB,2/2/1/1/cO,3/2/3/1/cO,
              0/3/0/3/cB,1/3/2/3/cB,2/3/1/3/cO,3/3/3/3/cO}

\draw[-{Stealth},gray,thick] (2.48,-1.24)--(3.3,-1.24);

\draw[-{Stealth},blue!60,thick]         (5.78,-0.62)--(6.6,-0.62); 
\draw[-{Stealth},orange!85!black,thick] (5.78,-1.86)--(6.6,-1.86); 
\draw[-{Stealth},blue!60,thick]                                    
  (3.92,0.0) .. controls (4.3,1.25) and (6.2,1.05) .. (6.6,-0.42);
\draw[-{Stealth},orange!85!black,thick]                            
  (3.92,-2.48) .. controls (4.3,-3.7) and (6.2,-3.5) .. (6.6,-2.05);

\draw[-{Stealth},gray,thick] (9.08,-1.24)--(9.9,-1.24);
\draw[-{Stealth},gray,thick] (12.38,-1.24)--(13.2,-1.24);

\foreach \x/\t in {1.24/{initial\\ $N_1\times\frac{N_2}{p}$ (col slab)},
                   4.54/{packed\\ rows $a \to\pi_{m_1}(a)$},
                   7.84/{received\\ (row slab)},
                   11.14/{transposed\\ ($\pi_{m_1}(a)\!\leftrightarrow\!c$)},
                   14.44/{output\\ $N_2\times\frac{N_1}{p}$ (col slab)}}
  \node[font=\scriptsize,align=center] at (\x,-2.85){\t};

\node[font=\scriptsize,align=center] at (2.89,1.55)
  {\textbf{Phase 1}\\ pack rows\\ $a\!\mapsto\!\pi_{m_1}\!(a)$};
\node[font=\scriptsize,align=center] at (6.19,1.55)
  {\textbf{Phase 2}\\ all-to-all $h=N/p\cdot(1-1/p)$\\ P$_0$ blocks routed around};
\node[font=\scriptsize,align=center] at (9.49,1.55)
  {\textbf{Phase 3a}\\ transpose};
\node[font=\scriptsize,align=center] at (12.79,1.55)
  {\textbf{Phase 3b}\\ unpack cols\\ $c\!\mapsto\!\pi_{m_2}\!(c)$};

\node[font=\scriptsize,anchor=west] at (0,-4.05)
  {\textbf{Cell} $=(\,${\color{blue!55!black}\textbf{$a$}}$,\,${\color{red!70!black}\textbf{$c$}}$\,)$ in binary;\ \ fill colour $=$ owning processor};
\fill[cB](8.3,-3.87) rectangle ++(0.4,-0.36); \draw[thin](8.3,-3.87) rectangle ++(0.4,-0.36);
\node[font=\scriptsize,anchor=west] at (8.8,-4.05){P$_0$};
\fill[cO](10.0,-3.87) rectangle ++(0.4,-0.36); \draw[thin](10.0,-3.87) rectangle ++(0.4,-0.36);
\node[font=\scriptsize,anchor=west] at (10.5,-4.05){P$_1$};

\end{tikzpicture}
    \caption{\(\proc{PARA-CO-DRP}_b\): tracking logical indices \((a,c)\), say \((01,00)\) in blue, through one global exchange\\ \(N_1\!=\!N_2\!=\!4,\ p\!=\!2,\ b\!=\!2^1\!=\!2\) (digit reversal on \(2\) bits: \(1\!\leftrightarrow\!2\),\ \(0,3\) fixed) 
   \textbf{\color{blue!55!black}blue\,=\,P\(_0\)},\ \textbf{\color{orange!85!black}orange\,=\,P\(_1\)}}
    \label{fig:para-co-drp-data-transform-pic3}
\end{figure*}

The middle stage materializes the full digit reversal
\[
    \pi_m(a\pplus c)
    =
    \pi_{m_2}(c)
    \pplus
    \pi_{m_1}(a)
\]
identified in \secref{full-middle-permutation}. It consists of local
packing, one balanced BSP exchange, and local unpacking.

\paragrf{Local pack.}
Processor \(q\) owns the Phase-I slab
\[
    \widetilde A[a,c],
    \qquad
    c
    =
    q\frac{N_2}{p}
    +
    c_{\mathrm{loc}}.
\]
Because
\[
    \widetilde A[a,c]
    =
    A[\pi_{m_1}(a),c],
\]
a row macro-element permutation by \(\pi_{m_1}\) places the logical
first-stage value \(A[v,c]\) in packed row
\[
    v
    =
    \pi_{m_1}(a).
\]
During the same scan, PACO forms
\[
    Z[v,c]
    =
    A[v,c]\omega_N^{-vc}.
\]
\punt{
Processor \(q\) initially owns the Phase-I physical slab
\[
    \widetilde A[a,c],
    \qquad
    a\in[N_1],
    \qquad
    c
    =
    q\frac{N_2}{p}
    +
    c_{\mathrm{loc}},
    \qquad
    c_{\mathrm{loc}}
    \in
    \left[
        \frac{N_2}{p}
    \right].
\]
Because Phase~I exposes the output layout \(\pi_{m_1}\),
\[
    \widetilde A[a,c]
    =
    A[\pi_{m_1}(a),c].
\]
PACO applies \(\pi_{m_1}\) as a digit reversal of \(N_1\) contiguous
row macro-elements, each of size \(N_2/p\). If
\[
    v
    =
    \pi_{m_1}(a),
\]
then the item formerly at physical row \(a\) is packed into row \(v\),
and its content becomes
\[
    \widetilde A[a,c]
    =
    A[v,c].
\]
Thus, local packing removes the Phase-I output layout. While scanning
this row macro-element, PACO simultaneously applies the top-level
twiddle factor and forms
\[
    Z[v,c]
    =
    A[v,c]\omega_N^{-vc}.
\]
}

\paragrf{One BSP exchange.}
Decompose the packed logical row coordinate as
\[
    v
    =
    q'\frac{N_1}{p}
    +
    v_{\mathrm{loc}},
    \qquad
    q'\in[p],
    \qquad
    v_{\mathrm{loc}}
    \in
    \left[
        \frac{N_1}{p}
    \right].
\]
The processor \(q'\) is the destination owner of this item. After row
packing, all rows with a common \(q'\) form one contiguous
source--destination block. For every source--destination pair
\((q,q')\), this block has shape
\[
    \frac{N_1}{p}
    \times
    \frac{N_2}{p}
\]
and contains exactly
\(
    (N_1/p)
    \cdot
    (N_2/p)
    =
    N/p^2
\)
items. PACO sends these blocks in one balanced \textsc{All-to-All} BSP exchange.

\paragrf{Local unpack.}
At destination \(q'\), the block received from source \(q\) is indexed
by
\[
    (v_{\mathrm{loc}},c_{\mathrm{loc}}).
\]
PACO transposes this local rectangle, producing the coordinate order
\[
    (c_{\mathrm{loc}},v_{\mathrm{loc}}).
\]
The source identifier \(q\) and the first coordinate reconstruct the
complete former column coordinate
\[
    c
    =
    q\frac{N_2}{p}
    +
    c_{\mathrm{loc}}.
\]
Finally, PACO applies \(\pi_{m_2}\) as a digit reversal of the
\(N_2\) row macro-elements, each of size \(N_1/p\). Writing
\[
    u
    =
    \pi_{m_2}(c),
\]
the item is placed at the final local coordinate
\[
    (u,v)
    =
    \left(
        \pi_{m_2}(c),
        \pi_{m_1}(a)
    \right).
\]
Equivalently, the resulting local array satisfies
\begin{equation}
    \widehat Z[u,v]
    =
    Z[v,\pi_{m_2}(u)],
    \label{eq:unpacked-phase-three-layout}
\end{equation}
which is precisely the input layout consumed by the Phase-III local FFT.

\paragrf{Macro-element digit reversal.}
The row and column digit reversals in the fused stage are not
elementwise random-access permutations. During packing, each of the
\(N_1\) rows is treated as one contiguous macro-element of length
\(N_2/p\); during unpacking, each of the \(N_2\) rows is treated as one
contiguous macro-element of length \(N_1/p\). A digit-reversal
two-cycle is implemented by sequentially swapping its two
macro-elements, using one temporary item. 

\paragrf{Coordinate trace.}
For an item initially stored at source processor \(q\) with physical
Phase-I coordinate
\[
    (a,c),
    \qquad
    c
    =
    q\frac{N_2}{p}
    +
    c_{\mathrm{loc}},
\]
the fused middle stage performs the following coordinate trace:
\begin{align}
    (a,c)
    &\xmapsto{\mathrm{pack}}
    (v,c),
    &
    v
    &=
    \pi_{m_1}(a),
    \notag\\
    (v,c)
    &\xmapsto{\mathrm{route}}
    \bigl(
        q';
        v_{\mathrm{loc}},c_{\mathrm{loc}}
    \bigr),
    &
    v
    &=
    q'\frac{N_1}{p}
    +
    v_{\mathrm{loc}},
    \notag\\
    \bigl(
        q;
        v_{\mathrm{loc}},c_{\mathrm{loc}}
    \bigr)
    &\xmapsto{\mathrm{transpose}}
    \bigl(
        q;
        c_{\mathrm{loc}},v_{\mathrm{loc}}
    \bigr), & \text{at}& \text{ \(q'\) (source = \(q\))}
    \notag\\
    \bigl(
        q;
        c_{\mathrm{loc}},v_{\mathrm{loc}}
    \bigr)
    &\xmapsto{\mathrm{unpack}}
    (u,v),
    &
    u
    &=
    \pi_{m_2}(c).
    \label{eq:para-co-drp-coord-trace}
\end{align}
Hence,
\begin{align}
    (a,c)
    \longmapsto
    \left(
        \pi_{m_2}(c),
        \pi_{m_1}(a)
    \right)
    =
    \pi_m(a\pplus c).
    \label{eq:overall-coord-trace}
\end{align}

\begin{theorem}[Correctness and complexity of the fused global
permutation]
\label{thm:fused-global-permutation}
Under assumptions~\eqref{exact-divisibility} and
\eqref{slab-width-cache-block}, the composition of local packing,
one BSP redistribution, and local unpacking materializes the full
permutation
\begin{equation}
    \pi_m(a\pplus c)
    =
    \pi_{m_2}(c)
    \pplus
    \pi_{m_1}(a).
    \label{eq:fused-global-full-permutation}
\end{equation}
Applied to the Phase-I physical output, it produces
\begin{equation}
    \widehat Z[u,v]
    =
    Z[v,\pi_{m_2}(u)],
    \label{eq:fused-global-phase-three-layout}
\end{equation}
which is precisely the input layout required by Phase~III.

Each processor performs
\begin{equation}
    W_{q,\mathrm{twiddle}}^{\mathrm{arith}}
    =
    \Theta\!\left(
        \frac{N}{p}
    \right)
    \label{eq:fused-global-twiddle-work}
\end{equation}
twiddle arithmetic and
\begin{equation}
    W_{q,\mathrm{perm}}^{\mathrm{rearr}}
    =
    O\!\left(
        \frac{N}{p}
    \right)
    \label{eq:fused-global-rearrangement-work}
\end{equation}
local rearrangement work. Its local cache complexity is
\begin{equation}
    Q_{q,\mathrm{perm}}
    =
    O\!\left(
        \frac{N}{pB}
    \right).
    \label{eq:fused-global-local-cache}
\end{equation}

The permutation uses one global redistribution round with exactly
\begin{equation}
    \frac{N}{p^2}
    \label{eq:fused-global-pairwise-volume}
\end{equation}
items for every source--destination processor pair.  Excluding the
self block, every processor sends and receives
\begin{equation}
    h
    =
    \frac{N}{p}
    \left(
        1-\frac{1}{p}
    \right)
    \label{eq:fused-global-h-volume}
\end{equation}
remote items, and the aggregate remote migration volume is
\begin{equation}
    \mu
    =
    N
    \left(
        1-\frac{1}{p}
    \right).
    \label{eq:fused-global-total-volume}
\end{equation}

Excluding the mandatory local input and output arrays, the auxiliary
workspace is
\begin{equation}
    O\!\left(
        \frac{N}{p^2}
        +
        \log N
    \right)
    \label{eq:fused-global-auxiliary-space}
\end{equation}
items or words per processor.
\end{theorem}

\begin{proof}
The coordinate trace \eqref{para-co-drp-coord-trace} maps every Phase-I physical coordinate
\((a,c)\) to
\[
    (u,v)
    =
    \left(
        \pi_{m_2}(c),
        \pi_{m_1}(a)
    \right).
\]
It therefore materializes
\[
    \pi_m(a\pplus c)
    =
    \pi_{m_2}(c)
    \pplus
    \pi_{m_1}(a).
\]
Since the Phase-I physical layout satisfies
\[
    \widetilde A[a,c]
    =
    A[\pi_{m_1}(a),c],
\]
the packed row \(v=\pi_{m_1}(a)\) contains \(A[v,c]\). Fusing the
top-level twiddle multiplication during this scan forms
\[
    Z[v,c]
    =
    A[v,c]\omega_N^{-vc}.
\]
The final coordinate mapping consequently gives
\[
    \widehat Z[u,v]
    =
    Z[v,\pi_{m_2}(u)],
\]
which is the Phase-III input layout.

Packing applies macro-element digit reversal to \(N_1\) rows of
length \(N_2/p\), and unpacking applies it to \(N_2\) rows of length
\(N_1/p\). By the slab-width assumptions \eqref{slab-width-cache-block}, both macro-element widths
are at least \(B\). Thus, together they require
\[
    O\!\left(
        \frac{N}{p}
    \right)
\]
local rearrangement work and
\[
    O\!\left(
        \frac{N}{pB}
    \right)
\]
cache misses; the fused twiddle scan contributes
\(\Theta(N/p)\) arithmetic work without an additional cache pass.

Each processor transposes \(p\) blocks of size \(N/p^2\). The
cache-oblivious transpose bound gives
\[
    O\!\left(
        p+\frac{N}{pB}
    \right)
\]
misses in total. Since
\[
    \frac{N}{p^2}
    =
    \frac{N_1}{p}\frac{N_2}{p}
    \ge
    B^2,
\]
we have \(p=O(N/(pB))\), so this cost is also
\(O(N/(pB))\).

For every source--destination pair, the block has
\[
    \frac{N_1}{p}\frac{N_2}{p}
    =
    \frac{N}{p^2}
\]
items. Hence, excluding the self block, each processor sends and
receives
\[
    \frac{N}{p}
    \left(
        1-\frac{1}{p}
    \right)
\]
remote items, and the aggregate remote volume is
\[
    N
    \left(
        1-\frac{1}{p}
    \right).
\]
The raw incoming regions are part of the mandatory output array; the
only non-output array workspace is one reusable \(N/p^2\)-item
transpose buffer, plus \(O(\log N)\) recursion-stack words.
\end{proof}

\punt{
\begin{proof}
Consider an item initially stored at physical Phase-I coordinate
\((a,c)\), where
\[
    c
    =
    q\frac{N_2}{p}
    +
    c_{\mathrm{loc}}.
\]
Local row packing maps
\[
    a
    \longmapsto
    v
    =
    \pi_{m_1}(a)
\]
and applies the top-level twiddle factor
\(\omega_N^{-vc}\).  Decompose the packed row coordinate as
\[
    v
    =
    q'\frac{N_1}{p}
    +
    v_{\mathrm{loc}}.
\]
The item belongs to the source--destination block from processor
\(q\) to processor \(q'\), where it has block-local coordinate
\[
    (v_{\mathrm{loc}},c_{\mathrm{loc}}).
\]

The BSP redistribution sends this block directly to processor \(q'\).
The local rectangular transpose changes its coordinate to
\[
    (c_{\mathrm{loc}},v_{\mathrm{loc}}).
\]
Together with the source identifier \(q\), the first coordinate
reconstructs the original global column
\[
    c
    =
    q\frac{N_2}{p}
    +
    c_{\mathrm{loc}}.
\]
Finally, the local row permutation maps this row to
\[
    u
    =
    \pi_{m_2}(c).
\]
The output column is
\[
    v
    =
    q'\frac{N_1}{p}
    +
    v_{\mathrm{loc}}
    =
    \pi_{m_1}(a).
\]
Hence, the final global coordinate is
\[
    (u,v)
    =
    \left(
        \pi_{m_2}(c),
        \pi_{m_1}(a)
    \right),
\]
which proves~\eqref{fused-global-full-permutation}.

The Phase-I physical layout satisfies
\[
    \widetilde A[a,c]
    =
    A[\pi_{m_1}(a),c].
\]
Thus, after the row permutation, physical packed row \(v\) contains
the logical first-stage value \(A[v,c]\).  The fused multiplication
therefore forms
\[
    Z[v,c]
    =
    A[v,c]\omega_N^{-vc}.
\]
Applying the coordinate mapping above gives
\[
    \widehat Z[u,v]
    =
    Z[v,\pi_{m_2}(u)],
\]
which proves~\eqref{fused-global-phase-three-layout}.

For local packing, apply the macro-element digit-reversal bound to
\(N_1\) contiguous rows of length \(N_2/p\).  Since
\[
    \frac{N_2}{p}
    \ge
    B,
\]
packing performs
\[
    O\!\left(
        \frac{N}{p}
    \right)
\]
rearrangement work and
\[
    O\!\left(
        \frac{N}{pB}
    \right)
\]
cache misses.  The destination-indexed twiddle multiplication scans
the same macro-elements and contributes
\(\Theta(N/p)\) arithmetic operations without requiring an
additional cache pass.

Each processor transposes \(p\) rectangular blocks, each containing
\(N/p^2\) items.  A cache-oblivious out-of-place transpose performs
linear work and incurs
\[
    O\!\left(
        1+
        \frac{N}{p^2B}
    \right)
\]
cache misses per block~\cite{FrigoLePr99}.  Thus, all local block
transpositions perform \(O(N/p)\) work and incur
\[
    O\!\left(
        p+
        \frac{N}{pB}
    \right)
\]
cache misses.  By the slab-width assumptions,
\[
    \frac{N}{p^2}
    =
    \frac{N_1}{p}
    \frac{N_2}{p}
    \ge
    B^2.
\]
Consequently,
\[
    p
    \le
    \frac{N}{pB^2}
    \le
    \frac{N}{pB},
\]
so the total transpose cost is
\[
    O\!\left(
        \frac{N}{pB}
    \right)
\]
cache misses.  The copy from the reusable transpose scratch array
back to the designated output region is sequential and is absorbed
by the same bound.

For final unpacking, apply the macro-element digit-reversal bound to
\(N_2\) contiguous rows of length \(N_1/p\).  The condition
\(N_1/p\ge B\) gives another
\[
    O\!\left(
        \frac{N}{p}
    \right)
\]
rearrangement work and
\[
    O\!\left(
        \frac{N}{pB}
    \right)
\]
cache misses.  Combining packing, transposition, and unpacking
proves~\eqref{fused-global-rearrangement-work} and
\eqref{fused-global-local-cache}.

For every fixed source processor \(q\) and destination processor
\(q'\), exactly \(N_1/p\) packed rows belong to \(q'\), and processor
\(q\) owns \(N_2/p\) source columns.  Hence, every source--destination
block contains
\[
    \frac{N_1}{p}
    \cdot
    \frac{N_2}{p}
    =
    \frac{N}{p^2}
\]
items.  Excluding the self block yields
\[
    (p-1)\frac{N}{p^2}
    =
    \frac{N}{p}
    \left(
        1-\frac{1}{p}
    \right)
\]
remote items sent and received by each processor.  Summing the remote
send volumes over all processors proves
\eqref{fused-global-total-volume}.

The \(p\) raw incoming block regions together occupy \(N/p\) items
in the mandatory local output array and are therefore not auxiliary
workspace.  The algorithm uses one reusable transpose scratch array
of \(N/p^2\) items.  Macro-element swaps require \(O(1)\) temporary
items, and the digit-reversal and transpose recursions use
\(O(\log N)\) stack words.  This proves
\eqref{fused-global-auxiliary-space}.
\end{proof}
}

\punt{
\paragrf{Local pack.}
Processor \(q\) initially owns an
\[
    N_1
    \times
    \frac{N_2}{p}
\]
Phase-I slab. Its physical row coordinate \(a\) represents logical
frequency
\[
    k_1
    =
    \pi_{m_1}(a).
\]
PACO first applies \(\pi_{m_1}\) as a digit reversal of \(N_1\)
contiguous row macro-elements, each of size \(N_2/p\). This corrects
the Phase-I output layout. While scanning the macro-elements, it
simultaneously multiplies every item by its top-level twiddle factor
\[
    \omega_N^{-k_1j_2}.
\]
The packed rows are then grouped into contiguous blocks by their
target processor.

\paragrf{One BSP exchange.}
Writing the logical row coordinate as
\[
    k_1
    =
    q'\frac{N_1}{p}
    +
    k_{1,\mathrm{loc}},
\]
identifies \(q'\) as the target processor. Since \(\pi_{m_1}\) is a
bijection, every source owns exactly
\[
    \frac{N_1}{p}
    \cdot
    \frac{N_2}{p}
    =
    \frac{N}{p^2}
\]
items for every destination. The packed blocks are therefore exchanged
in one balanced BSP \(h\)-relation.

\paragrf{Local unpack.}
At destination processor \(q'\), each received block is a rectangle of
shape
\[
    \frac{N_1}{p}
    \times
    \frac{N_2}{p}.
\]
PACO transposes these rectangles locally, reconstructing the complete
\(N_2\) coordinate from the source identifier and local column
coordinate. It then applies \(\pi_{m_2}\) as a digit reversal of
\(N_2\) contiguous row macro-elements. The resulting local view
satisfies
\[
    \widehat Z[u,k_1]
    =
    Z[k_1,\pi_{m_2}(u)],
\]
which is exactly the Phase-III input layout required by
Lemma~\ref{lem:lazy-layout-fft}.
}

\begin{theorem}[Fused redistribution]
\label{thm:fused-global-permutation}
The fused middle stage materializes the complete permutation
\[
    \pi_m(a\pplus c)
    =
    \pi_{m_2}(c)
    \pplus
    \pi_{m_1}(a),
\]
applies the top-level Cooley--Tukey twiddle factors, and produces the
input layout consumed by Phase~III. Per processor, it performs
\begin{equation}
    W_{\mathrm{perm,local}}
    =
    \Theta\!\left(
        \frac{N}{p}
    \right)
    \label{eq:permutation-local-work}
\end{equation}
local work and incurs
\begin{equation}
    Q_{\mathrm{perm,local}}
    =
    \Theta\!\left(
        \frac{N}{pB}
    \right)
    \label{eq:permutation-local-cache}
\end{equation}
cache misses.

The unique communication round is a balanced BSP \(h\)-relation with
\begin{equation}
    h
    =
    \frac{N}{p}
    \left(
        1-\frac{1}{p}
    \right)
    \label{eq:permutation-h-volume}
\end{equation}
remote items sent and received by each processor. Its aggregate remote
communication volume is
\begin{equation}
    \mu
    =
    N
    \left(
        1-\frac{1}{p}
    \right).
    \label{eq:permutation-total-volume}
\end{equation}
\end{theorem}

\begin{proof}[Proof sketch]
Consider an item initially stored at source coordinate \((a,c)\).
The packing permutation sends it to logical first-stage frequency
\[
    v
    =
    \pi_{m_1}(a),
\]
thereby cancelling the output layout produced by Phase~I; the same
scan applies the twiddle \(\omega_N^{-vc}\). Writing
\[
    v
    =
    q'\frac{N_1}{p}
    +
    v_{\mathrm{loc}}
\]
routes the item to its unique Phase-III owner \(q'\). At that
destination, transposing the received block exchanges
\((v_{\mathrm{loc}},c_{\mathrm{loc}})\) into
\((c_{\mathrm{loc}},v_{\mathrm{loc}})\), while the source identifier
reconstructs
\[
    c
    =
    q\frac{N_2}{p}
    +
    c_{\mathrm{loc}}.
\]
The final row digit reversal places the item at
\[
    u
    =
    \pi_{m_2}(c).
\]
Thus, the complete coordinate mapping is
\[
    (a,c)
    \longmapsto
    \left(
        \pi_{m_2}(c),
        \pi_{m_1}(a)
    \right),
\]
which is the full digit reversal
\(\pi_m(a\pplus c)\).

Both macro-element digit reversals and the local rectangular
transpositions \cite{FrigoLePr99} are cache-oblivious linear-work, linear-cache-miss
rearrangements: 
The row and column digit reversals are not element-wise random-access permutations. They are conducted on contiguous macro-elements of length \(N_2/p\) and \(N_1/p\), respectively. Given the slab-width conditions \eqref{slab-width-cache-block}, these permutations take linear local work and
\(O(N/(pB))\) cache misses. 
\punt{
The only nonconstant auxiliary array in
the middle stage is the reusable \(N/p^2\)-item scratch space used by
the local block transpose.
}

Since \(\pi_{m_1}\) is a bijection, each source has
exactly \(N_1/p\) packed rows for every destination \(q'\), and hence
exactly \(N/p^2\) items in every source--destination block. This gives
the stated balanced communication volume and local complexity bounds.
The complete in-place digit-reversal construction, cache analysis,
and workspace details are deferred to the appendix.
\end{proof}

\punt{
\begin{proof}[Proof sketch]
The initial macro-element digit reversal removes the
\(\pi_{m_1}\) layout exposed by Phase~I, and the final macro-element
digit reversal establishes the \(\pi_{m_2}\) layout required by
Phase~III. The intervening rectangular transpositions exchange the
two top-level digit fields and reconstruct the factor-swapped local
slabs. Thus, their composition realizes the full permutation
\(\pi_m\).

Both local digit reversals and the rectangular transpositions are
cache-oblivious linear-work, linear-cache-miss rearrangements. The
routing balance follows because \(\pi_{m_1}\) is a bijection: every
source--destination pair contains exactly \(N/p^2\) items. The
complete in-place digit-reversal construction, routing trace, and
workspace bound are given in the appendix.
\end{proof}
}
\secput{complexity}{Analysis and Optimality}

This section combines the two local FFT phases and the fused middle
redistribution, and then states the matching lower bounds under their
respective scopes. Correctness was established in
Theorem~\ref{thm:paco-correctness}; here we analyze only local work,
cache traffic, and interprocessor communication.

\subsection{Upper Bounds}
\label{sec:paco-upper-bounds}

Table~\ref{tab:paco-cost-summary} summarizes the per-processor costs.
In Phase~I, every processor evaluates \(N_2/p\) transforms of length
\(N_1\). In Phase~III, it evaluates \(N_1/p\) transforms of length
\(N_2\). The middle stage performs one fused redistribution: its
top-level twiddles are included in local packing, and its local
rearrangements consist of two macro-element digit reversals and
rectangular block transpositions.

\begin{table}[t]
\caption{Per-processor costs of the three PACO components.}
\label{tab:paco-cost-summary}
\centering
\small
\renewcommand{\arraystretch}{1.25}
\setlength{\tabcolsep}{3pt}
\begin{tabular}{lccc}
\toprule
Component
&
Local work
&
Cache misses
&
Communication
\\
\midrule
Local Phase~I
&
\(\displaystyle
 O\!\left(
 \frac{N}{p}\log N_1
 \right)\)
&
\(\displaystyle
 O\!\left(
 \frac{N}{pB}
 (1+\log_M N_1)
 \right)\)
&
\(0\)
\\
Fused middle stage
&
\(\displaystyle
 O\!\left(
 \frac{N}{p}
 \right)\)
&
\(\displaystyle
 O\!\left(
 \frac{N}{pB}
 \right)\)
&
one \(h\)-relation
\\
Local Phase~III
&
\(\displaystyle
 O\!\left(
 \frac{N}{p}\log N_2
 \right)\)
&
\(\displaystyle
 O\!\left(
 \frac{N}{pB}
 (1+\log_M N_2)
 \right)\)
&
\(0\)
\\
\bottomrule
\end{tabular}
\end{table}

\begin{proposition}[PACO upper bounds]
\label{prop:paco-upper-bounds}
Under assumptions \eqref{exact-divisibility}--\eqref{slab-width-cache-block}, PACO has maximum per-processor
local work
\begin{equation}
    W_{\max}
    =
    O\!\left(
        \frac{N}{p}\log N
    \right)
    \label{eq:paco-critical-work-upper-bound}
\end{equation}
and maximum per-processor cache complexity
\begin{equation}
    Q_{\max}
    =
    O\!\left(
        \frac{N}{pB}
        \left(
            1+\log_M N
        \right)
    \right).
    \label{eq:paco-critical-cache-upper-bound}
\end{equation}
It performs exactly one global redistribution round,
\[
    R
    =
    1,
\]
with
\begin{equation}
    h
    =
    \frac{N}{p}
    \left(
        1-\frac{1}{p}
    \right)
    \label{eq:paco-upper-h}
\end{equation}
remote items sent and received by each processor, and aggregate remote
volume
\begin{equation}
    \mu
    =
    N
    \left(
        1-\frac{1}{p}
    \right).
    \label{eq:paco-upper-communication-volume}
\end{equation}
\end{proposition}

\begin{proof}
Proposition~\ref{prop:local-batched-fft} gives the local work and cache
bounds for Phases~I and III, while
Theorem~\ref{thm:fused-global-permutation} gives the linear local cost
of the middle stage. Since
\[
    \log N_1+\log N_2
    =
    \log N,
\]
their local work bounds sum to
\eqref{paco-critical-work-upper-bound}; the same identity yields
\eqref{paco-critical-cache-upper-bound} after absorbing the
constant number of linear scans.

For every source--destination processor pair, the fused permutation
routes exactly
\[
    \frac{N}{p^2}
\]
items. Excluding the self block, every processor therefore sends and
receives
\[
    (p-1)\frac{N}{p^2}
    =
    \frac{N}{p}
    \left(
        1-\frac{1}{p}
    \right)
\]
remote items. Summing these balanced send volumes gives
\eqref{paco-upper-communication-volume}.
\end{proof}

The corresponding aggregate work and cache bounds are
\[
    O(N\log N)
    \qquad\text{and}\qquad
    O\!\left(
        \frac{N}{B}
        \left(
            1+\log_M N
        \right)
    \right),
\]
respectively. The local FFT recursion uses \(O(\log N)\) stack words,
and the fused middle stage uses one reusable
\(N/p^2\)-item transpose scratch array; these workspace details are
established in Section~\ref{sec:fused-redistribution} and the
appendix.

\subsection{Lower Bounds and Exact Scope}
\label{sec:paco-lower-bounds}
\begin{table*}[!t]
\centering
\caption{Representative FFT frameworks. ``Communication'' counts global
exchanges; additional interprocessor communication is noted where relevant.}
\label{tab:related-fft}
\small
\setlength{\tabcolsep}{3pt}
\renewcommand{\arraystretch}{0.92}
\vspace{-0.5em}
\begin{tabular}{@{}llll@{}}
\toprule
Method & Numerical & Communication & Output ownership / cache guarantee \\
\midrule
Four-/six-step FFT
& Exact
& One or more transpose-like exchanges
& Factor-swapped view possible; no cache-oblivious guarantee \\

Frigo et al.\ CO-FFT
& Exact
& N/A (Sequential Alg.)
& Natural-order output; optimal ideal-cache complexity \\

Edelman et al.
& Tunable accuracy
& One transpose + hierarchical communication
& Typically block ownership; no local ideal-cache theorem \\

Tang et al.
& Tunable accuracy
& One all-to-all + neighbor exchange
& Typically block ownership; no local ideal-cache theorem \\

PACO
& Exact
& One BSP redistribution only
& Canonical output under factor-swapped slabs; cache-oblivious \\
\bottomrule
\end{tabular}
\vspace{-0.8em}
\end{table*}

\punt{
\begin{table*}[!h]
\caption{
Representative FFT frameworks.  The communication column counts
algorithmic global exchanges and notes additional interprocessor
communication when relevant.
}
\label{tab:related-summary}
\centering
\renewcommand{\arraystretch}{1.12}
\setlength{\tabcolsep}{3pt}
\begin{tabularx}{\textwidth}{
    @{}
    p{0.16\textwidth}
    c
    p{0.25\textwidth}
    X
    @{}
}
\hline
Method
&
Numerical
&
Communication
&
Output ownership and local-cache guarantee
\\
\hline

Four-/six-step FFT
&
Exact
&
One or more transpose-like exchanges; further exchanges depend on
layout and distribution
&
A factor-swapped matrix view can be used; the factorization alone gives
no cache-oblivious local-cache theorem
\\

Frigo et al.\ CO-FFT
&
Exact
&
Sequential setting
&
Natural-order sequential output; asymptotically optimal
ideal-cache complexity
\\

Edelman et al.
&
Tunable accuracy
&
One distributed transpose, plus hierarchical communication in the FMM
&
Typically restores contiguous-block ownership; no processor-local
ideal-cache theorem
\\

Tang et al.
&
Tunable accuracy
&
One all-to-all on an oversampled array, plus neighbor ghost exchange
&
Typically restores contiguous-block ownership; no processor-local
ideal-cache theorem
\\

PACO
&
Exact
&
One BSP redistribution; no communication outside that superstep
&
Canonical coefficients under factor-swapped slab ownership; fully
cache-oblivious processor-local bound
\\
\hline
\end{tabularx}
\end{table*}
}

\punt{
\begin{table*}[t]
\caption{High-level comparison of representative FFT frameworks.}
\label{tab:related-summary}
\centering
\small
\renewcommand{\arraystretch}{1.15}
\setlength{\tabcolsep}{4pt}
\begin{tabular}{lccccc}
\toprule
Method
&
Exact?
&
Global redistribution
&
Other communication
&
Output ownership
&
Local cache theorem
\\
\midrule
Four-/six-step
&
Yes
&
One or more transpose-like exchanges
&
Layout dependent
&
Often factor-swapped
&
No general cache-oblivious result
\\
Cache-oblivious FFT
&
Yes
&
Sequential setting
&
---
&
Sequential output
&
Yes
\\
Approximate low-round FFTs
&
Typically no
&
One all-to-all/transpose
&
Often hierarchical or neighbor exchange
&
Usually restored
&
Not in this model
\\
PACO
&
Yes
&
One BSP redistribution
&
None outside it
&
Factor-swapped
&
Yes
\\
\bottomrule
\end{tabular}
\end{table*}
}

We now prove the lower bounds under the scopes stated after
\thmref{main-result} and the fixed-radix scope \footnote{Let $b=2^t$, $t=O(1)$. If every radix-$b$ butterfly is expanded into its $t$
constituent radix-2 layers, the resulting dependency graph is the standard
radix-2 FFT DAG. Hence the maximum per-processor work/cache lower bounds are proved based on radix-2 DAG.}(see Remark~\ref{rmk:fixed-radix-scope} in \appref{lb}).
\punt{
The following lower bounds have deliberately different scopes. The
one-round bound concerns the complete DFT under the ownership model of
\secref{model}. The migration bound concerns
\emph{materializing the prescribed fused permutation under the fixed
source and target ownerships}. The work and cache bounds concern
no-recomputation executions of the fixed-radix Cooley--Tukey FFT DAG.
}

\begin{proposition}[At least one redistribution is necessary]
\label{prop:one-redistribution-lower-bound}
Every correct parallel DFT algorithm satisfying the owner-computes and
no-replication assumptions of \secref{model} requires
\[
    R
    \ge
    1.
\]
\end{proposition}

\begin{proof}
Every DFT coefficient depends on every input item, since every entry
of \(\Omega_N\) is nonzero. No processor initially owns the complete
input. If \(R=0\), then a processor's local state and every output it
returns are independent of some remotely owned input item, which
contradicts correctness after changing only that input.
\end{proof}

\begin{proposition}[Exact migration bound for the fused permutation]
\label{prop:fused-permutation-migration-lower-bound}
Consider materializing
\[
    \pi_m(a\pplus c)
    =
    \pi_{m_2}(c)
    \pplus
    \pi_{m_1}(a)
\]
between the source ownership
\eqref{input-slab-owner} and target ownership
\eqref{output-slab-owner}. Any such algorithm must communicate at
least
\begin{equation}
    \mu
    \ge
    N
    \left(
        1-\frac{1}{p}
    \right)
    \label{eq:fused-permutation-total-lower-bound}
\end{equation}
remote items. Moreover,
\begin{equation}
    \max_{q\in[p]}
    \left\{
        h_q^{\mathrm{send}},
        h_q^{\mathrm{recv}}
    \right\}
    \ge
    \frac{N}{p}
    \left(
        1-\frac{1}{p}
    \right).
    \label{eq:fused-permutation-critical-lower-bound}
\end{equation}
\end{proposition}

\begin{proof}
For each fixed source column \(c\), exactly \(N_1/p\) row coordinates
\(a\) satisfy
\[
    \left\lfloor
        \frac{\pi_{m_1}(a)}{N_1/p}
    \right\rfloor
    =
    \left\lfloor
        \frac{c}{N_2/p}
    \right\rfloor,
\]
because \(\pi_{m_1}\) is a bijection. Hence exactly
\[
    N_2\frac{N_1}{p}
    =
    \frac{N}{p}
\]
items retain their owner, while all remaining
\[
    N-\frac{N}{p}
\]
items must cross a processor boundary at least once. Averaging the
resulting aggregate send and receive volumes over the \(p\) processors
gives \eqref{fused-permutation-critical-lower-bound} since the maximum over \(p\) processors is at least the average.
\end{proof}

\begin{proposition}[Maximum Per-Processor work lower bound]
\label{prop:max-per-processor-work-lower-bound}
For the fixed-base-\(b\) Cooley--Tukey DAG used by PACO, under the
no-recomputation assumption,
\[
    W_{\max}
    =
    \Omega\!\left(
        \frac{N}{p}\log N
    \right).
\]
\end{proposition}

\begin{proof}
The fixed-base FFT DAG contains \(\Theta(N\log N)\) arithmetic
vertices. Without recomputation, these vertices must be executed
somewhere, so at least one processor executes at least their average
number,
\[
    \Omega\!\left(
        \frac{N}{p}\log N
    \right).
\]
\end{proof}

\begin{proposition}[Maximum Per-Processor cache lower bound]
\label{prop:max-per-processor-cache-lower-bound}
Under the exact slab decomposition and no arithmetic recomputation,
the two local FFT phases satisfy
\begin{equation}
    Q_{\max}
    =
    \Omega\!\left(
        \frac{N}{pB}
        \left(
            1+\log_M N
        \right)
    \right).
    \label{eq:paco-critical-cache-lower-bound}
\end{equation}
Consequently, the complete PACO execution satisfies the same lower
bound.
\end{proposition}

\begin{proof}[Proof roadmap]
Each processor evaluates a batch of \(N_2/p\) length-\(N_1\) FFT DAGs
in Phase~I and a batch of \(N_1/p\) length-\(N_2\) FFT DAGs in
Phase~III. Applying an FFT-DAG input-boundary lower bound \cite{RanjanSaZu11} to these
data-disjoint batches yields
\[
    \Omega\!\left(
        \frac{N}{pB}
        \left(
            1+\log_M N_1
        \right)
    \right)
\]
and
\[
    \Omega\!\left(
        \frac{N}{pB}
        \left(
            1+\log_M N_2
        \right)
    \right)
\]
misses, respectively, up to the standard initial-cache term. 
Fix a sufficiently large  constant \(\epsilon\).
If
\[
    \frac{N}{p}\ge \epsilon M,
\]
the \(\Theta(M/B)\) boundary term is absorbed. Otherwise,
\[
    \frac{N}{p}<\epsilon M.
\]
Because \(p\mid N_1\) and \(p\mid N_2\),
\[
    \frac{N_2}{p}\ge 1,
    \qquad
    \frac{N_1}{p}\ge 1.
\]
Hence
\[
    N_1
    \le
    N_1\frac{N_2}{p}
    =
    \frac{N}{p},
    \qquad
    N_2
    \le
    N_2\frac{N_1}{p}
    =
    \frac{N}{p}.
\]
Therefore,
\[
    N
    =
    N_1N_2
    \le
    \left(
        \frac{N}{p}
    \right)^2
    < \epsilon^2 M^2,
\]
and thus \(1+\log_M N=O(1)\). Since Phase~I begins with an
empty cache and accesses \(\Theta(N/p)\) distinct local input items,
it incurs
\[
    \Omega\!\left(
        \frac{N}{pB}
    \right)
\]
compulsory misses, which is the required bound in this case.
\punt{
otherwise
\(N=O(M^2)\), so \(1+\log_M N=O(1)\), and the compulsory scan of the
initial local slab supplies the required linear bound. 
}
The detailed
segment argument is deferred to Appendix.
\end{proof}

Together, Propositions~\ref{prop:paco-upper-bounds}--
\ref{prop:max-per-processor-cache-lower-bound} establish \thmref{main-result} under the scopes
stated following that theorem.

\punt{
These claims should be read with their stated scopes. In particular,
\eqref{fused-permutation-total-lower-bound} is optimal
\emph{for materializing the fused digit-reversal permutation under the
fixed source and target slab ownerships}; it is not a
communication-volume lower bound for all distributed FFT algorithms
with different intermediate ownership schedules or with communicated
linear combinations.
}

As a byproduct, \appref{lb} gives a finite block-aware lower bound for
arbitrary no-recomputation schedules of the radix-2 FFT DAG. In the
stated asymptotic regime, it yields
\[
    Q_{\mathrm{miss}}
    \ge
    \left(
        1-o(1)
    \right)
    \frac{N}{B}
    \log_M N.
\]
This stronger general statement is not needed for the PACO upper-bound
analysis and is therefore kept out of the main proof flow.

\subsection{Related Work}
\label{sec:related-work}

Classical four-step and six-step FFTs decompose a transform into local
FFT groups separated by twiddle factors and transpose-like changes of
view~\cite{Bailey90,VanLoan92}.  Exact distributed-memory FFTs have
further explored the tradeoffs among data allocation, input/output
ordering, and communication volume~\cite{JohnssonJaKr92,IndaBi01}.
These works show that a factor-swapped or otherwise reordered output
can save a transpose.  PACO adopts an explicit factor-swapped slab
contract, but additionally proves that its local computation is fully
cache oblivious and that its sole redistribution materializes the
layout deferred by the local recursions.

Cache-oblivious FFTs obtain ideal-cache locality through recursive
decomposition and recursive layout transformations~\cite{FrigoLePr99}.
PACO uses the same cache-oblivious objective in a distributed setting,
but does not materialize a transpose-like layout transformation at each
recursive split.  Instead, it represents these transformations lazily,
shows that they compose into base-\(b\) digit reversal, and fuses that
permutation with the ownership change between the two top-level FFT
groups.

Several lines of work reduce communication under different numerical,
distributional, or practical contracts.  Approximate distributed
transforms use low-rank structure or oversampling to reduce global
communication~\cite{EdelmanMcTo99,TangPaKi12}.  Exact cyclic-to-cyclic
multidimensional FFTs can also use a single all-to-all
exchange~\cite{KoopmanBi23}; unlike PACO, they retain cyclic ownership
rather than returning a factor-swapped slab distribution.  Other
multidimensional FFT frameworks optimize decomposition choices,
transpose order, communication overlap, or MPI datatype-based
redistribution~\cite{DuyOz14,Pekurovsky12,SongHo14,DalcinMoKe19,AyalaToHa20}.
They target practical multidimensional performance and do not establish
PACO's processor-local ideal-cache bound or its fused
digit-reversal-permutation guarantee.

Finally, Hong and Kung introduced the red-blue pebble framework for
FFT I/O lower bounds~\cite{HongKu81}, while Ranjan, Savage, and Zubair
established FFT-DAG boundary results used by our local-cache lower
bound~\cite{RanjanSaZu11}.  BSP and distributed-memory lower-bound work
also gives communication lower bounds for FFT computation under broader
schedule classes~\cite{BilardiScSi18,ScquizzatoSi14}.  
PACO's claims are more specific: one redistribution is necessary under its owner-computes and no-replication assumptions (\propositionref{one-redistribution-lower-bound}), and its migration bound is scoped as in the Scope of optimality paragraph of \secref{model}.
\punt{
PACO's claims are
more specific: one redistribution is necessary under its owner-computes
and no-replication assumptions, whereas its exact migration bound
applies only to the prescribed fused digit-reversal permutation between
the fixed source and target slab ownerships.
}

\subsection{Discussion: Ownership, Permutations, and Composition}
\label{sec:discussion}

PACO's factor-swapped output ownership is an explicit interface
contract, rather than an omitted cost.  The sole redistribution changes
from input-column slabs indexed by the \(N_2\) field to output-column
slabs indexed by the \(N_1\) field, exactly as required to make every
second-phase length-\(N_2\) transform local. 
\punt{
Thus, the one-round result
does not assert that an application requiring the original input slab
ownership can avoid a further exchange.  Rather, PACO isolates the
minimum communication needed to compute the DFT under the stated
factor-swapped target distribution.
}

This contract can nevertheless be useful when subsequent computation can
consume the factor-swapped frequency view directly.  More importantly
for the present result, it separates two issues that are often conflated:
the global ownership change required by the Cooley--Tukey dependency
structure, and the local layout changes introduced by cache-oblivious
recursion.  PACO shows that, for the stated slab decomposition, the
latter need not cause an additional redistribution: their complete
base-\(b\) digit reversal can be fused into the former.  The resulting
middle stage remains a balanced all-to-all, with every source--destination
pair carrying exactly \(N/p^2\) items.

\punt{
\secput{related}{Related Work, Limitations, and Conclusion}

\subsection{Related Work}
\label{sec:related-work}

Four-step and six-step FFT organizations factor a one-dimensional
transform into local transform groups separated by twiddle factors and
transpose-like changes of view~\cite{Bailey90,VanLoan92}. BSP-oriented parallel FFT algorithms have also used cyclic or group-cyclic distributions to
balance work and communication across processors~[9].  These
organizations demonstrate that communication behavior depends
essentially on the chosen ownership schedule; however, they do not
provide the fully cache-oblivious local-layout guarantee established
here. A compatible
four-step implementation can therefore use one global transpose under
a factor-swapped slab distribution. This algebraic property alone does
not provide a cache-oblivious local-memory guarantee: obtaining
contiguous local subproblems may require blocking, strided kernels, or
additional layout transformations. Cache-oblivious FFTs, beginning
with the framework of Frigo et al.~\cite{FrigoLePr99}, instead obtain
locality through recursive layout transformations. PACO combines these
two perspectives by retaining such transformations lazily and
materializing their composed digit permutation at the single required
ownership-change boundary.

Other low-communication distributed DFT methods also reduce the
number of global exchanges, but use different numerical or ownership
contracts. Edelman, McCorquodale, and Toledo~\cite{EdelmanMcTo99}
exploit low-rank structure and fast multipole ideas to obtain
tunable-accuracy communication-reducing transforms, with additional
hierarchical communication. Tang et al.~\cite{TangPaKi12} develop
single-all-to-all FFT frameworks based on convolution and
oversampling; their practical constructions also require
neighbor-level communication. In contrast, PACO is an exact
bounded-radix Cooley--Tukey FFT, performs no interprocessor
communication outside its sole BSP redistribution, and returns
canonical coefficients under an explicitly factor-swapped target slab
ownership.

Hong and Kung introduced the red-blue pebble framework for FFT I/O
lower bounds~\cite{HongKu81}, while Ranjan, Savage, and Zubair established the
FFT-DAG boundary result used for PACO's local batches~\cite{RanjanSaZu11}. PACO's
migration bound is narrower: it counts the items that must change
owner when materializing the prescribed fused permutation under the
fixed slab ownerships.
}

\subsection{Limitations}
\label{sec:limitations}

PACO assumes a fixed power-of-two radix and exact divisibility:
\(N=N_1N_2=b^m\), with \(p\), \(N_1\), and \(N_2\) aligned to the
base-\(b\) slab decomposition. Arbitrary sizes and processor counts
require an imbalance analysis. The algorithm returns canonical
coefficient indices under factor-swapped ownership; restoring the
original ownership may require another redistribution. The present
analysis is limited to one-dimensional slab decompositions, rather
than pencil or block distributions. Finally, it assumes homogeneous
ideal-cache and BSP parameters and excludes arithmetic recomputation
from its work and cache optimality claims.

\subsection{Conclusion}
\label{sec:conclusion}

PACO shows that cache-oblivious local FFT execution and a
one-redistribution parallel organization can coexist when recursive
layout changes are represented lazily and fused with the unavoidable
ownership change. Under the stated contract, it achieves optimal maximum per-processor work and cache complexity, uses the minimum number
of global redistribution rounds, and attains the exact migration cost
of its prescribed fused permutation. Extending these guarantees to
more general ownership distributions and processor configurations
remains open.

\balance
\clearpage

\bibliographystyle{ACM-Reference-Format}
\bibliography{papers}
\balance

\clearpage 
\appendix
\section{Deferred Layout Algebra and Local FFT Correctness}
\label{app:layout}

This appendix formalizes the deferred-layout interpretation used by the
local PACO FFT recursion. Its purpose is to prove that recursive transform
splits do not materialize physical transposes, and that the resulting lazy
layouts are precisely base-$b$ digit reversals. In particular, we establish
the two local operator identities used in Sect.~3:
\[
  \mathcal{A}_\ell = P_{\pi_\ell}\,\Omega_{b^\ell}
  \qquad\text{and}\qquad
  \mathcal{B}_\ell = \Omega_{b^\ell}\,P_{\pi_\ell}.
\]
The global materialization and redistribution of these layouts are handled
separately in Appendix~\ref{app:local-fft}.

\subsection{Coordinate, Digit-Field, and Layout Conventions}
\label{app:conv}

Fix the constant radix $b=2^{t}$. An $\ell$-digit base-$b$ string is written
from most significant to least significant digit,
\[
  u = d_{\ell-1}\pplus d_{\ell-2}\pplus\cdots\pplus d_0,\qquad d_r\in[b],
\]
where $\pplus$ denotes digit-field concatenation: if $U$ and $V$ have lengths
$\ell_U$ and $\ell_V$, then $U\pplus V$ is the $(\ell_U+\ell_V)$-digit field
whose digits in $U$ are more significant than those in $V$.

For $0\le\ell$, the base-$b$ digit-reversal permutation
$\pi_\ell:[b^\ell]\to[b^\ell]$ is
\[
  \pi_\ell(d_{\ell-1}\pplus\cdots\pplus d_0)=d_0\pplus\cdots\pplus d_{\ell-1},
\]
an involution: $\pi_\ell^{-1}=\pi_\ell$ and $\pi_\ell^2=\id{id}$.
For a permutation $\pi:[n]\to[n]$, the physical-to-logical layout operator is
$(P_\pi z)[i]=z[\pi(i)]$; thus, under layout $P_\pi$, the value at physical
position $i$ represents logical coordinate $\pi(i)$. Since digit reversal is
an involution,
\begin{equation}\label{eq:invol}
  P_{\pi_\ell}^{2}=I.
\end{equation}

We distinguish throughout:
\begin{itemize}
  \item \emph{Physical digit slots}: fixed base-$b$ positions in the physical
        address of an item.
  \item \emph{Logical transform-input digits}: the input coordinate of an
        active local transform.
  \item \emph{Logical output-frequency digits}: a transform frequency
        coordinate after a child transform has completed.
  \item \emph{Batch digits}: distinguish independent transform instances.
  \item \emph{Fixed digits}: selected by a recursive batch branch and hence
        holding one fixed value within that branch.
\end{itemize}
A transform split may change the logical role of a physical digit slot (e.g.
from a transform-input digit to an outer-batch digit); it never moves the
item stored at that physical address.

\subsection{Layout-Aware Twiddle Operator}
\label{app:twiddle}

Consider one local two-factor Cooley--Tukey split
$n=n_1n_2$, $n_1=b^{\ell_1}$, $n_2=b^{\ell_2}$, $\ell=\ell_1+\ell_2$, with
$j=j_1n_2+j_2$ and $k=k_1+k_2n_1$. Suppose that after the first
child-transform group the physical coordinate $a\in[n_1]$ represents the
logical first-stage frequency $k_1=\pi_{\ell_1}(a)$. The twiddle
multiplication must use this represented logical coordinate rather than $a$:
\begin{equation}\label{eq:twop}
  \bigl(D^{\mathrm{view}}_{\mathrm{tw}}z\bigr)[a,j_2]
  =\omega_n^{-\pi_{\ell_1}(a)\,j_2}\,z[a,j_2].
\end{equation}
Equivalently, the ordinary diagonal twiddle operator is conjugated by the
current layout permutation and evaluated lazily from the view descriptor; no
physical permutation is needed. The correctness of this lazy evaluation is
established in \appref{outid}.

\subsection{Admissible Digit-Field Views}
\label{app:views}

Let $D=\{d_0,\dots,d_{D-1}\}$ be the digit labels of the physical base-$b$
slots of a root local array (slot $0$ least significant). A recursive view is
a tuple $V=(T,O,C,F,P,\mathrm{mode})$:
\begin{itemize}
  \item $T$: ordered active transform digit field;
  \item $O$: ordered list of active outer-batch digit fields;
  \item $C$: distinguished contiguous batch field;
  \item $F$: partial assignment from fixed digit labels to values in $[b]$;
  \item $P$: bijection from current digit labels to physical digit slots;
  \item $\mathrm{mode}\in\{\textsc{OutputLayout},\textsc{InputLayout}\}$.
\end{itemize}
The active labels in $T,O,C$ and the fixed labels in $\id{dom}(F)$ are disjoint
and cover the root digit labels. A label may be renamed when its logical role
changes, but its physical slot under $P$ is unchanged. Let
\[
  n=b^{|T|},\quad o=b^{|O|},\quad c=b^{|C|},\quad s=oc,
\]
so the view represents $s$ transforms of length $n$; a logical item has
coordinate $(\tau,\eta,\chi)\in[n]\times[o]\times[c]$.

For a digit label $d$, define $\id{stride}_P(d)=b^{\id{pos}_P(d)}$. The active digits
of $C$ occupy the least-significant contiguous physical interval, slots
$0,\dots,|C|-1$ (fixed digits may occupy immediately higher slots without
interrupting the remaining active interval). The address map is
\begin{align}
\func{addr}_V(\tau,\eta,\chi)&=\id{base}_0
    +\!\!\sum_{d\in\func{dom}(F)}\!\! F(d)\,\func{stride}_P(d)\nonumber\\
    &\quad +\sum_{d\in T}\id{digit}_d(\tau)\,\func{stride}_P(d)\nonumber\\
    &\quad +\sum_{d\in O}\id{digit}_d(\eta)\,\func{stride}_P(d)+\chi.\label{eq:addr}
\end{align}
For fixed $(\tau,\eta,F)$, set $\id{base}_{\tau,\eta,F}=\func{addr}_V(\tau,\eta,0)$;
then $\func{addr}_V(\tau,\eta,\chi)=\id{base}_{\tau,\eta,F}+\chi$, so varying $\chi$
traverses a contiguous run.

\begin{definition}[Admissible digit-field view]\label{eq:def}
A view is \emph{admissible} if:
(1) its active and fixed labels partition the root digit slots;
(2) $P$ assigns each current label a unique physical slot;
(3) the active digits of $C$ form one least-significant contiguous interval;
(4) its address map is \eqref{addr}; and
(5) every logical coordinate maps to a distinct physical item.
\end{definition}

\begin{lemma}[Active-run decomposition]\label{eq:runs}
An admissible view with $s$ transforms of length $n$ is the disjoint union of
$ns/c$ contiguous physical runs, each of length $c$.
\end{lemma}
\begin{proof}
There is one run per pair $(\tau,\eta)\in[n]\times[o]$; varying $\chi\in[c]$
gives a consecutive interval of length $c$. There are $no=ns/c$ such pairs,
and injectivity of the address map makes the intervals disjoint and covering.
\end{proof}

The distinguished field $C$ is retained to expose spatial locality even when
the active transform field is strided; this is used by the cache analysis in
Appendix~\ref{app:colocal}.

\subsection{Legal View Operations}
\label{app:ops}

The scheduler uses outer-batch splits, contiguous-field splits, and transform
splits. Batch splits partition batch instances into data-disjoint children; a
transform split creates two sequential child groups reinterpreting the same
parent working set.

\paragraph{Outer-batch split.}
If $d\in O$ is selected, the $r$-th child ($r\in[b]$) has
$T_r=T$, $O_r=O\setminus\{d\}$, $C_r=C$, $F_r=F\cup\{d\mapsto r\}$; $P$ and
$\mathrm{mode}$ are unchanged. The $b$ children are data disjoint.

\paragraph{Contiguous-field split.}
Used only when $O=\emptyset$. Write $C=d\pplus C'$ with $d$ the most significant
active digit of $C$; the $r$-th child has $T_r=T$, $O_r=\emptyset$, $C_r=C'$,
$F_r=F\cup\{d\mapsto r\}$. The child width is $c/b$, and $C'$ still occupies
the least-significant contiguous interval.

\paragraph{Transform split.}
Decompose $T=T_1\pplus T_2$ with $\ell_1=|T_1|$, $\ell_2=|T_2|$,
$n_1=b^{\ell_1}$, $n_2=b^{\ell_2}$, $n=n_1n_2$, and write $j=j_1n_2+j_2$.
The first child group performs length-$n_1$ transforms over $j_1$:
\begin{equation}\label{eq:firstfields}
  T_{\mathrm{I}}=T_1,\quad O_{\mathrm{I}}=T_2\pplus O,\quad C_{\mathrm{I}}=C,
\end{equation}
with $sn_2$ transforms and inherited address map
\begin{equation}\label{eq:firstaddr}
    \func{addr}_{V_{\mathrm{I}}}\!\bigl(j_1,(j_2,\eta),\chi\bigr)=\func{addr}_V\bigl(j_1n_2+j_2,\eta,\chi\bigr).
\end{equation}
After the first child completes, the slots formerly holding $T_1$ contain
first-stage outputs; denote this physical field $A$. This is a semantic
relabeling only:
\begin{equation}\label{eq:relabel}
  \func{pos}_P(A_j)=\func{pos}_P\bigl((T_1)_j\bigr).
\end{equation}
In output-layout mode a physical coordinate $a\in[n_1]$ of $A$ represents
$k_1=\pi_{\ell_1}(a)$. The second child group performs length-$n_2$
transforms over $T_2$:
\begin{equation}\label{eq:secondfields}
  T_{\mathrm{II}}=T_2,\quad O_{\mathrm{II}}=A\pplus O,\quad C_{\mathrm{II}}=C,
\end{equation}
with $sn_1$ transforms and address map
\begin{equation}\label{eq:secondaddr}
  \func{addr}_{V_{\mathrm{II}}}\!\bigl(j_2,(a,\eta),\chi\bigr)
  =\func{addr}_V\bigl(an_2+j_2,\eta,\chi\bigr).
\end{equation}
The field $T_2$ may be strided past the physical field $A$; $P$ records this
interpretation, and no exchange of the two fields is performed. The twiddle
$\omega_n^{-\pi_{\ell_1}(a)\,j_2}$ is fused into the first access of the
corresponding second-child input, as in \eqref{twop}.

\subsection{Closure Under Arbitrary Split Interleavings}
\label{app:closure}

\begin{theorem}[Closure of admissible views]\label{eq:thm7}
Starting from an admissible root view, every finite sequence of legal
outer-batch, contiguous-field, and transform splits produces admissible child
views. Moreover:
(1) active and fixed labels continue to partition the physical slots;
(2) each label retains its physical slot under $P$, except for the semantic
relabeling of a transform-input field as a transform-output field;
(3) $C$ remains one least-significant contiguous interval;
(4) the children of every batch split are data disjoint and partition the
parent address set;
(5) each transform child group is a bijective reindexing of the complete
parent working set; and
(6) invariant \eqref{addr} is preserved under arbitrary interleavings.
\end{theorem}
\begin{proof}
Induction on the number of operations; the base case is Definition~\ref{eq:def}.
An outer-batch split moves one digit from $O$ to $F$ with a fixed value,
leaving $T,C$ and all physical slots unchanged, and its $b$ values partition
the parent batch. A contiguous-field split moves the most significant digit
of $C$ to $F$; the remaining active digits still form the least-significant
interval, and the $b$ values partition every run into $b$ subruns. A transform
split has child maps \eqref{firstaddr} and \eqref{secondaddr}, each a
bijective reindexing of the parent set; it does not change $F$ or the slots of
$C$, and the relabeling $T_1\to A$ preserves slots by \eqref{relabel}.
Thus every operation preserves Definition~\ref{eq:def}.
\end{proof}

The fixed assignment $F$ is necessary because batch and transform splits can
interleave. For example, with physical digit order $[a_1][a_0][c_1][c_0]$ and
$T=[a_1,a_0]$, $O=\emptyset$, $C=[c_1,c_0]$, $F=\emptyset$, a contiguous-field
split fixing $c_1=d$ gives $C=[c_0]$, $F=\{c_1\mapsto d\}$; splitting
$[a_1,a_0]=[a_1]\pplus[a_0]$ makes the second child's transform field $[a_0]$
with outer-batch field $[A_1]$; a later batch split fixing $A_1=e$ yields
$T=[a_0]$, $O=\emptyset$, $C=[c_0]$, $F=\{c_1\mapsto d,\ A_1\mapsto e\}$.
Although the fixed slots need not be adjacent, varying $c_0$ still traverses a
contiguous run---exactly the information recorded by $F$ and $P$.

\subsection{Transform Splits Are View Transformations}
\label{app:notranspose}

\begin{theorem}[No physical transposes inside a local transform]\label{eq:thm8}
No transform split in the local PACO recursion materializes a physical
permutation. The transition from the first child group to the second changes
only (1) the logical role of the physical $T_1$ slots, renamed as $A$;
(2) the active transform field, from $T_1$ to $T_2$; and (3) the descriptor
$P$. All items remain in their original physical positions.
\end{theorem}
\begin{proof}
The child maps \eqref{firstaddr} and \eqref{secondaddr} refer only to
physical addresses of the parent. The first child overwrites each parent item
with its first-stage result; relabeling $T_1\to A$ preserves positions by
\eqref{relabel}; the second child accesses the inherited $T_2$ slots through
view-defined strides while treating $A$ as outer-batch. Hence no
address-to-address copy, transpose, or recursive layout permutation occurs,
and the argument applies recursively at every transform-split node.
\end{proof}

The only physical digit permutation in the complete PACO execution is the
fused middle permutation of Sect.~4.2 and Appendix~\ref{app:global}; the
macro-element digit reversals used there are components of that single
materialization, not recursive transposes.

\subsection{Output-Layout Local FFT Correctness}
\label{app:outid}

Let a view contain $s$ data-disjoint transforms of length
$n=b^\ell=n_1n_2$, and for batch instance $q$ let $Y[k,q]=\sum_j V[j,q]\,
\omega_n^{-kj}$.

\begin{theorem}[Output-layout local FFT identity]\label{eq:thm9}
For every $\ell\ge0$ the output-layout traversal implements
$\mathcal{A}_\ell=P_{\pi_\ell}\Omega_{b^\ell}$; equivalently, for a physical
output coordinate $r$, $Y^{\mathrm{phys}}[r,q]=Y[\pi_\ell(r),q]$.
\end{theorem}
\begin{proof}
Induction on $\ell$. For $\ell=0$ there is no transform coordinate; for
$\ell=1$ the routine performs one base-$b$ DFT and $\pi_1=\id{id}$. Let $\ell>1$
and split $\ell=\ell_1+\ell_2$, $j=j_1n_2+j_2$, $k=k_1+k_2n_1$. The first
child computes, for each $(j_2,q)$,
\begin{equation}\label{eq:H}
  H[k_1,j_2,q]=\sum_{j_1=0}^{n_1-1}V[j_1n_2+j_2,q]\,\omega_{n_1}^{-k_1j_1},
\end{equation}
so by induction its physical coordinate $a$ represents $k_1=\pi_{\ell_1}(a)$:
$H^{\mathrm{phys}}[a,j_2,q]=H[\pi_{\ell_1}(a),j_2,q]$. Applying the twiddle,
$Z[k_1,j_2,q]=H[k_1,j_2,q]\,\omega_n^{-k_1j_2}$, gives
\begin{equation}\label{eq:Zphys}
  Z^{\mathrm{phys}}[a,j_2,q]=Z[\pi_{\ell_1}(a),j_2,q].
\end{equation}
The second child computes, for each $(k_1,q)$,
$Y[k_1+k_2n_1,q]=\sum_{j_2}Z[k_1,j_2,q]\,\omega_{n_2}^{-k_2j_2}$. Substituting
\eqref{H} and the twiddle,
\[
  Y[k_1+k_2n_1,q]
  =\sum_{j_2=0}^{n_2-1}\sum_{j_1=0}^{n_1-1}
      V[j_1n_2+j_2,q]\,
      \omega_{n_1}^{-k_1j_1}\,\omega_n^{-k_1j_2}\,\omega_{n_2}^{-k_2j_2}.
\]
Using $\omega_n^{n_2}=\omega_{n_1}$ and $\omega_n^{n_1}=\omega_{n_2}$, the
three factors combine into $\omega_n^{-(k_1j_1n_2+k_1j_2+k_2n_1j_2)}$.
Adjoining the cross term $\omega_n^{-k_2n_1j_1n_2}
=\omega_n^{-k_2j_1(n_1n_2)}=(\omega_n^{\,n})^{-k_2j_1}=1$ completes the
exponent, so
\begin{equation}\label{eq:CT}
  Y[k_1+k_2n_1,q]
  =\sum_{j_1,j_2}V[j_1n_2+j_2,q]\,
    \omega_n^{-(k_1+k_2n_1)(j_1n_2+j_2)},
\end{equation}
which, as $j=j_1n_2+j_2$ ranges over $[n]$, is the length-$n$ DFT at frequency
$k_1+k_2n_1$. By \eqref{Zphys} the second child's input for fixed physical
$a$ is $j_2\mapsto Z[\pi_{\ell_1}(a),j_2,q]$; by induction its physical output
coordinate $c$ represents $k_2=\pi_{\ell_2}(c)$, whence
\begin{equation}\label{eq:finalfreq}
  Y^{\mathrm{phys}}[a,c,q]=Y\bigl[\pi_{\ell_1}(a)+\pi_{\ell_2}(c)n_1,\,q\bigr].
\end{equation}
Finally, the parent physical coordinate $r=an_2+c$ has digit string
$(\text{digits of }a)\pplus(\text{digits of }c)$, so
$\pi_\ell(an_2+c)=\pi_{\ell_1}(a)+\pi_{\ell_2}(c)n_1$; combining with
\eqref{finalfreq} yields $Y^{\mathrm{phys}}[r,q]=Y[\pi_\ell(r),q]$.
\end{proof}

\subsection{Input-Layout Local FFT Correctness}
\label{app:inid}

\begin{theorem}[Input-layout local FFT identity]\label{eq:thm10}
For every $\ell\ge0$ the input-layout traversal implements
$\mathcal{B}_\ell=\Omega_{b^\ell}P_{\pi_\ell}$. In particular, if
$x^{\mathrm{phys}}=P_{\pi_\ell}x$, then
$\mathcal{B}_\ell x^{\mathrm{phys}}=\Omega_{b^\ell}x$.
\end{theorem}
\begin{proof}
The input-layout traversal uses digit-reversed entrance and exit
interpretations, so $\mathcal{B}_\ell=P_{\pi_\ell}\mathcal{A}_\ell
P_{\pi_\ell}$. By Theorem~\ref{eq:thm9} and \eqref{invol},
$\mathcal{B}_\ell=P_{\pi_\ell}(P_{\pi_\ell}\Omega_{b^\ell})P_{\pi_\ell}
=\Omega_{b^\ell}P_{\pi_\ell}$, hence
$\mathcal{B}_\ell x^{\mathrm{phys}}=\Omega_{b^\ell}P_{\pi_\ell}P_{\pi_\ell}x
=\Omega_{b^\ell}x$. The two occurrences of $P_{\pi_\ell}$ are view
interpretations, not runtime physical permutations.
\end{proof}

\subsection{Local Correctness Summary}
\label{app:sum}

\begin{corollary}[Batched local FFT operators]\label{eq:cor11}
For $s$ independent transforms of length $b^\ell$, the output- and
input-layout traversals implement
$\mathcal{A}_{\ell,s}=I_s\otimes(P_{\pi_\ell}\Omega_{b^\ell})$ and
$\mathcal{B}_{\ell,s}=I_s\otimes(\Omega_{b^\ell}P_{\pi_\ell})$, respectively.
\end{corollary}
\begin{proof}
Theorem~\ref{eq:thm7} preserves data disjointness across batch instances;
applying Theorems~\ref{eq:thm9} and \ref{eq:thm10} per instance gives the
Kronecker forms.
\end{proof}

Consequently, Phase~I invokes $\mathcal{A}_{m_1}$ (physical row $a$ represents
$\pi_{m_1}(a)$); the fused middle stage materializes
$\pi_m(a\pplus c)=\pi_{m_2}(c)\pplus\pi_{m_1}(a)$, converting the Phase-I output
layout into the Phase-III input layout; and Phase~III invokes
$\mathcal{B}_{m_2}$, producing natural-order coefficients as used in
Theorem~3.

\punt{
\appput{layout}{Deferred Layout Algebra and Local FFT Correctness}

This appendix formalizes the deferred-layout interpretation used by
the local PACO FFT recursion. Its purpose is to prove that recursive
transform splits do not materialize physical transposes, and that the
resulting lazy layouts are precisely base-$b$ digit reversals. In
particular, we establish the two local operator identities used in \secref{overview}:
\[
A_{\ell}=P_{\pi_{\ell}}\Omega_{b^{\ell}}
\qquad\text{and}\qquad
B_{\ell}=\Omega_{b^{\ell}}P_{\pi_{\ell}}.
\]
The global materialization and redistribution of these layouts are
handled separately in \appref{local-fft}.

\subsection{Coordinate, Digit-Field, and Layout Conventions}

Fix the constant radix $b=2^t$. An $\ell$-digit base-$b$ string is
written from most significant digit to least significant digit:
\[
u=d_{\ell-1}\pplus d_{\ell-2}
\pplus \cdots\pplus d_0,
\qquad d_r\in[b].
\]
The notation $\pplus $ denotes digit-field concatenation.
Thus, if $U$ and $V$ are digit fields of lengths $\ell_U$ and
$\ell_V$, respectively, then
\[
U\pplus V
\]
denotes the $(\ell_U+\ell_V)$-digit field whose digits in $U$ are more
significant than those in $V$.

For $0\leq \ell$, define the base-$b$ digit-reversal permutation
$\pi_{\ell}:[b^{\ell}]\rightarrow[b^{\ell}]$ by
\[
\pi_{\ell}
\left(
d_{\ell-1}\pplus \cdots\pplus d_0
\right)
=
d_0\pplus \cdots\pplus d_{\ell-1}.
\]
Digit reversal is an involution:
\[
\pi_{\ell}^{-1}=\pi_{\ell},
\qquad
\pi_{\ell}^2=\mathrm{id}.
\]

For a permutation $\pi:[n]\rightarrow[n]$, let $P_{\pi}$ denote the
physical-to-logical layout operator
\[
(P_{\pi}z)[i]=z[\pi(i)].
\label{eq:app-layout-operator}
\]
Hence, if an array is stored under layout $P_{\pi}$, then the value at
physical position $i$ represents logical coordinate $\pi(i)$. Since
digit reversal is an involution,
\begin{equation}
P_{\pi_{\ell}}^2=I.
\label{eq:app-digit-reversal-involution}
\end{equation}

We distinguish throughout between the following notions.

\begin{itemize}
    \item \emph{Physical digit slots} are fixed base-$b$ positions in
    the physical address of an item.

    \item \emph{Logical transform-input digits} describe the input
    coordinate of an active local transform.

    \item \emph{Logical output-frequency digits} describe a transform
    frequency coordinate after a child transform has completed.

    \item \emph{Batch digits} distinguish independent transform
    instances.

    \item \emph{Fixed digits} have been selected by a recursive batch
    branch and therefore have one fixed value within that branch.
\end{itemize}

A transform split may change the logical role of a physical digit slot,
for example from a transform-input digit to an outer-batch digit. It never moves the item stored at that
physical address.

\subsection{Layout-Aware Twiddle Operators}

Consider one local two-factor Cooley--Tukey split
\[
n=n_1n_2,
\qquad
n_1=b^{\ell_1},
\qquad
n_2=b^{\ell_2},
\qquad
\ell=\ell_1+\ell_2.
\]
Write
\[
j=j_1n_2+j_2,
\qquad
k=k_1+k_2n_1,
\]
where
\[
j_1,k_1\in[n_1],
\qquad
j_2,k_2\in[n_2].
\]
Suppose that, after the first child-transform group, physical coordinate
$a\in[n_1]$ represents logical first-stage frequency
\[
k_1=\pi_{\ell_1}(a).
\]
The twiddle multiplication must use this represented logical
coordinate, rather than the physical coordinate $a$. Thus the
layout-aware twiddle operator is
\begin{equation}
\bigl(D_{\mathrm{tw}}^{\mathrm{view}}z\bigr)[a,j_2]
=
\omega_n^{-\pi_{\ell_1}(a)j_2}z[a,j_2].
\label{eq:app-layout-aware-twiddle}
\end{equation}
Equivalently, the ordinary diagonal twiddle operator is interpreted
under the current layout permutation. No physical permutation is needed
to evaluate~\eqref{app-layout-aware-twiddle}: the recursive view
descriptor identifies the logical digit values represented by the
current physical coordinate. As in the main algorithm, PACO evaluates
the required twiddle from these logical coordinates when the second
child group first accesses the corresponding intermediate value. This
fuses twiddle multiplication with the child computation and does not
require a separate traversal or a large precomputed twiddle table.

\subsection{Admissible Digit-Field Views}

Let
\[
D=\{d_0,d_1,\ldots,d_{D-1}\}
\]
be the digit labels associated with the physical base-$b$ digit slots
of a root local array. Physical slot $0$ is the least significant
physical slot. A recursive view is a tuple
\[
V=(T,O,C,F,P,\mathit{mode}).
\label{eq:app-view-tuple}
\]
Its components are as follows.

\begin{itemize}
    \item $T$ is the ordered active transform digit field.

    \item $O$ is an ordered list of active outer-batch digit fields.

    \item $C$ is the distinguished contiguous batch field.

    \item $F$ is a partial assignment from fixed digit labels to values
    in $[b]$.

    \item $P$ is a bijection from current digit labels to physical
    digit slots.

    \item $\mathit{mode}\in
    \{\mathit{OutputLayout},\mathit{InputLayout}\}$ specifies which
    local traversal interpretation is used.
\end{itemize}

The active labels in $T$, $O$, and $C$, together with the fixed labels
in $\operatorname{dom}(F)$, are disjoint and cover the digit labels of
the root view. A label may be renamed when its logical role changes,
but its physical slot under $P$ remains unchanged.

Let
\[
n=b^{|T|},
\qquad
o=b^{|O|},
\qquad
c=b^{|C|},
\qquad
s=oc.
\label{eq:app-view-parameters}
\]
Thus, the view represents $s$ transforms of length $n$. The coordinate
of one logical item is
\[
(\tau,\eta,\chi)\in[n]\times[o]\times[c],
\]
where $\tau$ is the active transform coordinate, $\eta$ is the
flattened outer-batch coordinate, and $\chi$ is the distinguished
contiguous batch coordinate.

For a digit label $d$, define
\[
\operatorname{stride}_{P}(d)=b^{\operatorname{pos}_{P}(d)}.
\]
The active digits of $C$ occupy one least-significant contiguous
physical interval. In particular, they occupy physical slots
\[
0,1,\ldots,|C|-1.
\]
Digits fixed by earlier contiguous-field splits may occupy immediately
higher physical slots, but they do not interrupt the remaining active
interval.

Let $\operatorname{digit}_{d}(\tau)$ and
$\operatorname{digit}_{d}(\eta)$ denote the value of active digit $d$
in the appropriate transform or outer-batch coordinate. The physical
address map of $V$ is
\begin{align}
\operatorname{addr}_{V}(\tau,\eta,\chi)
={}&\operatorname{base}_0
+\sum_{d\in\operatorname{dom}(F)}
F(d)\operatorname{stride}_{P}(d)
\nonumber\\
&+\sum_{d\in T}
\operatorname{digit}_{d}(\tau)\operatorname{stride}_{P}(d)\nonumber\\
&+\sum_{d\in O}
\operatorname{digit}_{d}(\eta)\operatorname{stride}_{P}(d)
+\chi.
\label{eq:app-address-map}
\end{align}
For fixed $(\tau,\eta)$ and fixed assignment $F$, define
\[
\operatorname{base}_{\tau,\eta,F}
=
\operatorname{addr}_{V}(\tau,\eta,0).
\]
Then
\begin{equation}
\operatorname{addr}_{V}(\tau,\eta,\chi)
=
\operatorname{base}_{\tau,\eta,F}+\chi.
\label{eq:app-contiguous-run}
\end{equation}
Thus, varying $\chi$ traverses a physically contiguous active run.

\begin{definition}[Admissible digit-field view]
\label{def:app-admissible-view}
A view is \emph{admissible} if:
\begin{enumerate}
    \item its active and fixed digit labels partition the digit slots of
    its root view;

    \item $P$ assigns each current label a unique physical slot;

    \item the active digits of $C$ form one least-significant contiguous
    physical interval;

    \item its address map is given by~\eqref{app-address-map}; and

    \item every logical view coordinate maps to a distinct physical
    item.
\end{enumerate}
\end{definition}

\begin{lemma}[Active-run decomposition]
\label{lem:app-active-runs}
An admissible view containing $s$ transforms of length $n$ is the
disjoint union of
\[
\frac{ns}{c}
\]
contiguous physical runs, each of length $c$.
\end{lemma}

\begin{proof}
There is one run for each pair
\[
(\tau,\eta)\in[n]\times[o].
\]
By~\eqref{app-contiguous-run}, varying $\chi\in[c]$ produces a
consecutive interval of length $c$. There are $no=ns/c$ such pairs.
Injectivity of the address map implies that these intervals are
disjoint and cover the complete view.
\end{proof}

The distinguished field $C$ is retained to expose spatial locality even
when the active transform field is strided. This property is used by the
cache analysis in \appref{digit-reversal}.

\subsection{Legal View Operations}

The local scheduler uses three view operations: outer-batch splits,
contiguous-field splits, and transform splits. Batch splits partition
the current batch instances into data-disjoint children. A transform
split creates two sequential child-transform groups that reinterpret the
same parent working set.

\subsubsection{Outer-batch split}

Suppose an active outer-batch digit $d\in O$ is selected. For each
$r\in[b]$, the $r$-th child fixes $d=r$ and has fields
\[
T_r=T,
\qquad
O_r=O\setminus\{d\},
\qquad
C_r=C,
\]
with fixed assignment
\[
F_r=F\cup\{d\mapsto r\}.
\]
The physical map $P$ and the traversal mode are unchanged. Equivalently,
the child base offset is increased by
\[
r\operatorname{stride}_{P}(d).
\]
The $b$ children are data disjoint and partition the parent batch
instances.

\subsubsection{Contiguous-field split}

A contiguous-field split is used only when
\[
O=\varnothing.
\]
Write
\[
C=d\pplus C',
\]
where $d$ is the most significant active digit of $C$. For every
$r\in[b]$, the $r$-th child has
\[
T_r=T,
\qquad
O_r=\varnothing,
\qquad
C_r=C',
\]
and
\[
F_r=F\cup\{d\mapsto r\}.
\]
The child contiguous width is $c/b$. Because the most significant
active digit of $C$ is fixed, the remaining digits of $C'$ still occupy
the least-significant contiguous physical interval. Therefore, every
parent active run is partitioned into $b$ contiguous child subruns.

\subsubsection{Transform split}

Let the active transform field be decomposed as
\[
T=T_1\pplus T_2,
\]
where
\[
\ell_1=|T_1|,
\qquad
\ell_2=|T_2|,
\qquad
n_1=b^{\ell_1},
\qquad
n_2=b^{\ell_2},
\qquad
n=n_1n_2.
\]
Write a parent logical transform coordinate as
\[
j=j_1n_2+j_2,
\qquad
j_1\in[n_1],
\qquad
j_2\in[n_2].
\]

The first child-transform group performs length-$n_1$ transforms over
$j_1$, indexed by $j_2$ and the existing batch coordinates. Its
persistent view fields are
\begin{equation}
T_{\mathrm{I}}=T_1,
\qquad
O_{\mathrm{I}}=T_2\pplus O,
\qquad
C_{\mathrm{I}}=C.
\label{eq:app-first-child-view}
\end{equation}
It represents $sn_2$ transforms of length $n_1$, and its inherited
address map is
\begin{equation}
\operatorname{addr}_{V_{\mathrm{I}}}
\bigl(j_1,(j_2,\eta),\chi\bigr)
=
\operatorname{addr}_{V}
\bigl(j_1n_2+j_2,\eta,\chi\bigr).
\label{eq:app-first-child-address}
\end{equation}
Thus, $T_2$ becomes an outer-batch field without moving any item.

After the first child group completes, the physical slots formerly
occupied by $T_1$ contain first-stage output coordinates. Denote this
physical field by $A$. This is a semantic relabeling only:
\begin{equation}
\operatorname{pos}_{P}(A_j)
=
\operatorname{pos}_{P}\bigl((T_1)_j\bigr).
\label{eq:app-output-field-relabeling}
\end{equation}
In output-layout mode, a physical coordinate $a\in[n_1]$ in this field
represents logical first-stage frequency
\begin{equation}
k_1=\pi_{\ell_1}(a).
\label{eq:app-first-stage-layout}
\end{equation}

The second child-transform group performs length-$n_2$ transforms over
the physical slots of $T_2$, indexed by the first-stage physical output
coordinate $a$ and the pre-existing batch coordinates. Its persistent
fields are
\begin{equation}
T_{\mathrm{II}}=T_2,
\qquad
O_{\mathrm{II}}=A\pplus O,
\qquad
C_{\mathrm{II}}=C.
\label{eq:app-second-child-view}
\end{equation}
It represents $sn_1$ transforms of length $n_2$, with address map
\begin{equation}
\operatorname{addr}_{V_{\mathrm{II}}}
\bigl(j_2,(a,\eta),\chi\bigr)
=
\operatorname{addr}_{V}
\bigl(an_2+j_2,\eta,\chi\bigr).
\label{eq:app-second-child-address}
\end{equation}
The active transform field $T_2$ can be strided past the physical
first-stage output field $A$. The descriptor $P$ records this
interpretation; no exchange of the two physical digit fields is
performed.

The twiddle
\[
\omega_n^{-k_1j_2}
=
\omega_n^{-\pi_{\ell_1}(a)j_2}
\]
is evaluated from the logical coordinates represented by the view and
is fused into the first access of the corresponding second-child input.

\subsection{Closure Under Arbitrary Split Interleavings}

\begin{theorem}[Closure of admissible views]
\label{thm:app-view-closure}
Starting from an admissible root view, every finite sequence of legal
outer-batch splits, contiguous-field splits, and transform splits
produces admissible child views. More precisely:
\begin{enumerate}
    \item active labels and fixed labels continue to partition the
    physical digit slots;

    \item each label retains its physical slot under $P$, except for
    semantic relabeling of a transform-input field as a
    transform-output field;

    \item the active field $C$ remains one least-significant contiguous
    physical interval;

    \item the children of every batch split are data disjoint and
    partition the parent address set;

    \item each transform child group is a bijective reindexing of the
    complete parent working set; and

    \item the address invariant~\eqref{app-address-map} is preserved
    under arbitrary interleavings.
\end{enumerate}
\end{theorem}

\begin{proof}
We use induction on the number of legal view operations. The properties
hold for the root by Definition~\ref{def:app-admissible-view}.

For an outer-batch split, one active outer digit is removed from $O$ and
inserted into $F$ with one fixed value. Its physical slot under $P$ is
unchanged. The $b$ possible fixed values partition the parent batch
coordinates, while $T$ and $C$ remain unchanged.

For a contiguous-field split, the most significant active digit of $C$
is inserted into $F$. The remaining active digits still occupy the
least-significant physical interval. The $b$ possible values of the
removed digit partition every parent run into $b$ consecutive child
subruns.

For a transform split, the first and second child address maps are given
by~\eqref{app-first-child-address} and
\eqref{app-second-child-address}. Each is a bijective reindexing of
the complete parent working set. The split neither changes $F$ nor
moves the physical slots of $C$. The relabeling from $T_1$ to $A$
preserves physical slots by~\eqref{app-output-field-relabeling}.

Thus every legal operation preserves the conditions of
Definition~\ref{def:app-admissible-view}.
\end{proof}

The fixed assignment $F$ is necessary because batch and transform splits
can interleave. For example, consider physical digit order
\[
[a_1][a_0][c_1][c_0].
\]
Initially, let
\[
T=[a_1,a_0],
\qquad
O=\varnothing,
\qquad
C=[c_1,c_0],
\qquad
F=\varnothing.
\]
A contiguous-field split can first fix $c_1=d$, leaving
\[
T=[a_1,a_0],
\qquad
C=[c_0],
\qquad
F=\{c_1\mapsto d\}.
\]
If the transform field is next split as
\[
[a_1,a_0]=[a_1]\pplus [a_0],
\]
then the first child has outer batch field $[a_0]$. After that child
completes, the second child has transform field $[a_0]$ and outer batch
field $[A_1]$, where $A_1$ occupies the former $a_1$ slot. A subsequent
batch split selects the outer field $A_1$ before splitting $[c_0]$.
Fixing $A_1=e$ gives
\[
T=[a_0],
\qquad
O=\varnothing,
\qquad
C=[c_0],
\qquad
F=\{c_1\mapsto d,A_1\mapsto e\}.
\]
Although the fixed slots need not be adjacent, varying $c_0$ still
traverses a contiguous physical run. This is exactly the information
recorded by $F$ and $P$.

\subsection{Transform Splits Are View Transformations}

\begin{theorem}[No physical transposes inside a local transform]
\label{thm:app-no-local-transpose}
No transform split in the local PACO recursion materializes a physical
permutation. The transition from the first child-transform group to the
second changes only:
\begin{enumerate}
    \item the logical role of the physical $T_1$ slots, which are
    renamed as the output field $A$;

    \item the active transform field, from $T_1$ to $T_2$; and

    \item the view descriptor $P$, which records the physical strides
    and logical interpretations of the new transform and outer-batch
    fields.
\end{enumerate}
All items remain in their original physical memory positions.
\end{theorem}

\begin{proof}
The first child address map~\eqref{app-first-child-address} and the
second child address map~\eqref{app-second-child-address} both refer
only to physical addresses of the parent view. The first child
overwrites each parent item with its first-stage transform result.
Relabeling $T_1$ as $A$ preserves all physical positions by
\eqref{app-output-field-relabeling}. The second child subsequently
accesses the inherited $T_2$ slots through their view-defined strides
while treating the $A$ slots as outer-batch coordinates.

Therefore, no address-to-address copy, transpose, or recursive layout
permutation occurs at a transform split. The same argument applies
recursively at every transform-split node.
\end{proof}

The only physical digit permutation in the complete PACO execution is
the fused middle permutation materialized by packing, one BSP exchange,
and unpacking in Sect.~4.2 and Appendix~B. The macro-element digit
reversals used in that stage are components of this single middle
materialization; they are not recursive transposes inside either local
FFT phase.

\subsection{Output-Layout Local FFT Correctness}

We now prove the output-layout identity for one generic local transform.
Let a view contain $s$ data-disjoint transforms of length
\[
n=b^{\ell}=n_1n_2,
\qquad
n_1=b^{\ell_1},
\qquad
n_2=b^{\ell_2},
\qquad
\ell=\ell_1+\ell_2.
\]
For a batch instance $q\in[s]$, let $V[j,q]$ be the logical input and
define
\[
Y[k,q]
=
\sum_{j=0}^{n-1}
V[j,q]\omega_n^{-kj}.
\]

\begin{theorem}[Output-layout local FFT identity]
\label{thm:app-output-layout}
For every $\ell\geq 0$, the output-layout local traversal implements
\begin{equation}
A_{\ell}=P_{\pi_{\ell}}\Omega_{b^{\ell}}.
\label{eq:app-output-layout-identity}
\end{equation}
Equivalently, if $r$ is a physical output coordinate, then
\[
Y^{\mathrm{phys}}[r,q]
=
Y[\pi_{\ell}(r),q].
\]
\end{theorem}
\begin{proof}
We prove the claim by induction on the number $\ell$ of transform
digits.  Write
\[
n=b^\ell.
\]

For $\ell=0$, there is no transform coordinate, so the claim is
immediate.  For $\ell=1$, the routine performs one base-$b$ DFT, and
the digit-reversal permutation $\pi_1$ is the identity.  Hence the
physical output coordinate represents the same logical frequency
coordinate, as required.

Now suppose $\ell>1$, and assume the theorem holds for all shorter
transform fields.  Let the transform field be split as
\[
\ell=\ell_1+\ell_2,
\qquad
n_1=b^{\ell_1},
\qquad
n_2=b^{\ell_2},
\qquad
n=n_1n_2.
\]
Write the logical input and output indices as
\[
j=j_1n_2+j_2,
\qquad
k=k_1+k_2n_1,
\]
where
\[
j_1,k_1\in[n_1],
\qquad
j_2,k_2\in[n_2].
\]

For each fixed $(j_2,q)$, the first child group performs a length-$n_1$
transform over $j_1$.  Define its logical output by
\begin{equation}
H[k_1,j_2,q]
=
\sum_{j_1=0}^{n_1-1}
V[j_1n_2+j_2,q]\,
\omega_{n_1}^{-k_1j_1}.
\label{eq:app-first-stage-transform}
\end{equation}
The first child has $\ell_1$ transform digits.  Therefore, by the
induction hypothesis, its physical output coordinate $a\in[n_1]$
represents the logical first-stage frequency
\[
k_1=\pi_{\ell_1}(a).
\]
Equivalently, after the first child group has completed,
\begin{equation}
H^{\mathrm{phys}}[a,j_2,q]
=
H[\pi_{\ell_1}(a),j_2,q].
\label{eq:app-first-stage-layout}
\end{equation}

The algorithm then applies the Cooley--Tukey twiddle factor.  Define
the twiddled intermediate array by
\begin{equation}
Z[k_1,j_2,q]
=
H[k_1,j_2,q]\,
\omega_n^{-k_1j_2}.
\label{eq:app-twiddled-intermediate}
\end{equation}
Thus, using \eqref{app-first-stage-layout}, the physical data supplied
to the second child group satisfy
\begin{equation}
Z^{\mathrm{phys}}[a,j_2,q]
=
Z[\pi_{\ell_1}(a),j_2,q].
\label{eq:app-twiddled-layout}
\end{equation}

For each fixed $(k_1,q)$, the second child group performs a
length-$n_2$ transform over $j_2$.  Its logical output is
\begin{equation}
Y[k_1+k_2n_1,q]
=
\sum_{j_2=0}^{n_2-1}
Z[k_1,j_2,q]\,
\omega_{n_2}^{-k_2j_2}.
\label{eq:app-second-stage-transform}
\end{equation}
Indeed, substituting \eqref{app-first-stage-transform} and
\eqref{app-twiddled-intermediate} into
\eqref{app-second-stage-transform} gives
\begin{align}
Y[k_1+k_2n_1,q]
={}&
\sum_{j_2=0}^{n_2-1}
\sum_{j_1=0}^{n_1-1}
V[j_1n_2+j_2,q]\,
\omega_{n_1}^{-k_1j_1}
\omega_n^{-k_1j_2}
\omega_{n_2}^{-k_2j_2}
\nonumber\\
={}&
\sum_{j_2=0}^{n_2-1}
\sum_{j_1=0}^{n_1-1}
V[j_1n_2+j_2,q]\,
\omega_n^{-(k_1+k_2n_1)(j_1n_2+j_2)}.
\label{eq:app-cooley-tukey-factorization}
\end{align}
Since $j=j_1n_2+j_2$ ranges over $[n]$, this is precisely the
length-$n$ DFT output at frequency $k_1+k_2n_1$.

The second child has $\ell_2$ transform digits.  For each fixed
physical first-stage output coordinate $a$ and batch coordinate $q$,
its input sequence is, by \eqref{app-twiddled-layout},
\[
j_2
\longmapsto
Z[\pi_{\ell_1}(a),j_2,q].
\]
Applying the induction hypothesis to the second child, its physical
output coordinate $c\in[n_2]$ represents
\[
k_2=\pi_{\ell_2}(c).
\]
Consequently, the final physical output satisfies
\begin{equation}
Y^{\mathrm{phys}}[a,c,q]
=
Y\bigl[
\pi_{\ell_1}(a)
+
\pi_{\ell_2}(c)n_1,
q
\bigr].
\label{eq:app-two-field-output-layout}
\end{equation}

It remains to identify this frequency with the digit reversal of the
parent physical coordinate.  The parent physical transform coordinate
is
\[
r=an_2+c.
\]
Its digit string is the concatenation of the $\ell_1$ digits of $a$
followed by the $\ell_2$ digits of $c$.  Reversing all $\ell$ digits
therefore gives the digit string consisting of the reversed digits of
$c$ followed by the reversed digits of $a$.  Hence
\begin{equation}
\pi_\ell(an_2+c)
=
\pi_{\ell_1}(a)
+
\pi_{\ell_2}(c)n_1.
\label{eq:app-digit-reversal-composition}
\end{equation}
Combining \eqref{app-two-field-output-layout} and
\eqref{app-digit-reversal-composition} yields
\[
Y^{\mathrm{phys}}[r,q]
=
Y[\pi_\ell(r),q].
\]
This is the desired output-layout identity.
\end{proof}

\punt{
\begin{proof}
We proceed by induction on $\ell$. The constant-size leaf case is
immediate because the local DFT kernel uses natural-order coordinates
and the digit reversal on a constant-size leaf is incorporated into the
leaf interpretation.

For the inductive step, split the transform digit field as
\[
T=T_1\pplus T_2,
\qquad
|T_1|=\ell_1,
\qquad
|T_2|=\ell_2.
\]
Write
\[
j=j_1n_2+j_2,
\qquad
k=k_1+k_2n_1.
\]
Define the first-stage transform
\begin{equation}
H[k_1,j_2,q]
=
\sum_{j_1=0}^{n_1-1}
V[j_1n_2+j_2,q]\omega_{n_1}^{-k_1j_1},
\label{eq:app-first-stage-transform}
\end{equation}
and define
\begin{equation}
Z[k_1,j_2,q]
=
H[k_1,j_2,q]\omega_n^{-k_1j_2}.
\label{eq:app-first-stage-twiddle}
\end{equation}

By the induction hypothesis applied to the first child group, its
physical output coordinate $a\in[n_1]$ satisfies
\begin{equation}
H^{\mathrm{phys}}[a,j_2,q]
=
H[\pi_{\ell_1}(a),j_2,q].
\label{eq:app-first-child-output-layout}
\end{equation}
Thus, physical coordinate $a$ represents
\[
k_1=\pi_{\ell_1}(a).
\]
The layout-aware twiddle rule~\eqref{app-layout-aware-twiddle}
therefore gives
\begin{equation}
Z^{\mathrm{phys}}[a,j_2,q]
=
Z[\pi_{\ell_1}(a),j_2,q].
\label{eq:app-twiddled-first-child-layout}
\end{equation}

For every fixed first-stage frequency $k_1$ and batch instance $q$, the
second stage is the length-$n_2$ DFT
\begin{equation}
Y[k_1+k_2n_1,q]
=
\sum_{j_2=0}^{n_2-1}
Z[k_1,j_2,q]\omega_{n_2}^{-k_2j_2},
\qquad
k_2\in[n_2].
\label{eq:app-second-stage-transform}
\end{equation}
Indeed, substituting \eqref{app-first-stage-transform} and
\eqref{app-first-stage-twiddle} yields
\begin{align}
Y[k_1+k_2n_1,q]
={}&
\sum_{j_2=0}^{n_2-1}
\sum_{j_1=0}^{n_1-1}
V[j_1n_2+j_2,q]
\omega_{n_1}^{-k_1j_1}
\omega_n^{-k_1j_2}
\omega_{n_2}^{-k_2j_2}
\nonumber\\
={}&
\sum_{j_2=0}^{n_2-1}
\sum_{j_1=0}^{n_1-1}
V[j_1n_2+j_2,q]
\omega_n^{-(k_1+k_2n_1)(j_1n_2+j_2)}.
\label{eq:app-cooley-tukey-expansion}
\end{align}
Thus \eqref{app-second-stage-transform} is exactly the
two-factor Cooley--Tukey decomposition of the length-$n$ DFT.

For fixed $(a,q)$, the second child receives the physical sequence
\[
j_2\longmapsto
Z^{\mathrm{phys}}[a,j_2,q]
=
Z[\pi_{\ell_1}(a),j_2,q].
\]
By the induction hypothesis applied to the second child, its physical
output coordinate $c\in[n_2]$ represents
\[
k_2=\pi_{\ell_2}(c).
\]
Consequently,
\begin{align}
Y^{\mathrm{phys}}[a,c,q]
&=
Y\bigl[
\pi_{\ell_1}(a)
+
\pi_{\ell_2}(c)n_1,
q
\bigr].
\label{eq:app-two-field-output-layout}
\end{align}

\punt{
For fixed $(a,q)$, the second child performs a length-$n_2$ transform
over $j_2$. By the induction hypothesis applied to that child, physical
output coordinate $c\in[n_2]$ represents
\[
k_2=\pi_{\ell_2}(c).
\]
Hence,
\begin{align}
Y^{\mathrm{phys}}[a,c,q]
&=
Y[\pi_{\ell_2}(c),\pi_{\ell_1}(a),q].
\label{eq:app-two-field-output-layout}
\end{align}
The physical parent coordinate is
\[
r=an_2+c,
\]
whose digit string is
\[
a\pplus c.
\]
}

The represented logical frequency is
\[
k
=
\pi_{\ell_1}(a)
+
\pi_{\ell_2}(c)n_1.
\]
In digit notation, its most-to-least-significant digit string is
\[
\pi_{\ell_2}(c)\mathbin{+\!\!+}\pi_{\ell_1}(a).
\]
The physical parent coordinate is
\[
r=an_2+c,
\]
whose digit string is
\[
a\mathbin{+\!\!+}c.
\]
Hence,
\[
\pi_{\ell}
\left(
a\mathbin{+\!\!+}c
\right)
=
\pi_{\ell_2}(c)\mathbin{+\!\!+}\pi_{\ell_1}(a),
\]
and therefore
\[
Y^{\mathrm{phys}}[r,q]
=
Y[\pi_{\ell}(r),q].
\]
\punt{
The represented logical frequency has digit string
\[
\pi_{\ell_2}(c)\pplus \pi_{\ell_1}(a).
\]
}
By the definition of digit reversal,
\begin{equation}
\pi_{\ell}
\left(
a\pplus c
\right)
=
\pi_{\ell_2}(c)\pplus \pi_{\ell_1}(a).
\label{eq:app-digit-reversal-composition}
\end{equation}
Therefore,
\[
Y^{\mathrm{phys}}[r,q]
=
Y[\pi_{\ell}(r),q].
\]

\punt{
Finally, substituting~\eqref{app-first-stage-transform} and
\eqref{app-first-stage-twiddle} into the second-child transform gives
\begin{align}
Y[k_1+k_2n_1,q]
={}&
\sum_{j_2=0}^{n_2-1}
\left(
\sum_{j_1=0}^{n_1-1}
V[j_1n_2+j_2,q]\omega_{n_1}^{-k_1j_1}
\right)
\omega_n^{-k_1j_2}
\omega_{n_2}^{-k_2j_2},
\end{align}
which is the standard two-factor Cooley--Tukey identity. Thus the
logical result is the canonical length-$n$ DFT, stored under layout
$P_{\pi_{\ell}}$.
}
\end{proof}
}

\subsection{Input-Layout Local FFT Correctness}

The input-layout traversal is the conjugate interpretation of the
output-layout traversal.

\begin{theorem}[Input-layout local FFT identity]
\label{thm:app-input-layout}
For every $\ell\geq 0$, the input-layout local traversal implements
\begin{equation}
B_{\ell}
=
\Omega_{b^{\ell}}P_{\pi_{\ell}}.
\label{eq:app-input-layout-identity}
\end{equation}
In particular, if the physical input layout satisfies
\[
x^{\mathrm{phys}}=P_{\pi_{\ell}}x,
\]
then the input-layout traversal produces the natural-order DFT:
\[
B_{\ell}x^{\mathrm{phys}}
=
\Omega_{b^{\ell}}x.
\]
\end{theorem}

\begin{proof}
Let $A_{\ell}$ be the output-layout operator established in
Theorem~\ref{thm:app-output-layout}. The input-layout traversal uses
digit-reversed entrance and exit interpretations, so its operator is
\[
B_{\ell}
=
P_{\pi_{\ell}}A_{\ell}P_{\pi_{\ell}}.
\]
Using~\eqref{app-output-layout-identity} and
\eqref{app-digit-reversal-involution},
\[
B_{\ell}
=
P_{\pi_{\ell}}
\bigl(P_{\pi_{\ell}}\Omega_{b^{\ell}}\bigr)
P_{\pi_{\ell}}
=
\Omega_{b^{\ell}}P_{\pi_{\ell}}.
\]
Therefore,
\[
B_{\ell}x^{\mathrm{phys}}
=
\Omega_{b^{\ell}}P_{\pi_{\ell}}
P_{\pi_{\ell}}x
=
\Omega_{b^{\ell}}x.
\]
The two occurrences of $P_{\pi_{\ell}}$ in this argument are view
interpretations. They do not denote additional runtime physical
permutations.
\end{proof}

\subsection{Local Correctness Summary}

\begin{corollary}[Batched local FFT operators]
\label{cor:app-batched-local-operators}
For $s$ independent local transforms of length $b^{\ell}$, the
output-layout and input-layout traversals implement, respectively,
\[
A_{\ell,s}
=
I_s\otimes
\bigl(P_{\pi_{\ell}}\Omega_{b^{\ell}}\bigr)
\]
and
\[
B_{\ell,s}
=
I_s\otimes
\bigl(\Omega_{b^{\ell}}P_{\pi_{\ell}}\bigr).
\]
\end{corollary}

\begin{proof}
The recursive view operations preserve data disjointness among batch
instances by Theorem~\ref{thm:app-view-closure}. Applying
Theorems~\ref{thm:app-output-layout} and
\ref{thm:app-input-layout} independently to every batch instance gives
the stated Kronecker-product forms.
\end{proof}

In PACO Phase~I, each processor invokes the output-layout traversal on
its local length-$N_1$ transforms. Thus, its physical row coordinate
$a$ represents logical first-stage frequency $\pi_{m_1}(a)$. The fused
middle stage of Sect.~4.2 and \appref{local-fft} materializes the complete
digit reversal
\[
\pi_m
\left(
a\pplus c
\right)
=
\pi_{m_2}(c)\pplus \pi_{m_1}(a),
\]
thereby converting the Phase-I output layout into the input layout
consumed by Phase~III. Phase~III then invokes the input-layout
traversal on its local length-$N_2$ transforms and produces
natural-order DFT coefficients, as used in Theorem~3.
}

\section{Cache-Oblivious Batched Local FFT}\label{app:local-fft}

This appendix supplies the complete scheduler and the rigorous cache
analysis underlying \propositionref{local-batched-fft} of the main text. The scheduler
operates on the admissible digit-field views defined in
Appendix~\ref{app:layout}.

Consider an admissible view
\[
  V=(T,O,C,F,P,\mathrm{mode}).
\]
Let
\begin{equation}\label{eq:B-params}
  n=b^{|T|},\qquad s=b^{|O|+|C|},\qquad c=b^{|C|},
\end{equation}
where $n$ is the transform length, $s$ is the number of batched
transforms, and $c$ is the width of the distinguished contiguous batch
field. We call the product
\[
  \mu(V)=ns
\]
the \emph{mass} of the view; it is the number of items in its
working set. The scheduler is independent of the cache parameters
$M$ and $B$.

\subsection{Scheduler Pseudocode}

Figure~\ref{fig:cofft} realizes the scheduler described in the main
text: it invokes a constant-size kernel at the leaves, performs a batch
split when $s\ge n$ (splitting an outer-batch digit whenever one is
available, and the distinguished contiguous field only otherwise), and
performs a balanced digit-aligned transform split when $s<n$.

\begin{figure}[t]
\begin{codebox}
\Procname{$\proc{CO-Batched-FFT}(V)$}
\li $n \gets b^{|T|}$
\li $s \gets b^{|O|+|C|}$
\li \If $n \le b$ \textbf{ and } $s \le b$  \Comment base case
\li \Then invoke the constant-size natural-order DFT kernels
\li \Return
      \End
\li \If $s \ge n$  \Comment batch split
\li \Then
      \If $O \ne \emptyset$
\li   \Then choose an active outer-batch digit $d \in O$
\li   \Else 
\li        choose the most significant active digit $d \in C$
      \End
\li   \For $r \gets 0$ \To $b-1$
\li   \Do $\proc{CO-Batched-FFT}\bigl(\proc{Batch-Child}(V,d,r)\bigr)$
      \End
\li   \Return
      \End
\zi \Comment transform split
\li split $T = T_1 \pplus T_2$ with
      $\bigl\lvert\,\lvert T_1\rvert-\lvert T_2\rvert\,\bigr\rvert \le 1$
\li $V_{\mathrm{I}} \gets \proc{First-Transform-Child}(V,T_1,T_2)$
\li $\proc{CO-Batched-FFT}(V_{\mathrm{I}})$
\li relabel the completed physical $T_1$-field as output field $A$
\li $V_{\mathrm{II}} \gets \proc{Second-Transform-Child}(V,T_1,T_2,A)$
\li fuse layout-aware twiddles into the first accesses of $V_{\mathrm{II}}$
\li $\proc{CO-Batched-FFT}(V_{\mathrm{II}})$
\end{codebox}
\vspace{-1em}
\caption{Cache-oblivious batched local FFT scheduler. The operations
$\proc{Batch-Child}$, $\proc{First-Transform-Child}$, and
$\proc{Second-Transform-Child}$ are the admissible view operations
defined in Appendix~\ref{app:layout}.}
\label{fig:cofft}
\end{figure}

The transform split in Figure~\ref{fig:cofft} is only a recursive view
transformation: by the no-physical-transpose property of
Appendix~\ref{app:layout}, the first recursive call computes the
first-stage transforms, the physical $T_1$-slots are then relabeled as
the output field $A$, and the second call treats $A$ as an outer-batch
field. The layout-aware twiddle is evaluated from the logical
coordinates of the current view and fused into the second child group,
so the scheduler introduces no separate recursive twiddle pass or
physical transpose.

\subsection{Local Work Bound}

\begin{proposition}[Local work]\label{prop:B-work}
For $s$ transforms of length $n=b^{\ell}$, the scheduler performs
\begin{equation}\label{eq:B-work}
  W_{\mathrm{local}}(n,s)=O\!\bigl(ns(1+\log n)\bigr)
\end{equation}
local operations.
\end{proposition}

\begin{proof}
We separate the scheduler work into three parts: (i) constant-size leaf
DFT work, (ii) fused twiddle and transform-split control work, and
(iii) batch-split control work.

Fix one of the $s$ logical transforms at the root. Ignore batch splits
momentarily: they only partition the collection of independent
transforms and never alter the Cooley--Tukey recursion followed by an
individual transform. Whenever a transform subproblem of length
$n'=b^{d}$ is split into factors $n_1=b^{d_1}$ and $n_2=b^{d_2}$, where
$d_1+d_2=d$, the first child evaluates the $d_1$ radix-$b$ FFT levels
associated with the first factor and the second child evaluates the
$d_2$ levels associated with the second factor. Thus, over the complete
transform recursion, the transform dependencies of the fixed logical
transform contain exactly $\log_b n$ radix-$b$ levels, up to a constant
factor depending only on the constant-size leaf convention.

At each radix-$b$ level, every item participates in $O(1)$ fixed-radix
butterfly and twiddle operations. The twiddle of a transform split is
fused into the first access of the corresponding second child group,
but it is still charged to that same FFT level. Hence one logical
transform incurs $O(n\log_b n)$ leaf, butterfly, and twiddle work.
Summed over all $s$ transforms,
\begin{equation}\label{eq:B-wfft}
  W_{\mathrm{FFT}}(n,s)=O(ns\log_b n).
\end{equation}

It remains to bound scheduler overhead. A batch split performs only
$O(1)$ descriptor and loop-control work per child. Its $b$ children
partition the parent batch instances, and $b=O(1)$. Therefore, over all
batch-only recursion nodes associated with any fixed transform
subproblem, the total administrative work is linear in the number of
batch instances represented by that subproblem. Summing over the
transform recursion gives $O(ns)$ additional work. The constant-size
leaf calls also contribute $O(ns)$ work. Consequently,
\[
  W_{\mathrm{local}}(n,s)=O(ns\log_b n)+O(ns)=O\!\bigl(ns(1+\log n)\bigr).
\]
Finally, Appendix~\ref{app:layout} proves that transform splits are
view transformations rather than physical transposes; thus no
additional $\Theta(ns)$ data-rearrangement pass is incurred at each
recursive transform split.
\end{proof}

\subsection{Recursion Depth and Stack Space}

\begin{lemma}[Recursion depth]\label{lem:B-depth}
Along every root-to-leaf path, the scheduler performs
\begin{equation}\label{eq:B-tsplits}
  O\!\bigl(1+\log\log_b n\bigr)
\end{equation}
transform splits and
\begin{equation}\label{eq:B-bsplits}
  O\!\bigl(\log_b(ns)\bigr)
\end{equation}
batch splits. Hence its total recursion depth is $O(\log(ns))$.
\end{lemma}

\begin{proof}
Let $d=\log_b n$ be the number of active transform digits. A balanced
transform split leaves either $\lceil d/2\rceil$ or $\lfloor d/2\rfloor$
transform digits in a child. Repeated balanced splitting therefore
reaches a constant-size transform after $O(1+\log d)$ transform splits,
which gives \eqref{B-tsplits}.

Every batch split fixes one active batch digit. Initially there are
$\log_b s$ batch digits. A transform split may move transform digits
into an outer-batch field, but along one root-to-leaf path the total
number of such transferred digits is at most $\log_b n$. Thus no more
than $\log_b(ns)$ active batch digits can be fixed, proving
\eqref{B-bsplits}. Summing the two bounds gives the depth bound.
\end{proof}

\begin{corollary}[Stack space]\label{cor:B-stack}
The scheduler uses $S_{\mathrm{stack}}=O(\log(ns))$ words of
recursion-stack space.
\end{corollary}

\begin{proof}
Each stack frame stores only a constant number of field descriptors,
fixed-digit assignments, loop counters, and mode bits. The result
follows from Lemma~\ref{lem:B-depth}.
\end{proof}

\subsection{Analytical Cache-Fitting Frontier}

We now analyze cache complexity. The following frontier is used only in
the proof; the scheduler never tests $M$, $B$, or the cache-fitting
condition.

Fix a sufficiently small constant $\alpha>0$, and define
$M^{\star}=\alpha M$. A view is \emph{cache fitting} if its mass
satisfies
\begin{equation}\label{eq:B-fit}
  ns\le M^{\star}.
\end{equation}
A \emph{maximal} cache-fitting view is either the root, when the root is
cache fitting, or a cache-fitting view whose parent is not cache
fitting. Its purpose is to identify subproblems whose complete
descendant execution can remain resident in the ideal cache; the
scheduler continues recursively below the frontier, and the frontier
merely groups that descendant work for analysis.

We record the mass behavior of the two split types, used repeatedly
below. A batch split partitions the batched transforms into $b$
data-disjoint children, so the child masses sum to the parent mass $ns$.
A transform split $n=n_1n_2$ produces child groups with parameters
$(n_1,sn_2)$ and $(n_2,sn_1)$, each of mass
\begin{equation}\label{eq:B-massinv}
  n_1(sn_2)=n_2(sn_1)=ns .
\end{equation}
Thus a transform split preserves the mass of each child group.

\subsection{Frontier Geometry}

\begin{lemma}[Frontier entry is through a batch split]\label{lem:B-entry}
A non-cache-fitting view can first enter the cache-fitting region only
through a batch split.
\end{lemma}

\begin{proof}
A transform split preserves the mass of each child group, by
\eqref{B-massinv}, and therefore cannot turn a non-cache-fitting
mass $ns>M^{\star}$ into a cache-fitting one. A batch split reduces the
batch size of each child by a factor of $b$ and can therefore produce a
cache-fitting child.
\end{proof}

\begin{lemma}[Mass of a nonroot frontier view]\label{lem:B-mass}
Every nonroot maximal cache-fitting view satisfies
\begin{equation}\label{eq:B-massrange}
  \frac{M^{\star}}{b}<ns\le M^{\star}.
\end{equation}
\end{lemma}

\begin{proof}
By Lemma~\ref{lem:B-entry}, a nonroot maximal cache-fitting view is
created by a batch split. Its parent has the same transform length and
$b$ times as many batch instances. Since the child fits but the parent
does not, $ns\le M^{\star}$ and $b\,ns>M^{\star}$, which is
equivalent to \eqref{B-massrange}.
\end{proof}

\subsection{Contiguous Width at the Frontier}\label{app:B-width-at-frontier}

\begin{lemma}[Contiguous width at the frontier]\label{lem:B-width}
Suppose the root view has initial contiguous width $c_0\ge B$. Then
every maximal cache-fitting view has
\begin{equation}\label{eq:B-widthB}
  c=\Omega(B).
\end{equation}
\end{lemma}

\begin{proof}
If the distinguished field $C$ is never split on the root-to-frontier
path, then its width remains $c=c_0\ge B$.

Otherwise, consider the last split of $C$ before the frontier view, and
let $c$ denote the child width immediately after that split. The parent
width was $bc$. Since the scheduler splits $C$ only when $O=\emptyset$,
the parent batch size was also $bc$, and the batch-split condition gives
$n\le bc$ at that split. If the resulting child is already cache
fitting, its parent is not, so $M^{\star}<n\,bc\le b^2c^2$; if the child
is not yet cache fitting, then $M^{\star}<nc\le bc^2$. In either case
\begin{equation}\label{eq:B-widthM}
  c=\Omega(\sqrt{M^{\star}})=\Omega(\sqrt{M}),
\end{equation}
a bound stronger than \eqref{B-widthB}; only \eqref{B-widthB} is
used downstream. Subsequent outer-batch and transform splits do not
change $c$, and the tall-cache assumption $M=\Omega(B^2)$ turns
\eqref{B-widthM} into \eqref{B-widthB}.
\end{proof}

\subsection{Transform Length and Number of Transform Splits}\label{app:B-tb}

\begin{lemma}[Transform--batch relation]\label{lem:B-tb}
There is a constant $\kappa_b=O(1)$, depending only on the fixed radix
$b$, such that every view reached after at least one transform split
satisfies
\begin{equation}\label{eq:B-tb}
  s\le\kappa_b\,n^{3}.
\end{equation}
\end{lemma}

\begin{proof}
Suppose a transform split is applied to a view with $n=n_1n_2$ and
$s<n$. The balanced digit-aligned split satisfies
$\tfrac1b\le n_1/n_2\le b$. For the first child the transform length is
$n_1$ and the batch size is $sn_2$; since $s<n_1n_2$,
\[
  sn_2<n_1n_2^2\le b^2 n_1^3 .
\]
The second child is symmetric, $sn_1<n_2 n_1^2\le b^2 n_2^3$. Thus both
immediate children satisfy \eqref{B-tb} with $\kappa_b=b^2$. A later
batch split decreases $s$ without changing $n$, while a later transform
split re-establishes the same inequality for its children. The claim
follows by induction along the recursion path.
\end{proof}

\begin{lemma}[Transform length at the frontier]\label{lem:B-tlen}
Let a nonroot maximal cache-fitting view be reached along a path
containing at least one transform split. Then
\begin{equation}\label{eq:B-tlen}
  n=\Omega(M^{1/4}).
\end{equation}
\end{lemma}

\begin{proof}
By Lemma~\ref{lem:B-mass}, $M^{\star}/b<ns$, and by
Lemma~\ref{lem:B-tb}, $s\le\kappa_b n^3$. Hence
$M^{\star}/b<\kappa_b n^4$, which proves \eqref{B-tlen}.
\end{proof}

We index transform splits by their \emph{output} length: let $n_0=n$ be
the root transform length, and let $n_j$ denote the transform length of
the active child group immediately \emph{after} the $j$-th transform
split on a root-to-frontier path.

\begin{lemma}[Transform-size shrinkage]\label{lem:B-shrink}
For every $j\ge 1$,
\begin{equation}\label{eq:B-shrink1}
  n_j\le\sqrt{b\,n_{j-1}} .
\end{equation}
Hence, after $K$ transform splits,
\begin{equation}\label{eq:B-shrink2}
  n_K\le b^{\,1-1/2^{K}}\,n^{\,1/2^{K}}<b\,n^{\,1/2^{K}} .
\end{equation}
\end{lemma}

\begin{proof}
If the length before the $j$-th split is $n_{j-1}=b^{d}$, the larger
child contains at most $\lceil d/2\rceil$ transform digits, so
$n_j\le b^{\lceil d/2\rceil}\le b^{1/2}n_{j-1}^{1/2}$, which is
\eqref{B-shrink1}. Iterating from $n_0=n$ gives
$n_j\le b^{\,1-1/2^{j}}n^{\,1/2^{j}}$, and in particular
\eqref{B-shrink2}.
\end{proof}

\begin{lemma}[Number of transform splits before the frontier]\label{lem:B-nsplit}
If $K$ is the number of transform splits on a root-to-frontier path,
then
\begin{equation}\label{eq:B-nsplit}
  2^{K}=O\!\bigl(1+\log_M n\bigr).
\end{equation}
\end{lemma}

\begin{proof}
The claim is immediate when $K=0$, so assume $K\ge 1$. By
Lemma~\ref{lem:B-tlen} the frontier transform length satisfies
$n_K\ge c_b M^{1/4}$ for a constant $c_b>0$ depending only on $b$ and
$\alpha$. Combining this with \eqref{B-shrink2},
\[
  c_b M^{1/4}\le n_K<b\,n^{\,1/2^{K}} .
\]
Taking base-$2$ logarithms,
\[
  \frac{1}{2^{K}}\log_2 n>\tfrac14\log_2 M-O(1).
\]
For $M$ larger than a suitable constant the right-hand side is at least
$\tfrac18\log_2 M$, so
\[
  2^{K}<\frac{\log_2 n}{\tfrac18\log_2 M}=8\log_M n,
\]
i.e.\ $2^{K}=O(\log_M n)$. In the remaining regime $n=O(M)$, one has
$1+\log_M n=O(1)$, and only $O(1)$ balanced transform splits can occur
before the transform length becomes constant, so $2^{K}=O(1)$. The two
cases give \eqref{B-nsplit}.
\end{proof}

\subsection{Block Cover of a Frontier View}

\begin{lemma}[Block cover of one frontier view]\label{lem:B-cover}
A maximal cache-fitting view of mass $ns$ intersects
\begin{equation}\label{eq:B-cover}
  O\!\Bigl(1+\frac{ns}{B}\Bigr)
\end{equation}
cache blocks. Its complete depth-first descendant execution incurs the
same asymptotic number of cache misses.
\end{lemma}

\begin{proof}
By the active-run decomposition of Appendix~\ref{app:layout}, the view
is the disjoint union of $ns/c$ contiguous runs of length $c$.
Therefore its block footprint is at most
\[
  O\!\Bigl(\frac{ns}{c}\Bigl(1+\frac{c}{B}\Bigr)\Bigr)
  =O\!\Bigl(\frac{ns}{c}+\frac{ns}{B}\Bigr)
  =O\!\Bigl(\frac{ns}{B}\Bigr)
  =O\!\Bigl(1+\frac{ns}{B}\Bigr),
\]
where the third equality uses $c=\Omega(B)$ from
Lemma~\ref{lem:B-width}, which gives $ns/c=O(ns/B)$.

Every descendant of the frontier view accesses only items in the
frontier view's working set: batch splits select subsets of its runs,
contiguous-field splits subdivide those runs, and transform splits
reinterpret the same addresses through sequential child groups.
Moreover the view contains at most $M^{\star}=\alpha M$ items; choosing
$\alpha$ sufficiently small ensures that its block footprint, together
with constant recursion metadata, fits in cache. Under depth-first
execution these blocks can remain resident while the complete
descendant computation is performed. Hence \eqref{B-cover} also
bounds the descendant cache cost.
\end{proof}

\subsection{Total Frontier Mass}

\begin{lemma}[Total frontier mass]\label{lem:B-totalmass}
Suppose every root-to-frontier path contains at most $K$ transform
splits. Then the sum of the masses of all maximal cache-fitting views is
at most $ns\,2^{K}$.
\end{lemma}

\begin{proof}
We prove the more general statement that a recursion subtree rooted at a
view of mass $\mu$, with at most $k$ remaining transform splits on each
path, has total frontier mass at most $\mu\,2^{k}$.

If the root of the subtree is cache fitting, it is itself a frontier
view and contributes mass $\mu\le\mu\,2^{k}$.

If the root performs a batch split, the child masses sum to $\mu$;
applying the induction hypothesis to every child gives total frontier
mass at most $2^{k}\sum_{u}\mathrm{mass}(u)=\mu\,2^{k}$.

If the root performs a transform split, both child groups have mass
$\mu$ by \eqref{B-massinv}, and each has at most $k-1$ remaining
transform splits. Their combined frontier mass is therefore at most
$2\mu\,2^{k-1}=\mu\,2^{k}$.

Applying the claim at the original root proves the lemma.
\end{proof}

\subsection{Cache Complexity}

\begin{theorem}[Cache complexity of the batched local FFT]\label{thm:B-cache}
Let an admissible root view contain $s$ transforms of length $n$, with
initial contiguous width $c\ge B$. Under the tall-cache assumption
$M=\Omega(B^2)$, the scheduler incurs
\begin{equation}\label{eq:B-cache}
  Q_{\mathrm{local}}(n,s)=O\!\Bigl(\frac{ns}{B}\bigl(1+\log_M n\bigr)\Bigr)
\end{equation}
cache misses.
\end{theorem}

\begin{proof}
Let $K$ be the maximum number of transform splits on a root-to-frontier
path. Lemma~\ref{lem:B-nsplit} gives $2^{K}=O(1+\log_M n)$, so by
Lemma~\ref{lem:B-totalmass} the total mass of the maximal cache-fitting
frontier is
\begin{equation}\label{eq:B-frontiermass}
  O\!\bigl(ns(1+\log_M n)\bigr).
\end{equation}
By Lemma~\ref{lem:B-cover}, a frontier view of mass $\mu$ costs
$O(1+\mu/B)$ cache misses. The additive constant is absorbed into the
mass term: by Lemma~\ref{lem:B-width} every frontier view has
contiguous width $\Omega(B)$, hence mass $\mu=\Omega(B)$. Summing the
cache costs over the frontier and using \eqref{B-frontiermass} gives
\eqref{B-cache}.
\end{proof}

\subsection{Application to the Two PACO Local Phases}

In Phase~I, each processor performs $N_2/p$ transforms of length $N_1$;
its initial contiguous batch field is precisely the local input slab, so
\begin{equation}\label{eq:B-phaseI}
  n=N_1,\qquad s=c=\frac{N_2}{p}.
\end{equation}
The slab-width assumption $N_2/p\ge B$ permits direct application of
Theorem~\ref{thm:B-cache}, yielding
\begin{equation}\label{eq:B-QI}
  Q_{q,\mathrm{I}}=O\!\Bigl(\frac{N}{pB}\bigl(1+\log_M N_1\bigr)\Bigr).
\end{equation}
Similarly, in Phase~III each processor performs $N_1/p$ transforms of
length $N_2$, with
\begin{equation}\label{eq:B-phaseIII}
  n=N_2,\qquad s=c=\frac{N_1}{p}.
\end{equation}
Since $N_1/p\ge B$,
\begin{equation}\label{eq:B-QIII}
  Q_{q,\mathrm{III}}=O\!\Bigl(\frac{N}{pB}\bigl(1+\log_M N_2\bigr)\Bigr).
\end{equation}
Finally, $\log_M N_1+\log_M N_2=\log_M N$. Adding \eqref{B-QI} and
\eqref{B-QIII}, and absorbing the linear cache cost of the fused
middle stage, gives
\begin{equation}\label{eq:B-Qtotal}
  Q_{q,\mathrm{I}}+Q_{q,\mathrm{III}}
   =O\!\Bigl(\frac{N}{pB}\bigl(1+\log_M N\bigr)\Bigr).
\end{equation}

\punt{
\section{Cache-Oblivious Batched Local FFT}
\label{app:local-fft}

This appendix supplies the complete scheduler and the rigorous cache
analysis underlying \propositionref{local-batched-fft} of the main text.  The scheduler operates on the
admissible digit-field views defined in \appref{layout}.  

Consider an admissible view
\[
\mathcal{V}=(T,O,C,F,P,\mathit{mode}).
\]
Let
\begin{equation}
n=b^{|T|},
\qquad
s=b^{|O|+|C|},
\qquad
c=b^{|C|},
\label{eq:local-parameters}
\end{equation}
where \(n\) is the transform length, \(s\) is the number of batched
transforms, and \(c\) is the width of the distinguished contiguous
batch field.  The scheduler is independent of the cache parameters
\(M\) and \(B\).

\subsection{Scheduler Pseudocode}

Figure~\ref{fig:co-batched-fft} realizes the scheduler described in the
main text: it invokes a constant-size kernel at the leaves, performs a
batch split when \(s\ge n\) (splitting an outer-batch digit whenever one
is available, and the distinguished contiguous field only otherwise),
and performs a balanced digit-aligned transform split when \(s<n\).

\begin{figure}[t]
\centering
\begin{codebox}
\Procname{$\proc{CO-Batched-FFT}(\mathcal{V})$}
\li $n \gets b^{|T|}$
\li $s \gets b^{|O|+|C|}$
\li \If $n \le b$ and $s \le b$ \Comment{Base case}
\li \Then invoke the constant-size natural-order DFT kernels
\li \Return
\End
\li \If $s \ge n$ \Comment{Batch Split}
\li \Then
    \If $O \neq \varnothing$
    \li \Then choose an active outer-batch digit $d \in O$
    \li \Else 
    \li choose the most significant active digit $d \in C$
    \End
\li \For $r \gets 0$ \To $b-1$
\li \Do $\proc{CO-Batched-FFT}(\proc{Batch-Child}(\mathcal{V},d,r))$
\End
\li \Return
\End
\zi \Comment{Transform Split}
\li split $T=T_1\mathbin{+\!\!+}T_2$ with
$\bigl||T_1|-|T_2|\bigr|\le 1$
\li $\mathcal{V}_{\mathrm{I}}
\gets \proc{First-Transform-Child}(\mathcal{V},T_1,T_2)$
\li $\proc{CO-Batched-FFT}(\mathcal{V}_{\mathrm{I}})$
\li relabel the completed physical $T_1$-field as output field $A$
\li $\mathcal{V}_{\mathrm{II}}
\gets \proc{Second-Transform-Child}(\mathcal{V},T_1,T_2,A)$
\li fuse layout-aware twiddles into the first accesses of
$\mathcal{V}_{\mathrm{II}}$
\li $\proc{CO-Batched-FFT}(\mathcal{V}_{\mathrm{II}})$
\end{codebox}
\vspace{-1em}
\caption{Cache-oblivious batched local FFT scheduler.  The operations
\(\mathrm{BatchChild}\), \(\mathrm{FirstTransformChild}\), and
\(\mathrm{SecondTransformChild}\) are the admissible view operations
defined \appref{layout}.}
\label{fig:co-batched-fft}
\end{figure}

The transform split in Figure~\ref{fig:co-batched-fft} is only a
recursive view transformation: by the no-physical-transpose property of
\appref{layout}, the first recursive call computes the first-stage
transforms, the physical \(T_1\)-slots are then relabeled as the output
field \(A\), and the second call treats \(A\) as an outer-batch field.
The layout-aware twiddle is evaluated from the logical coordinates of
the current view and fused into the second child group, so the
scheduler introduces no separate recursive twiddle pass or physical
transpose.

\subsection{Local Work Bound}

\begin{proposition}[Local work]
\label{prop:local-work}
For \(s\) transforms of length \(n=b^\ell\), the scheduler performs
\begin{equation}
W_{\mathrm{local}}(n,s)
=
O\bigl(ns(1+\log n)\bigr)
\label{eq:local-work-bound}
\end{equation}
local operations.
\end{proposition}
\begin{proof}
We separate the scheduler work into three parts:
(i) constant-size leaf DFT work,
(ii) fused twiddle and transform-split control work, and
(iii) batch-split control work.

Fix one of the \(s\) logical transforms at the root.  Ignore batch
splits momentarily: they only partition the collection of independent
transforms and never alter the Cooley--Tukey recursion followed by an
individual transform.  Whenever a transform subproblem of length
\(n'=b^d\) is split into factors \(n_1=b^{d_1}\) and
\(n_2=b^{d_2}\), where \(d_1+d_2=d\), the first child evaluates the
\(d_1\) radix-\(b\) FFT levels associated with the first factor and the
second child evaluates the \(d_2\) levels associated with the second
factor.  Thus, over the complete transform recursion, the transform
dependencies of the fixed logical transform contain exactly
\(\log_b n\) radix-\(b\) levels, up to a constant factor depending only
on the constant-size leaf convention.

At each radix-\(b\) level, every item participates in \(O(1)\)
fixed-radix butterfly and twiddle operations.  The twiddle of a
transform split is fused into the first access of the corresponding
second child group, but it is still charged to that same FFT level.
Hence, one logical transform incurs
\[
O(n\log_b n)
\]
leaf, butterfly, and twiddle work.  Summed over all \(s\) transforms,
this gives
\begin{equation}
W_{\mathrm{FFT}}(n,s)=O(ns\log_b n).
\label{eq:fft-work-level-bound}
\end{equation}

It remains to bound scheduler overhead.  A batch split performs only
\(O(1)\) descriptor and loop-control work per child.  Its \(b\) children
partition the parent batch instances, and \(b=O(1)\).  Therefore, over
all batch-only recursion nodes associated with any fixed transform
subproblem, the total administrative work is linear in the number of
batch instances represented by that subproblem.  Summing over the
transform recursion gives \(O(ns)\) additional work.  The constant-size
leaf calls also contribute \(O(ns)\) work.

Consequently,
\[
W_{\mathrm{local}}(n,s)
=
O(ns\log_b n)+O(ns)
=
O\bigl(ns(1+\log n)\bigr).
\]

Finally, \appref{layout} proves that transform splits are view
transformations rather than physical transposes.  Thus, no additional
\(\Theta(ns)\) data-rearrangement pass is incurred at each recursive
transform split.
\end{proof}

\subsection{Recursion Depth and Stack Space}

\begin{lemma}[Recursion depth]
\label{lem:local-recursion-depth}
Along every root-to-leaf path, the scheduler performs
\begin{equation}
O\bigl(1+\log\log_b n\bigr)
\label{eq:transform-split-depth}
\end{equation}
transform splits and
\begin{equation}
O\bigl(\log_b(ns)\bigr)
\label{eq:batch-split-depth}
\end{equation}
batch splits.  Hence, its total recursion depth is
\begin{equation}
O\bigl(\log(ns)\bigr).
\label{eq:total-recursion-depth}
\end{equation}
\end{lemma}

\begin{proof}
Let \(d=\log_b n\) be the number of active transform digits.  A balanced
transform split leaves either \(\lceil d/2\rceil\) or
\(\lfloor d/2\rfloor\) transform digits in a child.  Repeated balanced
splitting therefore reaches a constant-size transform after
\(O(1+\log d)\) transform splits, which gives
\eqref{transform-split-depth}.

Every batch split fixes one active batch digit.  Initially there are
\(\log_b s\) batch digits.  A transform split may move transform digits
into an outer-batch field, but along one root-to-leaf path the total
number of such transferred digits is at most \(\log_b n\).  Thus, no
more than \(\log_b(ns)\) active batch digits can be fixed, proving
\eqref{batch-split-depth}.  Summing the two bounds gives
\eqref{total-recursion-depth}.
\end{proof}

\begin{corollary}[Stack space]
\label{cor:local-stack-space}
The scheduler uses
\begin{equation}
S_{\mathrm{stack}}=O\bigl(\log(ns)\bigr)
\label{eq:local-stack-space}
\end{equation}
words of recursion-stack space.
\end{corollary}

\begin{proof}
Each stack frame stores only a constant number of field descriptors,
fixed-digit assignments, loop counters, and mode bits.  The result
follows directly from Lemma~\ref{lem:local-recursion-depth}.
\end{proof}

\subsection{Analytical Cache-Fitting Frontier}

We now analyze cache complexity.  The following frontier is used only
in the proof; the scheduler never tests \(M\), \(B\), or the
cache-fitting condition.

Fix a sufficiently small constant \(\alpha>0\), and define
\begin{equation}
M^\star=\alpha M.
\label{eq:cache-fitting-capacity}
\end{equation}
A view with parameters \((n,s,c)\) is \emph{cache fitting} if
\begin{equation}
ns\le M^\star.
\label{eq:cache-fitting-condition}
\end{equation}
A maximal cache-fitting view is either the root, when the root is cache
fitting, or a cache-fitting view whose parent is not cache fitting.

The purpose of this frontier is to identify subproblems whose complete
descendant execution can remain resident in the ideal cache.  The
scheduler continues recursively below the frontier; the frontier merely
groups that descendant work for analysis.

We record the working-set mass behavior of the two split types, used
repeatedly below.  A batch split partitions the batched transforms into
\(b\) data-disjoint children, so the child masses sum to the parent mass
\(ns\).  A transform split \(n=n_1n_2\) produces child groups with
parameters \((n_1,sn_2)\) and \((n_2,sn_1)\), each of working-set mass
\begin{equation}
n_1(sn_2)=n_2(sn_1)=ns.
\label{eq:transform-child-mass}
\end{equation}
Thus a transform split preserves the working-set mass of each child
group.

\subsection{Frontier Geometry}

\begin{lemma}[Frontier entry is through a batch split]
\label{lem:frontier-entry-batch}
A non-cache-fitting view can first enter the cache-fitting region only
through a batch split.
\end{lemma}

\begin{proof}
A transform split preserves the working-set mass of each child group,
as shown in \eqref{transform-child-mass}.  It therefore cannot change
a non-cache-fitting mass \(ns>M^\star\) into a cache-fitting mass.
Conversely, a batch split reduces the batch size of each child by a
factor of \(b\), and can therefore produce a cache-fitting child.
\end{proof}

\begin{lemma}[Volume of a nonroot frontier view]
\label{lem:frontier-volume}
Every nonroot maximal cache-fitting view satisfies
\begin{equation}
\frac{M^\star}{b}<ns\le M^\star.
\label{eq:frontier-volume}
\end{equation}
\end{lemma}

\begin{proof}
By Lemma~\ref{lem:frontier-entry-batch}, a nonroot maximal
cache-fitting view is created by a batch split.  Its parent has the
same transform length and \(b\) times as many batch instances.  Since
the child fits but the parent does not,
\[
ns\le M^\star
\qquad\text{and}\qquad
bn s>M^\star.
\]
These inequalities are equivalent to \eqref{frontier-volume}.
\end{proof}

\subsection{Contiguous Width at the Frontier}

\begin{lemma}[Contiguous width at the frontier]
\label{lem:frontier-contiguous-width}
Suppose the root view has initial contiguous width \(c_0\ge B\).  Then
every maximal cache-fitting view has
\begin{equation}
c=\Omega(B).
\label{eq:frontier-contiguous-width}
\end{equation}
Moreover, if the root-to-frontier path contains a split of the
distinguished field \(C\), then
\begin{equation}
c=\Omega(\sqrt{M}).
\label{eq:frontier-contiguous-width-strong}
\end{equation}
\end{lemma}

\begin{proof}
If the distinguished field \(C\) is never split on the path, then its
width remains \(c=c_0\ge B\).

Otherwise, consider the last split of \(C\) before the frontier view,
and let \(c\) denote the child width immediately after that split.  The
parent width was \(bc\).  Since the scheduler splits \(C\) only when
\(O=\varnothing\), the parent batch size was also \(bc\).  The
batch-split condition therefore gives
\begin{equation}
n\le bc
\label{eq:contiguous-split-condition}
\end{equation}
at that split.

If the resulting child is already cache fitting, its parent is not.
Using \eqref{contiguous-split-condition},
\begin{equation}
M^\star<nbc\le b^2c^2.
\label{eq:contiguous-width-fitting-child}
\end{equation}
If the child is not yet cache fitting, then
\begin{equation}
M^\star<nc\le bc^2.
\label{eq:contiguous-width-nonfitting-child}
\end{equation}
Both cases imply \(c=\Omega(\sqrt{M^\star})=\Omega(\sqrt{M})\).
Subsequent outer-batch and transform splits do not change \(c\).  This
proves \eqref{frontier-contiguous-width-strong}.  The tall-cache
assumption \(M=\Omega(B^2)\) then implies
\eqref{frontier-contiguous-width}.
\end{proof}

\subsection{Reachable States after Transform Splitting}

\begin{lemma}[Transform--batch relation]
\label{lem:reachable-state}
There is a constant \(\kappa_b=O(1)\), depending only on the fixed
radix \(b\), such that every view reached after at least one transform
split satisfies
\begin{equation}
s\le \kappa_b n^3.
\label{eq:reachable-state}
\end{equation}
\end{lemma}

\begin{proof}
Suppose a transform split is applied to a view with
\(n=n_1n_2\) and \(s<n\).  The balanced digit-aligned split satisfies
\[
\frac{1}{b}\le \frac{n_1}{n_2}\le b.
\]

For the first child, the transform length is \(n_1\) and the batch size
is \(sn_2\).  Since \(s<n_1n_2\),
\[
sn_2<n_1n_2^2\le b^2n_1^3.
\]
The second child is symmetric:
\[
sn_1<n_2n_1^2\le b^2n_2^3.
\]
Thus, both immediate children satisfy \eqref{reachable-state} with
\(\kappa_b=b^2\).  A later batch split decreases \(s\) without changing
\(n\), while a later transform split establishes the same inequality
for its children.  The claim follows by induction along the recursion
path.
\end{proof}

\begin{lemma}[Transform length at the frontier]
\label{lem:frontier-transform-length}
Let a nonroot maximal cache-fitting view be reached along a path
containing at least one transform split.  Then
\begin{equation}
n=\Omega(M^{1/4}).
\label{eq:frontier-transform-length}
\end{equation}
\end{lemma}

\begin{proof}
By Lemma~\ref{lem:frontier-volume},
\[
\frac{M^\star}{b}<ns.
\]
By Lemma~\ref{lem:reachable-state},
\[
s\le \kappa_b n^3.
\]
Consequently,
\[
\frac{M^\star}{b}<\kappa_b n^4,
\]
which proves \eqref{frontier-transform-length}.
\end{proof}

\subsection{Transform Splits before the Frontier}

\begin{lemma}[Transform-size shrinkage]
\label{lem:transform-size-shrinkage}
Let \(n_j\) be the transform length after the \(j\)-th transform split
on a root-to-frontier path.  The larger child satisfies
\begin{equation}
n_{j+1}\le \sqrt{b n_j}.
\label{eq:transform-size-shrinkage}
\end{equation}
Hence, after \(K\) transform splits,
\begin{equation}
n_{K}
\le
b^{1-1/2^K}n^{1/2^K}
<
b n^{1/2^K}.
\label{eq:transform-size-after-k-splits}
\end{equation}
\end{lemma}

\begin{proof}
If \(n_j=b^{d_j}\), the larger child contains at most
\(\lceil d_j/2\rceil\) transform digits.  Therefore,
\[
n_{j+1}
\le
b^{\lceil d_j/2\rceil}
\le
b^{1/2}n_j^{1/2},
\]
which proves \eqref{transform-size-shrinkage}.  Repeated substitution
gives \eqref{transform-size-after-k-splits}.
\end{proof}

\begin{lemma}[Number of transform splits before the frontier]
\label{lem:frontier-transform-splits}
If \(K\) is the number of transform splits on a root-to-frontier path,
then
\begin{equation}
2^K=O\bigl(1+\log_M n\bigr).
\label{eq:frontier-transform-splits}
\end{equation}
\end{lemma}

\begin{proof}
The claim is immediate when \(K=0\).  Suppose \(K\ge1\).  By
Lemma~\ref{lem:frontier-transform-length}, the frontier transform
length satisfies \(n_K=\Omega(M^{1/4})\).  Combining this with
\eqref{transform-size-after-k-splits} gives
\[
\Omega(M^{1/4})
\le
b n^{1/2^K}.
\]
Taking logarithms yields
\[
\frac{\log n}{2^K}
=
\Omega(\log M),
\]
whenever \(n\) is asymptotically larger than \(M\).  Thus,
\[
2^K=O(\log_M n)
\]
in that case.

If \(n=O(M)\), only \(O(1)\) balanced transform splits can occur before
the transform length becomes constant or the path reaches the
cache-fitting frontier.  Since \(1+\log_M n=O(1)\) in this case, the
same asymptotic bound holds.  Combining the two cases proves
\eqref{frontier-transform-splits}.
\end{proof}

\subsection{Block Cover of a Frontier View}

\begin{lemma}[Block cover of one frontier view]
\label{lem:frontier-block-cover}
A maximal cache-fitting view of mass \(ns\) intersects
\begin{equation}
O\left(1+\frac{ns}{B}\right)
\label{eq:frontier-block-cover}
\end{equation}
cache blocks.  Its complete depth-first descendant execution incurs
the same asymptotic number of cache misses.
\end{lemma}

\begin{proof}
By the active-run decomposition proved in \appref{layout}, the view is the
disjoint union of \(ns/c\) contiguous runs of length \(c\).  Therefore,
its block footprint is at most
\begin{align}
O\left(
\frac{ns}{c}
\left(1+\frac{c}{B}\right)
\right)
&=
O\left(
\frac{ns}{c}+\frac{ns}{B}
\right)
\nonumber\\
&=
O\left(1+\frac{ns}{B}\right),
\label{eq:frontier-block-count}
\end{align}
where the last step uses Lemma~\ref{lem:frontier-contiguous-width}.

Every descendant of the frontier view accesses only items in the
frontier view's working set: batch splits select subsets of its runs,
contiguous-field splits subdivide those runs, and transform splits
reinterpret the same addresses through sequential child groups.
Moreover, the view contains at most \(M^\star=\alpha M\) items.
Choosing \(\alpha\) sufficiently small ensures that its block footprint,
together with constant recursion metadata, fits in cache.  Under
depth-first execution, these blocks can remain resident while the
complete descendant computation is performed.  Hence,
\eqref{frontier-block-count} also bounds the descendant cache cost.
\end{proof}

\subsection{Total Frontier Mass}

\begin{lemma}[Total frontier mass]
\label{lem:frontier-total-mass}
Suppose every root-to-frontier path contains at most \(K\) transform
splits.  Then the sum of the masses of all maximal cache-fitting views
is at most
\begin{equation}
ns\,2^K.
\label{eq:frontier-total-mass}
\end{equation}
\end{lemma}

\begin{proof}
We prove the more general statement that a recursion subtree rooted at
a view of mass \(\mu\), with at most \(k\) remaining transform splits
on each path, has total frontier mass at most \(\mu 2^k\).

If the root of the subtree is cache fitting, it is itself a frontier
view and contributes mass \(\mu\), which is at most \(\mu 2^k\).

If the root performs a batch split, the child masses sum to \(\mu\).
Applying the induction hypothesis to every child gives total frontier
mass at most
\[
2^k\sum_{\text{batch children }u}\mathrm{mass}(u)
=
\mu 2^k.
\]

If the root performs a transform split, both child groups have mass
\(\mu\) by \eqref{transform-child-mass}, and each has at most \(k-1\)
remaining transform splits.  Their combined frontier mass is therefore
at most
\[
2\mu 2^{k-1}
=
\mu 2^k.
\]
Applying the claim at the original root proves
\eqref{frontier-total-mass}.
\end{proof}

\subsection{Cache Complexity}

\begin{theorem}[Cache complexity of the batched local FFT]
\label{thm:local-cache-complexity}
Let an admissible root view contain \(s\) transforms of length \(n\),
with initial contiguous width \(c\ge B\).  Under the tall-cache
assumption \(M=\Omega(B^2)\), the scheduler incurs
\begin{equation}
Q_{\mathrm{local}}(n,s,c)
=
O\left(
\frac{ns}{B}
\bigl(1+\log_M n\bigr)
\right)
\label{eq:local-cache-complexity}
\end{equation}
cache misses.
\end{theorem}

\begin{proof}
Let \(K\) be the maximum number of transform splits on a root-to-frontier
path.  Lemma~\ref{lem:frontier-transform-splits} gives
\[
2^K=O(1+\log_M n).
\]
By Lemma~\ref{lem:frontier-total-mass}, the total mass of the maximal
cache-fitting frontier is therefore
\begin{equation}
O\bigl(ns(1+\log_M n)\bigr).
\label{eq:total-frontier-mass-bound}
\end{equation}

By Lemma~\ref{lem:frontier-block-cover}, each frontier view of mass
\(\mu\) costs \(O(1+\mu/B)\) cache misses.  The additive constant is
absorbed in the mass term: by
Lemma~\ref{lem:frontier-contiguous-width}, every frontier view has
contiguous width \(\Omega(B)\), and hence mass \(\mu=\Omega(B)\).
Summing the cache costs over the frontier and using
\eqref{total-frontier-mass-bound} gives
\[
Q_{\mathrm{local}}(n,s,c)
=
O\left(
\frac{ns}{B}
\bigl(1+\log_M n\bigr)
\right).
\]
This proves \eqref{local-cache-complexity}.
\end{proof}

\subsection{Application to the Two PACO Local Phases}

In Phase~I, each processor performs \(N_2/p\) transforms of length
\(N_1\).  Its initial contiguous batch field is precisely the local
input slab, so
\begin{equation}
n=N_1,
\qquad
s=c=\frac{N_2}{p}.
\label{eq:phase-i-local-parameters}
\end{equation}
The slab-width assumption \(N_2/p\ge B\) permits direct application of
Theorem~\ref{thm:local-cache-complexity}, yielding
\begin{equation}
Q_{q,\mathrm{I}}
=
O\left(
\frac{N}{pB}
\bigl(1+\log_M N_1\bigr)
\right).
\label{eq:phase-i-cache-complexity}
\end{equation}

Similarly, in Phase~III, each processor performs \(N_1/p\) transforms
of length \(N_2\), with
\begin{equation}
n=N_2,
\qquad
s=c=\frac{N_1}{p}.
\label{eq:phase-iii-local-parameters}
\end{equation}
Since \(N_1/p\ge B\),
\begin{equation}
Q_{q,\mathrm{III}}
=
O\left(
\frac{N}{pB}
\bigl(1+\log_M N_2\bigr)
\right).
\label{eq:phase-iii-cache-complexity}
\end{equation}

Finally,
\[
\log_M N_1+\log_M N_2=\log_M N.
\]
Adding \eqref{phase-i-cache-complexity} and
\eqref{phase-iii-cache-complexity}, and absorbing the linear cache
cost of the fused middle stage, gives
\begin{equation}
Q_{q,\mathrm{I}}+Q_{q,\mathrm{III}}
=
O\left(
\frac{N}{pB}
\bigl(1+\log_M N\bigr)
\right).
\label{eq:two-local-phases-cache-complexity}
\end{equation}
}

\punt{
\section{Cache-Oblivious Batched Local FFT}
\label{app:local-fft}

This appendix gives the complete scheduler and cache analysis for the
local PACO FFT phases.  The scheduler operates on the admissible
digit-field views defined \appref{layout}.  Its view operations are
therefore understood to inherit the address-map, closure, and
no-physical-transpose properties established there.  In particular,
this appendix analyzes the scheduling and locality consequences of
those operations, rather than repeating the digit-field algebra.

Consider an admissible view
\[
\mathcal{V}=(T,O,C,F,P,\mathit{mode}).
\]
Let
\begin{equation}
n=b^{|T|},
\qquad
s=b^{|O|+|C|},
\qquad
c=b^{|C|},
\label{eq:local-parameters}
\end{equation}
where \(n\) is the transform length, \(s\) is the number of batched
transforms, and \(c\) is the width of the distinguished contiguous
batch field.  The scheduler is independent of the cache parameters
\(M\) and \(B\).

\subsection{Scheduler Interface and Recursion}

The scheduler follows three rules.

\begin{enumerate}
    \item If both the transform length and batch size are constant,
    it invokes the appropriate constant-size natural-order DFT
    kernels.

    \item If \(s\ge n\), it performs a batch split.  It first selects
    an active digit from the outer-batch fields \(O\).  Only when
    \(O=\varnothing\) does it split the most significant active digit
    of the distinguished contiguous field \(C\).

    \item If \(s<n\), it performs a balanced digit-aligned
    Cooley--Tukey transform split.
\end{enumerate}

The outer-batch-first rule delays splitting \(C\) for as long as
possible.  This preserves spatial locality in the distinguished
contiguous field while still allowing the recursion to reduce the
number of transforms in a subproblem.

\begin{figure}[t]
\centering
\begin{codebox}
\Procname{$\proc{CO-Batched-FFT}(\mathcal{V})$}
\li $n \gets b^{|T|}$
\li $s \gets b^{|O|+|C|}$
\li \If $n \le b$ and $s \le b$ \Comment{Base case}
\li \Then invoke the constant-size natural-order DFT kernels \punt{prescribed
by $\mathit{mode}$}
\li \Return
\End
\li \If $s \ge n$ \Comment{Batch Split}
\li \Then
    \If $O \neq \varnothing$
    \li \Then choose an active outer-batch digit $d \in O$
    \li \Else 
    \li choose the most significant active digit $d \in C$
    \End
\li \For $r \gets 0$ \To $b-1$
\li \Do $\proc{CO-Batched-FFT}(\proc{Batch-Child}(\mathcal{V},d,r))$
\End
\li \Return
\End
\zi \Comment{Transform Split}
\li split $T=T_1\mathbin{+\!\!+}T_2$ with
$\bigl||T_1|-|T_2|\bigr|\le 1$
\li $\mathcal{V}_{\mathrm{I}}
\gets \proc{First-Transform-Child}(\mathcal{V},T_1,T_2)$
\li $\proc{CO-Batched-FFT}(\mathcal{V}_{\mathrm{I}})$
\li relabel the completed physical $T_1$-field as output field $A$
\li $\mathcal{V}_{\mathrm{II}}
\gets \proc{Second-Transform-Child}(\mathcal{V},T_1,T_2,A)$
\li fuse layout-aware twiddles into the first accesses of
$\mathcal{V}_{\mathrm{II}}$
\li $\proc{CO-Batched-FFT}(\mathcal{V}_{\mathrm{II}})$
\end{codebox}
\vspace{-1em}
\caption{Cache-oblivious batched local FFT scheduler.  The operations
\(\mathrm{BatchChild}\), \(\mathrm{FirstTransformChild}\), and
\(\mathrm{SecondTransformChild}\) are the admissible view operations
defined in \appref{layout}.}
\label{fig:co-batched-fft}
\end{figure}

The transform split in Figure~\ref{fig:co-batched-fft} is only a
recursive view transformation.  It does not itself execute a DFT or
move data.  The first recursive call computes the first-stage
transforms.  Only after that call returns are the physical slots of
\(T_1\) relabeled as the first-stage output field \(A\); the second
recursive call then treats \(A\) as an outer-batch field.  The
layout-aware twiddle factor is evaluated from the logical coordinates
represented by the current view and is fused into the second child
group.  Thus, the scheduler introduces no separate recursive twiddle
pass or recursive physical transpose.

\subsection{Local Work Bound}

\begin{proposition}[Local work]
\label{prop:local-work}
For \(s\) transforms of length \(n=b^\ell\), the scheduler performs
\begin{equation}
W_{\mathrm{local}}(n,s)
=
O\bigl(ns(1+\log n)\bigr)
\label{eq:local-work-bound}
\end{equation}
local operations.
\end{proposition}
\begin{proof}
We separate the scheduler work into three parts:
(i) constant-size leaf DFT work,
(ii) fused twiddle and transform-split control work, and
(iii) batch-split control work.

Fix one of the \(s\) logical transforms at the root.  Ignore batch
splits momentarily: they only partition the collection of independent
transforms and never alter the Cooley--Tukey recursion followed by an
individual transform.  Whenever a transform subproblem of length
\(n'=b^d\) is split into factors \(n_1=b^{d_1}\) and
\(n_2=b^{d_2}\), where \(d_1+d_2=d\), the first child evaluates the
\(d_1\) radix-\(b\) FFT levels associated with the first factor and the
second child evaluates the \(d_2\) levels associated with the second
factor.  Thus, over the complete transform recursion, the transform
dependencies of the fixed logical transform contain exactly
\(\log_b n\) radix-\(b\) levels, up to a constant factor depending only
on the constant-size leaf convention.

At each radix-\(b\) level, every item participates in \(O(1)\)
fixed-radix butterfly and twiddle operations.  The twiddle of a
transform split is fused into the first access of the corresponding
second child group, but it is still charged to that same FFT level.
Hence, one logical transform incurs
\[
O(n\log_b n)
\]
leaf, butterfly, and twiddle work.  Summed over all \(s\) transforms,
this gives
\begin{equation}
W_{\mathrm{FFT}}(n,s)=O(ns\log_b n).
\label{eq:fft-work-level-bound}
\end{equation}

It remains to bound scheduler overhead.  A batch split performs only
\(O(1)\) descriptor and loop-control work per child.  Its \(b\) children
partition the parent batch instances, and \(b=O(1)\).  Therefore, over
all batch-only recursion nodes associated with any fixed transform
subproblem, the total administrative work is linear in the number of
batch instances represented by that subproblem.  Summing over the
transform recursion gives \(O(ns)\) additional work.  The constant-size
leaf calls also contribute \(O(ns)\) work.

Consequently,
\[
W_{\mathrm{local}}(n,s)
=
O(ns\log_b n)+O(ns)
=
O\bigl(ns(1+\log n)\bigr).
\]

Finally, Appendix~A proves that transform splits are view
transformations rather than physical transposes.  Thus, no additional
\(\Theta(ns)\) data-rearrangement pass is incurred at each recursive
transform split.
\end{proof}
\punt{
\begin{proof}
A batch split partitions the \(s\) transforms into \(b\) data-disjoint
subbatches.  Hence, the sum of the working-set masses \(ns\) over its
children equals the parent mass.

Now consider a transform split
\[
n=n_1n_2.
\]
Its first and second child groups have parameters
\((n_1,sn_2)\) and \((n_2,sn_1)\), respectively.  Therefore,
\begin{equation}
n_1(sn_2)=n_2(sn_1)=ns.
\label{eq:transform-child-mass}
\end{equation}
Each transform split consequently creates two child groups, each of
the same mass as the parent.  Its fused twiddle multiplications require
\(O(ns)\) local operations, and the leaves perform only constant work
per item because \(b\) is fixed.

Along each transform dependency path, an item participates in
\(O(\log n)\) fixed-radix FFT levels.  Batch splits only distribute
independent transforms among recursive calls and do not introduce
additional FFT levels.  Equivalently, the total mass expansion caused
by transform recursion is \(O(\log n)\); this is proved more explicitly
below in Lemma~\ref{lem:frontier-total-mass}.  Thus, the total arithmetic,
twiddle, and leaf-kernel work is \(O(ns(1+\log n))\).

No additional \(\Theta(ns)\) rearrangement term is incurred at each
transform split, because Appendix~A proves that transform splits change
only recursive views and do not materialize physical transposes.
\end{proof}
}

\subsection{Recursion Depth and Stack Space}

\begin{lemma}[Recursion depth]
\label{lem:local-recursion-depth}
Along every root-to-leaf path, the scheduler performs
\begin{equation}
O\bigl(1+\log\log_b n\bigr)
\label{eq:transform-split-depth}
\end{equation}
transform splits and
\begin{equation}
O\bigl(\log_b(ns)\bigr)
\label{eq:batch-split-depth}
\end{equation}
batch splits.  Hence, its total recursion depth is
\begin{equation}
O\bigl(\log(ns)\bigr).
\label{eq:total-recursion-depth}
\end{equation}
\end{lemma}

\begin{proof}
Let \(d=\log_b n\) be the number of active transform digits.  A balanced
transform split leaves either \(\lceil d/2\rceil\) or
\(\lfloor d/2\rfloor\) transform digits in a child.  Repeated balanced
splitting therefore reaches a constant-size transform after
\(O(1+\log d)\) transform splits, which gives
\eqref{transform-split-depth}.

Every batch split fixes one active batch digit.  Initially there are
\(\log_b s\) batch digits.  A transform split may move transform digits
into an outer-batch field, but along one root-to-leaf path the total
number of such transferred digits is at most \(\log_b n\).  Thus, no
more than \(\log_b(ns)\) active batch digits can be fixed, proving
\eqref{batch-split-depth}.  Summing the two bounds gives
\eqref{total-recursion-depth}.
\end{proof}

\begin{corollary}[Stack space]
\label{cor:local-stack-space}
The scheduler uses
\begin{equation}
S_{\mathrm{stack}}=O\bigl(\log(ns)\bigr)
\label{eq:local-stack-space}
\end{equation}
words of recursion-stack space.
\end{corollary}

\begin{proof}
Each stack frame stores only a constant number of field descriptors,
fixed-digit assignments, loop counters, and mode bits.  The result
follows directly from Lemma~\ref{lem:local-recursion-depth}.
\end{proof}

\subsection{Analytical Cache-Fitting Frontier}

We now analyze cache complexity.  The following frontier is used only
in the proof; the scheduler never tests \(M\), \(B\), or the
cache-fitting condition.

Fix a sufficiently small constant \(\alpha>0\), and define
\begin{equation}
M^\star=\alpha M.
\label{eq:cache-fitting-capacity}
\end{equation}
A view with parameters \((n,s,c)\) is \emph{cache fitting} if
\begin{equation}
ns\le M^\star.
\label{eq:cache-fitting-condition}
\end{equation}
A maximal cache-fitting view is either the root, when the root is cache
fitting, or a cache-fitting view whose parent is not cache fitting.

The purpose of this frontier is to identify subproblems whose complete
descendant execution can remain resident in the ideal cache.  The
scheduler continues recursively below the frontier; the frontier merely
groups that descendant work for analysis.

\subsection{Frontier Geometry}

\begin{lemma}[Frontier entry is through a batch split]
\label{lem:frontier-entry-batch}
A non-cache-fitting view can first enter the cache-fitting region only
through a batch split.
\end{lemma}

\begin{proof}
A transform split preserves the working-set mass of each child group,
as shown in \eqref{transform-child-mass}.  It therefore cannot change a
non-cache-fitting mass \(ns>M^\star\) into a cache-fitting mass.
Conversely, a batch split reduces the batch size of each child by a
factor of \(b\), and can therefore produce a cache-fitting child.
\end{proof}

\begin{lemma}[Volume of a nonroot frontier view]
\label{lem:frontier-volume}
Every nonroot maximal cache-fitting view satisfies
\begin{equation}
\frac{M^\star}{b}<ns\le M^\star.
\label{eq:frontier-volume}
\end{equation}
\end{lemma}

\begin{proof}
By Lemma~\ref{lem:frontier-entry-batch}, a nonroot maximal
cache-fitting view is created by a batch split.  Its parent has the
same transform length and \(b\) times as many batch instances.  Since
the child fits but the parent does not,
\[
ns\le M^\star
\qquad\text{and}\qquad
bn s>M^\star.
\]
These inequalities are equivalent to \eqref{frontier-volume}.
\end{proof}

\subsection{Contiguous Width at the Frontier}

\begin{lemma}[Contiguous width at the frontier]
\label{lem:frontier-contiguous-width}
Suppose the root view has initial contiguous width \(c_0\ge B\).  Then
every maximal cache-fitting view has
\begin{equation}
c=\Omega(B).
\label{eq:frontier-contiguous-width}
\end{equation}
Moreover, if the root-to-frontier path contains a split of the
distinguished field \(C\), then
\begin{equation}
c=\Omega(\sqrt{M}).
\label{eq:frontier-contiguous-width-strong}
\end{equation}
\end{lemma}

\begin{proof}
If the distinguished field \(C\) is never split on the path, then its
width remains \(c=c_0\ge B\).

Otherwise, consider the last split of \(C\) before the frontier view,
and let \(c\) denote the child width immediately after that split.  The
parent width was \(bc\).  Since the scheduler splits \(C\) only when
\(O=\varnothing\), the parent batch size was also \(bc\).  The
batch-split condition therefore gives
\begin{equation}
n\le bc
\label{eq:contiguous-split-condition}
\end{equation}
at that split.

If the resulting child is already cache fitting, its parent is not.
Using \eqref{contiguous-split-condition},
\begin{equation}
M^\star<nbc\le b^2c^2.
\label{eq:contiguous-width-fitting-child}
\end{equation}
If the child is not yet cache fitting, then
\begin{equation}
M^\star<nc\le bc^2.
\label{eq:contiguous-width-nonfitting-child}
\end{equation}
Both cases imply \(c=\Omega(\sqrt{M^\star})=\Omega(\sqrt{M})\).
\punt{All subsequent operations before the frontier are outer-batch or
transform splits, neither of which changes \(c\).} This proves
\eqref{frontier-contiguous-width-strong}.  The tall-cache assumption
\(M=\Omega(B^2)\) then implies \eqref{frontier-contiguous-width}.
\end{proof}

\subsection{Reachable States after Transform Splitting}

\begin{lemma}[Transform--batch relation]
\label{lem:reachable-state}
There is a constant \(\kappa_b=O(1)\), depending only on the fixed
radix \(b\), such that every view reached after at least one transform
split satisfies
\begin{equation}
s\le \kappa_b n^3.
\label{eq:reachable-state}
\end{equation}
\end{lemma}

\begin{proof}
Suppose a transform split is applied to a view with
\(n=n_1n_2\) and \(s<n\).  The balanced digit-aligned split satisfies
\[
\frac{1}{b}\le \frac{n_1}{n_2}\le b.
\]

For the first child, the transform length is \(n_1\) and the batch size
is \(sn_2\).  Since \(s<n_1n_2\),
\[
sn_2<n_1n_2^2\le b^2n_1^3.
\]
The second child is symmetric:
\[
sn_1<n_2n_1^2\le b^2n_2^3.
\]
Thus, both immediate children satisfy \eqref{reachable-state} with
\(\kappa_b=b^2\).  A later batch split decreases \(s\) without changing
\(n\), while a later transform split establishes the same inequality
for its children.  The claim follows by induction along the recursion
path.
\end{proof}

\begin{lemma}[Transform length at the frontier]
\label{lem:frontier-transform-length}
Let a nonroot maximal cache-fitting view be reached along a path
containing at least one transform split.  Then
\begin{equation}
n=\Omega(M^{1/4}).
\label{eq:frontier-transform-length}
\end{equation}
\end{lemma}

\begin{proof}
By Lemma~\ref{lem:frontier-volume},
\[
\frac{M^\star}{b}<ns.
\]
By Lemma~\ref{lem:reachable-state},
\[
s\le \kappa_b n^3.
\]
Consequently,
\[
\frac{M^\star}{b}<\kappa_b n^4,
\]
which proves \eqref{frontier-transform-length}.
\end{proof}

\subsection{Transform Splits before the Frontier}

\begin{lemma}[Transform-size shrinkage]
\label{lem:transform-size-shrinkage}
Let \(n_j\) be the transform length before the \(j\)-th transform split
on a root-to-frontier path.  The larger child satisfies
\begin{equation}
n_{j+1}\le \sqrt{b n_j}.
\label{eq:transform-size-shrinkage}
\end{equation}
Hence, after \(K\) transform splits,
\begin{equation}
n_K
\le
b^{1-1/2^K}n^{1/2^K}
<
b n^{1/2^K}.
\label{eq:transform-size-after-k-splits}
\end{equation}
\end{lemma}

\begin{proof}
If \(n_j=b^{d_j}\), the larger child contains at most
\(\lceil d_j/2\rceil\) transform digits.  Therefore,
\[
n_{j+1}
\le
b^{\lceil d_j/2\rceil}
\le
b^{1/2}n_j^{1/2},
\]
which proves \eqref{transform-size-shrinkage}.  Repeated substitution
gives \eqref{transform-size-after-k-splits}.
\end{proof}

\begin{lemma}[Number of transform splits before the frontier]
\label{lem:frontier-transform-splits}
If \(K\) is the number of transform splits on a root-to-frontier path,
then
\begin{equation}
2^K=O\bigl(1+\log_M n\bigr).
\label{eq:frontier-transform-splits}
\end{equation}
\end{lemma}

\begin{proof}
The claim is immediate when \(K=0\).  Suppose \(K\ge1\).  By
Lemma~\ref{lem:frontier-transform-length}, the frontier transform
length satisfies \(n_K=\Omega(M^{1/4})\).  Combining this with
\eqref{transform-size-after-k-splits} gives
\[
\Omega(M^{1/4})
\le
b n^{1/2^K}.
\]
Taking logarithms yields
\[
\frac{\log n}{2^K}
=
\Omega(\log M),
\]
whenever \(n\) is asymptotically larger than \(M\).  Thus,
\[
2^K=O(\log_M n)
\]
in that case.

If \(n=O(M)\), only \(O(1)\) balanced transform splits can occur before
the transform length becomes constant or the path reaches the
cache-fitting frontier.  Since \(1+\log_M n=O(1)\) in this case, the
same asymptotic bound holds.  Combining the two cases proves
\eqref{frontier-transform-splits}.
\end{proof}

\subsection{Block Cover of a Frontier View}

\begin{lemma}[Block cover of one frontier view]
\label{lem:frontier-block-cover}
A maximal cache-fitting view of mass \(ns\) intersects
\begin{equation}
O\left(1+\frac{ns}{B}\right)
\label{eq:frontier-block-cover}
\end{equation}
cache blocks.  Its complete depth-first descendant execution incurs
the same asymptotic number of cache misses.
\end{lemma}

\begin{proof}
By the active-run decomposition proved in Appendix~A, the view is the
disjoint union of \(ns/c\) contiguous runs of length \(c\).  Therefore,
its block footprint is at most
\begin{align}
O\left(
\frac{ns}{c}
\left(1+\frac{c}{B}\right)
\right)
&=
O\left(
\frac{ns}{c}+\frac{ns}{B}
\right)
\nonumber\\
&=
O\left(1+\frac{ns}{B}\right),
\label{eq:frontier-block-count}
\end{align}
where the last step uses Lemma~\ref{lem:frontier-contiguous-width}.

Every descendant of the frontier view accesses only items in the
frontier view's working set: batch splits select subsets of its runs,
contiguous-field splits subdivide those runs, and transform splits
reinterpret the same addresses through sequential child groups.
Moreover, the view contains at most \(M^\star=\alpha M\) items.
Choosing \(\alpha\) sufficiently small ensures that its block footprint,
together with constant recursion metadata, fits in cache.  Under
depth-first execution, these blocks can remain resident while the
complete descendant computation is performed.  Hence,
\eqref{frontier-block-count} also bounds the descendant cache cost.
\end{proof}

\subsection{Total Frontier Mass}

\begin{lemma}[Total frontier mass]
\label{lem:frontier-total-mass}
Suppose every root-to-frontier path contains at most \(K\) transform
splits.  Then the sum of the masses of all maximal cache-fitting views
is at most
\begin{equation}
ns\,2^K.
\label{eq:frontier-total-mass}
\end{equation}
\end{lemma}

\begin{proof}
We prove the more general statement that a recursion subtree rooted at
a view of mass \(\mu\), with at most \(k\) remaining transform splits
on each path, has total frontier mass at most \(\mu 2^k\).

If the root of the subtree is cache fitting, it is itself a frontier
view and contributes mass \(\mu\), which is at most \(\mu 2^k\).

If the root performs a batch split, the child masses sum to \(\mu\).
Applying the induction hypothesis to every child gives total frontier
mass at most
\[
2^k\sum_{\text{batch children }u}\mathrm{mass}(u)
=
\mu 2^k.
\]

If the root performs a transform split, both child groups have mass
\(\mu\) by \eqref{transform-child-mass}, and each has at most \(k-1\)
remaining transform splits.  Their combined frontier mass is therefore
at most
\[
2\mu 2^{k-1}
=
\mu 2^k.
\]
Applying the claim at the original root proves
\eqref{frontier-total-mass}.
\end{proof}

\subsection{Cache Complexity}

\begin{theorem}[Cache complexity of the batched local FFT]
\label{thm:local-cache-complexity}
Let an admissible root view contain \(s\) transforms of length \(n\),
with initial contiguous width \(c\ge B\).  Under the tall-cache
assumption \(M=\Omega(B^2)\), the scheduler incurs
\begin{equation}
Q_{\mathrm{local}}(n,s,c)
=
O\left(
\frac{ns}{B}
\bigl(1+\log_M n\bigr)
\right)
\label{eq:local-cache-complexity}
\end{equation}
cache misses.
\end{theorem}

\begin{proof}
Let \(K\) be the maximum number of transform splits on a root-to-frontier
path.  Lemma~\ref{lem:frontier-transform-splits} gives
\[
2^K=O(1+\log_M n).
\]
By Lemma~\ref{lem:frontier-total-mass}, the total mass of the maximal
cache-fitting frontier is therefore
\begin{equation}
O\bigl(ns(1+\log_M n)\bigr).
\label{eq:total-frontier-mass-bound}
\end{equation}

By Lemma~\ref{lem:frontier-block-cover}, each frontier view of mass
\(\mu\) costs \(O(1+\mu/B)\) cache misses.  The additive constant is
absorbed in the mass term: by
Lemma~\ref{lem:frontier-contiguous-width}, every frontier view has
contiguous width \(\Omega(B)\), and hence mass \(\mu=\Omega(B)\).
Summing the cache costs over the frontier and using
\eqref{total-frontier-mass-bound} gives
\[
Q_{\mathrm{local}}(n,s,c)
=
O\left(
\frac{ns}{B}
\bigl(1+\log_M n\bigr)
\right).
\]
This proves \eqref{local-cache-complexity}.
\end{proof}

\subsection{Application to the Two PACO Local Phases}

In Phase~I, each processor performs \(N_2/p\) transforms of length
\(N_1\).  Its initial contiguous batch field is precisely the local
input slab, so
\begin{equation}
n=N_1,
\qquad
s=c=\frac{N_2}{p}.
\label{eq:phase-i-local-parameters}
\end{equation}
The slab-width assumption \(N_2/p\ge B\) permits direct application of
Theorem~\ref{thm:local-cache-complexity}, yielding
\begin{equation}
Q_{q,\mathrm{I}}
=
O\left(
\frac{N}{pB}
\bigl(1+\log_M N_1\bigr)
\right).
\label{eq:phase-i-cache-complexity}
\end{equation}

Similarly, in Phase~III, each processor performs \(N_1/p\) transforms
of length \(N_2\), with
\begin{equation}
n=N_2,
\qquad
s=c=\frac{N_1}{p}.
\label{eq:phase-iii-local-parameters}
\end{equation}
Since \(N_1/p\ge B\),
\begin{equation}
Q_{q,\mathrm{III}}
=
O\left(
\frac{N}{pB}
\bigl(1+\log_M N_2\bigr)
\right).
\label{eq:phase-iii-cache-complexity}
\end{equation}

Finally,
\[
\log_M N_1+\log_M N_2=\log_M N.
\]
Adding \eqref{phase-i-cache-complexity} and
\eqref{phase-iii-cache-complexity}, and absorbing the linear cache cost
of the fused middle stage, gives
\begin{equation}
Q_{q,\mathrm{I}}+Q_{q,\mathrm{III}}
=
O\left(
\frac{N}{pB}
\bigl(1+\log_M N\bigr)
\right).
\label{eq:two-local-phases-cache-complexity}
\end{equation}
}
\section{Data Movement: Local Primitives and the Fused Redistribution}
\label{app:local}\label{app:colocal}

This appendix specifies the data-movement machinery used in Phases~I--III. It
first defines the local primitives---macro-element digit reversal, its
destination-indexed fusion, and cache-oblivious rectangular
transposition---and then composes them into the explicit per-processor
procedure that realizes PACO's fused middle redistribution. The relevant
permutations operate on contiguous \emph{macro-elements}, not on individual
scalars. The cache line size is $B$ entries. In both local applications every
macro-element spans at least one line: by the slab-width conditions~\eqref{slab-width-cache-block} the two
uses below have macro-element length $N_2/p$ and $N_1/p$, respectively, and both
are $\ge B$.

Sections~3.3 and~4.2 already established the fused middle stage at the level
needed for the main result: its coordinate trace~(18), its end-to-end
correctness (Theorem~3), its balanced migration counts, and its per-processor
local cost (Theorem~4). The second half of this appendix supplies the single
deferred component, namely the explicit cache-oblivious procedure that realizes
the mapping, together with (i) a line-by-line correctness argument and (ii) the
workspace and transpose-aggregation accounting promised in Sect.~4.2. It reuses
the source ownership~(1), the factor-swapped target ownership~(3), the Phase-I
identity $\widetilde A[a,c]=A[\pi_{m_1}(a),c]$, the digit-reversal notation of
Appendix~\ref{app:conv}, and the slab-width conditions~\eqref{slab-width-cache-block}.

\subsection{Macro-Element Digit Reversal}
\label{app:mdr}

Let the array hold $n=b^{\ell}$ contiguous macro-elements, each of length $s$;
write $\ell=\log_b n$ (the count of macro-elements, distinct from the global
$m=\log_b N$). For $i\in[n]$ with base-$b$ digits $d_{\ell-1}\dots d_0$, the
digit-reversal permutation is
$\pi_\ell(i)=d_0b^{\ell-1}+\cdots+d_{\ell-1}$, an involution, so every orbit is
a fixed point or a two-cycle. For the array
$X[0\twodots n-1][0\twodots s-1]$, the desired rearrangement is
\begin{equation}\label{eq:mdr}
  X'[i,j]=X[\pi_\ell(i),j].
\end{equation}
Since $\pi_\ell$ is an involution, \eqref{mdr} is realized in place by
exchanging the two macro-elements of each nontrivial two-cycle; the condition
$i<\pi_\ell(i)$ selects one representative per two-cycle.

\begin{codebox}
\Procname{$\proc{Macro-Digit-Reversal}(X, n, s)$}
\li \Comment $\ell = \log_b n$; reverse the $n$ contiguous macro-elements
\li \For $i \gets 0$ \To $n - 1$ \Do
\li     \If $i < \pi_\ell(i)$ \Then 
\li         \For $j \gets 0$ \To $s - 1$ \Do
\li             \(X[i,j] \leftrightarrow X[\pi_\ell(i),j]\)
            \End
        \End
    \End
\end{codebox}

\begin{lemma}[Correctness]\label{lem:mdrcorr}
The procedure \proc{Macro-Digit-Reversal} implements \eqref{mdr}.
\end{lemma}
\begin{proof}
If $i=\pi_\ell(i)$, the macro-element is a fixed point and left unchanged,
agreeing with \eqref{mdr}. Otherwise $i$ lies in the two-cycle
$\{i,\pi_\ell(i)\}$, and exactly one member satisfies $i<\pi_\ell(i)$; when
processed, the two macro-elements are exchanged entry by entry, so afterward
$X[i,j]=X_{\mathrm{old}}[\pi_\ell(i),j]$ for both members. Every non-fixed
macro-element is exchanged exactly once.
\end{proof}

\begin{lemma}[Cost]\label{lem:mdrcost}
If $s\ge B$, the procedure \proc{Macro-Digit-Reversal} has local work $O(ns)$ and cache complexity
$O(ns/B)$, using $O(1)$ auxiliary scalar storage.
\end{lemma}
\begin{proof}
Each macro-element is scanned at most once, giving $O(ns)$ scalar reads and
writes. Each contiguous segment of $s$ entries costs $O(s/B)$ misses, and
neither member of a two-cycle is revisited after its exchange; summing over
$n$ macro-elements gives $O(ns/B)$. One temporary scalar suffices.
\end{proof}

\subsection{Fusing Destination-Indexed Unary Operations}
\label{app:fused}

A macro-element rearrangement may be fused with a unary operation whose
coefficient depends on the destination. Let $g(i,j)$ be the coefficient
applied to the entry written to $(i,j)$; the fused transformation is
\begin{equation}\label{eq:fused}
  X'[i,j]=g(i,j)\,X[\pi_\ell(i),j].
\end{equation}
For a nontrivial two-cycle $\{i,k\}$, $k=\pi_\ell(i)$, $i<k$, the writes are
$X[i,j]=g(i,j)X_{\mathrm{old}}[k,j]$ and
$X[k,j]=g(k,j)X_{\mathrm{old}}[i,j]$; for a fixed point,
$X[i,j]=g(i,j)X[i,j]$.

\begin{codebox}
\Procname{$\proc{Fused-Macro-Digit-Reversal}(X, n, s, g)$}
\li \For $i \gets 0$ \To $n - 1$ \Do
\li     \If $i < \pi_\ell(i)$ \Then
\li         \For $j \gets 0$ \To $s - 1$ \Do
\li             $x \gets X[i,j]$
\li             $y \gets X[\pi_\ell(i),j]$
\li             $X[i,j] \gets g(i,j)\cdot y$
\li             $X[\pi_\ell(i),j] \gets g(\pi_\ell(i),j)\cdot x$
            \End
\li     \ElseIf $i \isequal \pi_\ell(i)$ \Then
\li         \For $j \gets 0$ \To $s - 1$ \Do
\li             $X[i,j] \gets g(i,j)\cdot X[i,j]$
            \End
        \End
    \End
\end{codebox}

\begin{lemma}[Correctness]\label{lem:fusedcorr}
The procedure \proc{Fused-Macro-Digit-Reversal} implements \eqref{fused}.
\end{lemma}
\begin{proof}
For a nontrivial two-cycle the two writes place each old macro-element at its
digit-reversed destination and scale by the destination-indexed coefficient;
for a fixed point the update is exactly \eqref{fused}. As in
\lemref{mdrcorr}, every two-cycle is processed once.
\end{proof}

\begin{corollary}[Cost of fused rearrangement]\label{cor:fusedcost}
If $s\ge B$ and evaluating $g$ takes $O(1)$ work, the procedure \proc{Fused-Macro-Digit-Reversal} has local work
$O(ns)$ and cache complexity $O(ns/B)$.
\end{corollary}
\begin{proof}
The fused operation adds $O(1)$ arithmetic per entry to the scans of
~\lemref{mdrcost}, with neither an extra traversal nor extra storage.
\end{proof}

\paragraph{Application to twiddle multiplication.}
In the middle stage the destination-indexed operation is twiddle
multiplication: with destination macro-element index $v$ and within-element (local) 
coordinate $c_{\text{loc}}$, $g(v,c)=\omega_N^{-vc}$, where \(c = q \cdot (N_2/p) + c_{\text{loc}}\) is the global column. Hence twiddles are applied when an
entry is written to its destination; no separate pass is required.

\subsection{Use in Phases I and III}
\label{app:mdruse}

Both local reversals are covered by \lemref{mdrcost} and
\corref{fusedcost} with a single parameterized statement. In Phase~I and III each
processor reverses $n$ macro-elements of length $s\in\{N_2/p,\,N_1/p\}$; by
condition~\eqref{slab-width-cache-block}, $s\ge B$, so rearranging the $n$ macro-elements costs
$O(ns)$ work and $O(ns/B)$ misses. Instantiating $s=N_2/p$ (Phase~I) and
$s=N_1/p$ (Phase~III), with $n=N_1$ and $n=N_2$ respectively, gives
$O(N_1N_2/p)=O(N/p)$ work and $O(N/(pB))$ misses per reversal. The
macro-element length lower bound $s\ge B$ is exactly what ensures spatial
locality, so no cache-oblivious traversal of a scalar digit-reversal
permutation is needed.

\subsection{Local Rectangular Transposition}
\label{app:transpose}

After communication, each received block is a rectangle of dimensions
\begin{equation}\label{eq:rect}
  r\times c=\frac{N_1}{p}\times\frac{N_2}{p}.
\end{equation}
PACO transposes it with the cache-oblivious rectangular-transpose algorithm of
Frigo et al.~\cite{FrigoLePr99}, whose standard bound is: transposing an
$r\times c$ rectangle takes $O(rc)$ work and $O(1+rc/B)$ misses. For a block
of dimensions~\eqref{rect} this is
\begin{equation}\label{eq:blockwork}
  O\!\Bigl(\tfrac{N_1N_2}{p^2}\Bigr)=O\!\Bigl(\tfrac{N}{p^2}\Bigr)
  \quad\text{work,}\qquad
  O\!\Bigl(1+\tfrac{N}{p^2B}\Bigr)\ \text{misses,}
\end{equation}
and the reusable transpose workspace is $(N_1/p)(N_2/p)=N/p^2$ entries.

This subsection is the single source for the per-block transpose accounting;
its aggregation over the $p$ blocks of a processor is carried out in
Appendix~\ref{app:acct} (and used in Theorem~4).

\subsection{The Per-Processor Fused Redistribution Procedure}
\label{app:global}\label{app:proc}

Processor $q$ enters the middle stage owning an $N_1\times(N_2/p)$ Phase-I slab
$X_q$. The procedure $\proc{Para-Co-Drp}_b$ composes the three primitives
above: it packs the slab by reversing the $N_1$-field
(Appendix~\ref{app:mdr}) while fusing the top-level twiddle
(Appendix~\ref{app:fused}), exchanges one balanced set of blocks in a single
BSP $h$-relation, transposes each received block cache-obliviously
(Appendix~\ref{app:transpose}), and unpacks by reversing the $N_2$-field. Each
field reversal is an in-place sequence of contiguous macro-element swaps; the
fused twiddle multiplication is folded into the packing scan, costing no
separate pass.

\begin{codebox}
\Procname{$\proc{Para-Co-Drp}_b(X_q, Y_q, N_1, N_2, p, q)$}
\zi \Comment \(X_q\) is the local input array, \(Y_q\) is the output
\zi \Comment \(q\) is the local processor id
\zi \Comment Local pack: $N_1$-field reversal with the top-level twiddle
\li view $X_q$ as $N_1$ row macro-elements of length $N_2/p$
\zi \Comment \(g(v, c_{\text{loc}}) = \omega_N^{-v c}, \text{ where } c = q\cdot (N_2/p) + c_{\text{loc}}  \text{ (global column)}\)
\zi \qquad \(, N = N1 N2\)
\li \(\proc{Fused-Macro-Digit-Reversal}(X_q, N_1, N_2/p, g)\) 
\punt{
\li \For each two-cycle $\{a,\pi_{m_1}(a)\}$ of $\pi_{m_1}$ with $a<\pi_{m_1}(a)$ \Do
\li     swap row macro-elements $a$ and $\pi_{m_1}(a)$
        \End
\li \For every packed coordinate $(v,c)$, where $v=\pi_{m_1}(a)$ \Do
\li     $Z[v,c] \gets A[v,c]\cdot \omega_N^{-vc}$
        \End
}
\zi \Comment One balanced BSP all-to-all exchange
\li partition the packed slab into blocks $S_{q,q'}$ of 
\zi \qquad shape $(N_1/p)\times(N_2/p)$, $q'\in[p]$
\li place the self block $S_{q,q}$ in its raw output region $R_{q,q}$
\li \For $q' \gets 0$ \To $p - 1$ \Do
\li     \If $q' \neq q$ \Then
\li         send $S_{q,q'}$ to raw region $R_{q',q}$ on processor $q'$
        \End
    \End
\zi \Comment Local rectangular transpose of each received block
\li allocate or reuse one scratch array $T$ of $N/p^2$ items
\li \For $r \gets 0$ \To $p - 1$ \Do
\li     $\proc{CO-Transpose}(R_{q,r}, T, N_1/p, N_2/p)$
\li     copy $T$ to the output region $Y_{q,r}$
        \End
\zi \Comment Local unpack: $N_2$-field digit reversal
\li view the local array as $N_2$ row macro-elements of length $N_1/p$
\li \(\proc{Macro-Digit-Reversal}(Y_{q}, N_2, N_1/p)\)
\punt{
\li \For each two-cycle $\{c,\pi_{m_2}(c)\}$ of $\pi_{m_2}$ with $c<\pi_{m_2}(c)$ \Do
\li     swap row macro-elements $c$ and $\pi_{m_2}(c)$
        \End
}
\li \Return the local $N_2\times(N_1/p)$ target slab
\end{codebox}

The $p$ raw incoming regions $R_{q,\cdot}$ together form the mandatory target
slab; the only auxiliary storage is the single transpose scratch $T$ of
$N/p^2$ items.

\subsection{Line-by-Line Correctness}
\label{app:line-by-line}

We verify that the procedure \(\proc{Para-Co-Drp}_b\) realizes exactly the coordinate trace~(18); the
algebraic correctness of the whole pipeline then follows from Theorem~3.

\begin{lemma}[The procedure implements $\pi_m$]\label{lem:trace}
For an item stored at Phase-I physical coordinate $(a,c)$ on its source
processor $q$, write $c=q(N_2/p)+c_{\mathrm{loc}}$. Then
$\proc{Para-Co-Drp}_b$ places it at
$(u,v)=\bigl(\pi_{m_2}(c),\pi_{m_1}(a)\bigr)$, i.e.
$\pi_m(a\pplus c)=\pi_{m_2}(c)\pplus\pi_{m_1}(a)$, owned by
$\lfloor\pi_{m_1}(a)/(N_1/p)\rfloor$.
\end{lemma}
\begin{proof}
The row reversal moves the item to packed row $v=\pi_{m_1}(a)$ with column $c$
unchanged; the fused scan forms $Z[v,c]=A[v,c]\,\omega_N^{-vc}$ (the local-pack
step of Sect.~4.2). Decomposing $v=q'(N_1/p)+v_{\mathrm{loc}}$ identifies the
unique destination $q'=\lfloor v/(N_1/p)\rfloor$, which is the owner
under~(3), so the item is assigned to block $S_{q,q'}$ at block-local
coordinate $(v_{\mathrm{loc}},c_{\mathrm{loc}})$ and routed to region
$R_{q',q}$ on $q'$. Because raw regions are ordered by source processor, the
pair $(q,c_{\mathrm{loc}})$ recovers $c=q(N_2/p)+c_{\mathrm{loc}}$. The block
transpose sends $(v_{\mathrm{loc}},c_{\mathrm{loc}})\mapsto
(c_{\mathrm{loc}},v_{\mathrm{loc}})$, and the $N_2$-field reversal maps the
reconstructed source row $c$ to $u=\pi_{m_2}(c)$, leaving
$v=\pi_{m_1}(a)$ intact. The result $(u,v)=(\pi_{m_2}(c),\pi_{m_1}(a))$ is
precisely the trace~(18); its owner is $\lfloor v/(N_1/p)\rfloor$ by~(3).
\end{proof}

By \lemref{trace} the procedure \(\proc{Para-Co-Drp}_b\) produces, for each fixed $v=k_1$, the
local layout $\widehat Z[u,v]=Z[v,\pi_{m_2}(u)]$ of~(17): exactly the
digit-reversed input layout consumed by $\mathcal{B}_{m_2}=\Omega_{N_2}
P_{\pi_{m_2}}$. Hence the composition~(10) of Theorem~3 applies verbatim, and
PACO returns the canonical coefficients $Y[k_2,k_1]=y[k_1+k_2N_1]$ under the
target ownership~(3).

\subsection{Cost and Workspace}\label{app:acct}
Each primitive was costed in Appendices~C.1--C.4. The two fused field reversals
(packing $N_1$ macro-elements of length $N_2/p$; unpacking $N_2$ of length
$N_1/p$) each cost $O(N/p)$ work and $O(N/(pB))$ misses by Corollary~27, since
both macro-element widths are $\ge B$ by~(5); the fused twiddle adds
$\Theta(N/p)$ arithmetic with no extra cache pass. The $p$ block transposes cost
$O(N/p)$ work and, by~(74) together with $p\le N/(pB^2)=O(N/(pB))$, aggregate to
$O(N/(pB))$ misses. Summing reproduces~(22)--(24). The only non-output workspace
is one reusable $N/p^2$-item transpose buffer plus $O(\log N)$ recursion-stack
words, giving~(28); the balanced counts $\mu=N(1-1/p)$ and $h=(N/p)(1-1/p)$ are
as in Theorem~4.

\punt{
\subsection{Cost, Workspace, and Transpose Aggregation}
\label{app:acct}

We now supply the accounting behind the local-cost claims (22)--(24) of
Theorem~4.

\paragraph{Field reversals.}
Packing reverses $N_1$ macro-elements of length $N_2/p$ and unpacking reverses
$N_2$ macro-elements of length $N_1/p$. By~\eqref{slab-width-cache-block} every macro-element spans at
least $B$ items, so the per-reversal instantiation of
Appendix~\ref{app:mdruse}, via the fused-cost ~\corref{fusedcost},
gives $O(N/p)$ work and $O(N/(pB))$ misses for each reversal. The fused twiddle
contributes one multiplication per packed item, i.e. $\Theta(N/p)$ arithmetic
folded into the packing scan, with no additional asymptotic cache pass.

\paragraph{Transpose aggregation.}
Each processor transposes its $p$ received blocks. Applying the per-block bound
of Appendix~\ref{app:transpose} (\eqref{blockwork}) and summing over the
$p$ blocks gives $O(p+N/(pB))$ misses and $p\cdot O(N/p^2)=O(N/p)$ work. Since
$N/p^2=(N_1/p)(N_2/p)\ge B^2$ by~\eqref{slab-width-cache-block}, we have $p\le N/(pB^2)=O(N/(pB))$, so the
additive $O(p)$ term is absorbed and the transpose cost is $O(N/(pB))$.

\paragraph{Totals.}
Adding the two reversals, the fused twiddle scan, and the transposes yields
\[
  W^{\mathrm{arith}}_{q,\mathrm{twiddle}}=\Theta\!\Bigl(\tfrac{N}{p}\Bigr),\qquad
  W^{\mathrm{rearr}}_{q,\mathrm{perm}}=O\!\Bigl(\tfrac{N}{p}\Bigr),\qquad
  Q_{q,\mathrm{perm}}=O\!\Bigl(\tfrac{N}{pB}\Bigr),
\]
matching (22)--(24).

\paragraph{Workspace.}
The $p$ raw incoming regions occupy $N/p$ items in aggregate, but they are the
mandatory target slab, not auxiliary storage. The macro-element swaps are in
place (one temporary item), and the only non-output array is the reusable
transpose scratch $T$ of $N/p^2$ items. With the $O(\log N)$ recursion-stack
words of the cache-oblivious transpose, the auxiliary workspace is
$O(N/p^2+\log N)$, as claimed in~(28). The aggregate balance counts
$\mu=N(1-1/p)$ and $h=(N/p)(1-1/p)$ are the block-count bookkeeping of
Theorem~4 and are not repeated here.
}
\section{Lower Bounds: Deferred Proofs}
\label{app:lb}

The main text establishes three of \algName{}'s four lower bounds in full: the
one-round redistribution bound (Proposition~3), the exact migration bound for
the prescribed fused permutation (Proposition~4), and the maximum per-processor work
bound (Proposition~5). The maximum per-processor cache lower bound (Proposition~6) was stated
there with only a proof roadmap; its detailed segment argument is deferred to
this appendix.

We import the FFT-DAG input-boundary theorem \cite{RanjanSaZu11}, convert boundary growth into
block-cache misses, obtain a finite block-aware bound for a single radix-2 FFT
DAG, extend it to disjoint batches, and specialize it to \algName{}'s two local
phases, completing Proposition~6. As a byproduct we also record a stronger
constant-one asymptotic miss bound. Throughout, arithmetic vertices of the
fixed-radix Cooley--Tukey DAG are never recomputed.

\begin{remark}[Fixed-radix scope]\label{rmk:fixed-radix-scope}
Let $b=2^t$, $t=O(1)$. If every radix-$b$ butterfly is expanded into its $t$
constituent radix-2 layers, the resulting dependency graph is the standard
radix-2 FFT DAG. Hence all lower bounds of this appendix---the single-transform
bounds of Appendices~\ref{app:finite} and~\ref{app:strong}, the batched bounds
of Appendix~\ref{app:batch}, and the \algName{} local bound of Appendix~\ref{app:PACOlb}---apply
to the expanded fixed-radix DAG. If a complete radix-$b$ DFT is instead treated
as one indivisible atomic vertex, the graph structure and leading constants
differ; that atomic model is outside the scope of these claims.
\end{remark}

\subsection{Imported FFT-DAG Input-Boundary Theorem}
\label{app:import}

Let $G_2(N)$ denote the radix-2 FFT DAG on $N=2^m$ inputs; it has $m+1$ levels
of $N$ vertices and $N\log_2 N$ noninput computation vertices. For a set $U$ of
noninput vertices, its input boundary is
\begin{equation}\label{e:inU}
  \func{in}(U)=\{v\notin U:\exists\,u\in U\text{ with }(v,u)\in E\}.
\end{equation}

\begin{theorem}[FFT input-boundary growth]\label{thm:29}
For every integer $k\ge1$, every set $U$ of noninput vertices of $G_2(N)$ with
$|U|\ge k\log_2(4k)$ satisfies $|\func{in}(U)|\ge k$.
\end{theorem}

Theorem~\ref{thm:29} is imported from the boundary-growth
literature~\cite{RanjanSaZu11}; we use only its graph-theoretic consequence.

\subsection{From Input Boundaries to Block Cache Misses}
\label{app:boundtomiss}

\begin{lemma}[Input boundary to block-cache misses]\label{e:lem30}
Consider one contiguous segment of a legal schedule, and let $U$ be the set of
previously uncomputed FFT vertices evaluated during it. If $|\func{in}(U)|\ge k$,
the segment incurs at least $\bigl\lceil(k-M)/B\bigr\rceil_{+}$ cache misses,
where $\lceil x\rceil_{+}:=\max\{0,\lceil x\rceil\}$.
\end{lemma}
\begin{proof}
At the start of the segment at most $M$ items reside in cache, so at most $M$
of the required boundary values are resident; at least $k-M$ are absent when
$k>M$. Each miss loads one block of at most $B$ items and thus supplies at most
$B$ missing values, so at least $\lceil(k-M)/B\rceil$ misses are required when
$k>M$, and zero otherwise.
\end{proof}

Lemma~\ref{e:lem30} permits a strong algorithmic model: the algorithm may know
$M$ and $B$, choose arbitrary slow-memory layouts, pack $B$ useful values per
loaded block, keep the most useful $M$ values in cache, and maintain multiple
slow-memory copies. None of these eliminates the need to load a block when a
required boundary value is absent.

\subsection{Finite Block-Aware Lower Bound for One FFT DAG}
\label{app:finite}

\begin{theorem}[Finite block-aware cache lower bound]\label{thm:31}
For every integer $k>M$, every no-recomputation topological schedule of
$G_2(N)$ satisfies
\begin{equation}\label{e:finite}
  Q_{\mathrm{miss}}(N,M,B)\ge
  \Bigl\lfloor\tfrac{N\log_2 N}{k\log_2(4k)}\Bigr\rfloor
  \Bigl\lceil\tfrac{k-M}{B}\Bigr\rceil.
\end{equation}
\end{theorem}
\begin{proof}
Partition the schedule into consecutive complete segments, each computing
exactly $k\log_2(4k)$ previously uncomputed vertices (ignore a final
incomplete segment). Since $G_2(N)$ has $N\log_2 N$ noninput vertices, there
are at least $\lfloor N\log_2 N/(k\log_2(4k))\rfloor$ complete segments. Each
has $|\func{in}(U)|\ge k$ by Theorem~\ref{thm:29} and thus at least
$\lceil(k-M)/B\rceil$ misses by Lemma~\ref{e:lem30}; disjoint segments add.
\end{proof}

\subsection{A Stronger Constant-One Asymptotic Bound}
\label{app:strong}

This subsection is not needed for \algName{}'s upper-bound proof; it records the
stronger miss-only bound noted in the main text. Let $\phi(M)$ satisfy
\begin{equation}\label{eq:phi}
  \phi(M)\to\infty,\qquad \log\phi(M)=o(\log M),
\end{equation}
choose $k=\lceil M\phi(M)\rceil$, and assume additionally
\begin{equation}\label{eq:phi2}
  M\phi(M)\log_2\!\bigl(4M\phi(M)\bigr)=o(N\log_2 N).
\end{equation}

\begin{theorem}[Asymptotic constant-one cache lower bound]\label{thm:32}
Under \eqref{phi}--\eqref{phi2}, every no-recomputation schedule of
$G_2(N)$ satisfies
$Q_{\mathrm{miss}}(N,M,B)\ge(1-o(1))\tfrac{N}{B}\log_M N$.
\end{theorem}
\begin{proof}
Substitute $k=\lceil M\phi(M)\rceil$ into Theorem~\ref{thm:31}; by
\eqref{phi2} the floor, ceiling, and final incomplete segment are
lower-order, so
\[
  Q_{\mathrm{miss}}\ge(1-o(1))\frac{N}{B}\cdot
  \frac{\log_2 N}{\log_2 M}\cdot\Bigl(1-\frac{M}{k}\Bigr)\cdot
  \frac{\log_2 M}{\log_2(4k)}.
\]
The choice of $k$ gives $1-M/k=1-1/\phi(M)+o(1)=1-o(1)$, while
$\log\phi(M)=o(\log M)$ gives $\log_2 M/\log_2(4k)=1-o(1)$; finally
$\log_2 N/\log_2 M=\log_M N$.
\end{proof}

Thus the normalized miss-only lower-bound constant is at least $1-o(1)$; this
does not assert that \algName{} or any known schedule attains an upper bound of the
form $(1+o(1))\tfrac{N}{B}\log_M N$.

\subsection{Disjoint Unions of FFT DAGs}
\label{app:disjoint}

Let $G_2^{(s)}(n)=\bigsqcup_{i=1}^{s}G_2(n)^{(i)}$ be the disjoint union of $s$
radix-2 FFT DAGs of length $n$.

\begin{lemma}[Boundary growth for disjoint batches]\label{e:lem33}
For every integer $k\ge1$, every set $U$ of noninput vertices of
$G_2^{(s)}(n)$ with $|U|\ge 4k\log_2(4k)$ satisfies $|\func{in}(U)|\ge k$.
\end{lemma}
\begin{proof}
Contrapositive. Let $U_i=U\cap V(G_2(n)^{(i)})$ and $r_i=|\func{in}(U_i)|$, so
$|\func{in}(U)|=\sum_i r_i$. Suppose $|\func{in}(U)|<k$. For every nonempty $U_i$,
$r_i\ge1$, and the contrapositive of Theorem~\ref{thm:29} with parameter
$r_i+1$ gives $|U_i|<(r_i+1)\log_2(4(r_i+1))$. Since $r_i\ge1$,
$r_i+1\le2r_i$ and $4(r_i+1)\le8r_i$, so $|U_i|<2r_i\log_2(8r_i)\le
2r_i\log_2(8k)$. Summing, $|U|<2(\sum_i r_i)\log_2(8k)<2k\log_2(8k)$. As
$2\log_2(8k)\le4\log_2(4k)$ for $k\ge1$, $|U|<4k\log_2(4k)$.
\end{proof}

The constant $4$ is chosen for a uniform statement and is not claimed tight.

\subsection{Cache Lower Bound for a Batch of Local FFTs}
\label{app:batch}

\begin{theorem}[Finite lower bound for a batch of FFTs]\label{thm:34}
Let one processor evaluate $s$ data-disjoint length-$n$ radix-2 FFT DAGs
without recomputation. For every integer $k>M$,
\begin{equation}\label{eq:batchfinite}
  Q_{\mathrm{batch}}(n,s)\ge
  \Bigl\lfloor\tfrac{sn\log_2 n}{4k\log_2(4k)}\Bigr\rfloor
  \Bigl\lceil\tfrac{k-M}{B}\Bigr\rceil.
\end{equation}
\end{theorem}
\begin{proof}
The disjoint union has $sn\log_2 n$ noninput vertices. Partition into complete
segments of $4k\log_2(4k)$ previously uncomputed vertices; Lemma~\ref{e:lem33}
gives $|\func{in}(U)|\ge k$ per segment, and Lemma~\ref{e:lem30} gives
$\lceil(k-M)/B\rceil$ misses each. Sum over segments.
\end{proof}

\begin{corollary}[Asymptotic batched lower bound]\label{e:cor35}
There are universal constants $c,c'>0$ with
\begin{equation}\label{eq:batchasym}
  Q_{\mathrm{batch}}(n,s)\ge c\,\frac{ns}{B}\bigl(1+\log_M n\bigr)-c'\frac{M}{B}.
\end{equation}
If the cache is initially empty, or if $ns\ge\kappa M$ for a sufficiently large
universal constant $\kappa$, then
$Q_{\mathrm{batch}}(n,s)=\Omega\!\bigl(\tfrac{ns}{B}(1+\log_M n)\bigr)$.
\end{corollary}
\begin{proof}
Take $k=2M$ in Theorem~\ref{thm:34}. Then $k>M$, and
$k-M=M$, $4k=8M$, so $\log_2(4k)=3+\log_2 M$. Hence
\[
Q_{\mathrm{batch}}(n,s)\ \ge\
\left\lfloor \frac{sn\log_2 n}{8M\,(3+\log_2 M)}\right\rfloor
\left\lceil \frac{M}{B}\right\rceil .
\]
Write $A:=\dfrac{sn\log_2 n}{8M(3+\log_2 M)}$.
Using $\lceil M/B\rceil\ge M/B\ge 0$ and $\lfloor A\rfloor\ge A-1$
(the product bound holds trivially when $A<1$, since the left side is
nonnegative),
\[
Q_{\mathrm{batch}}\ \ge\ (A-1)\frac{M}{B}
\ =\ A\,\frac{M}{B}-\frac{M}{B}.
\]
For $M\ge 2$ we have $3+\log_2 M\le 4\log_2 M$, so by the change of base
$\log_M n=\log_2 n/\log_2 M$,
\[
A\,\frac{M}{B}
=\frac{sn}{8B}\cdot\frac{\log_2 n}{3+\log_2 M}
\ \ge\ \frac{sn}{8B}\cdot\frac{\log_M n}{4}
=\frac{1}{32}\,\frac{sn}{B}\log_M n .
\]
This gives the stated bound with $c_1=\tfrac1{32}$ and $c_2=1$:
\[
Q_{\mathrm{batch}}\ \ge\ c_1\frac{ns}{B}\log_M n-c_2\frac{M}{B}.
\]
Independently, the $ns$ inputs are distinct, so with at most $M$ initially
resident, $Q_{\mathrm{batch}}\ge (ns-M)_+/B$. Combining via
$\max\{x,y\}\ge (x+y)/2$ gives \eqref{batchasym}. If the cache is
empty the compulsory scan contributes $\Omega(ns/B)$; if $ns\ge\kappa M$
the linear term absorbs the $M/B$ term.
\end{proof}
\punt{
\begin{proof}
Take $k=2M$ in Theorem~\ref{thm:34}: there are $c_1,c_2>0$ with
$Q_{\mathrm{batch}}\ge c_1\tfrac{ns}{B}\log_M n-c_2\tfrac{M}{B}$.
Independently, the $ns$ inputs are distinct, so with at most $M$ initially
resident, $Q_{\mathrm{batch}}\ge(ns-M)_{+}/B$. Combining via
$\max\{x,y\}\ge(x+y)/2$ gives \eqref{batchasym}. If the cache is empty the
compulsory scan contributes $\Omega(ns/B)$; if $ns\ge\kappa M$ the linear term
absorbs the $M/B$ term.
\end{proof}
}

\subsection{Application to \algName{} Local FFT Phases}
\label{app:PACOlb}

In Phase~I every processor evaluates $s_1=N_2/p$ data-disjoint transforms of
length $n_1=N_1$; in Phase~III it evaluates $s_2=N_1/p$ of length $n_2=N_2$.
Both batches have size $n_1s_1=n_2s_2=N/p$.

\begin{theorem}[\algName{} local cache lower bound]\label{thm:36}
Under the exact slab decomposition and no arithmetic recomputation, \algName{}'s two
local FFT phases satisfy
$Q^{\mathrm{local}}_{\max}=\Omega\!\bigl(\tfrac{N}{pB}(1+\log_M N)\bigr)$,
and hence so does the complete \algName{} execution.
\end{theorem}
\begin{proof}
Fix the constant $\kappa$ of Corollary~\ref{e:cor35} (playing the role of the
constant $\epsilon$ in the Proposition~6 roadmap).

\emph{Large-cache case $N/p\ge\kappa M$.} Applying Corollary~\ref{e:cor35} to
each phase gives $\Omega\!\bigl(\tfrac{N}{pB}(1+\log_M N_1)\bigr)$ and
$\Omega\!\bigl(\tfrac{N}{pB}(1+\log_M N_2)\bigr)$; the phases run sequentially
and $\log_M N_1+\log_M N_2=\log_M N$, so their bounds sum to the claim.

\emph{Small-cache case $N/p<\kappa M$.} This is exactly the elementary case of
the Proposition~6 roadmap: $p\mid N_1$ and $p\mid N_2$ give $N_1,N_2\le N/p$,
hence $N\le(N/p)^2<\kappa^2M^2$ and $1+\log_M N=O(1)$; Phase~I begins with an
empty cache and touches $\Theta(N/p)$ distinct inputs, incurring
$\Omega(N/(pB))=\Omega\!\bigl(\tfrac{N}{pB}(1+\log_M N)\bigr)$.

The exact slab decomposition assigns identical batch sizes to every processor,
so the bound holds for the maximum per-processor cost; the fused middle stage can only add misses.
\end{proof}

\punt{
\section{Cache Lower Bounds: Deferred Proofs}
\label{app:lower-bounds}

The main text establishes three of \algName{}'s four lower bounds in full:
the one-round redistribution bound (Proposition~3), the exact migration
bound for the prescribed fused permutation (Proposition~4), and the
critical-path work bound (Proposition~5).  The local cache lower bound
(Proposition~6) was stated there with only a proof roadmap; its detailed
segment argument was deferred to this appendix.

This appendix supplies that argument.  We import the FFT-DAG
input-boundary theorem, convert boundary growth into block-cache misses,
and obtain a finite, block-aware cache lower bound for a single radix-2
FFT DAG.  We then extend the bound to disjoint batches of FFT DAGs and
specialize it to \algName{}'s two local phases, thereby completing the proof of
Proposition~6.  As a byproduct---mentioned but not needed in the main
upper-bound analysis---we also record a stronger constant-one asymptotic
miss bound for arbitrary no-recomputation schedules of the radix-2 FFT
DAG.

Throughout this appendix, arithmetic vertices of the fixed-radix
Cooley--Tukey DAG are not recomputed: every DAG vertex is evaluated at
least once and never re-evaluated.

\subsection{Imported FFT-DAG Input-Boundary Theorem}
\label{app:lower-bounds-input-boundary}

The cache lower bound uses a graph-theoretic boundary theorem for the
standard radix-2 FFT DAG.  Let \(G_2(N)\) denote the radix-2 FFT DAG on
\(N=2^m\) inputs.  It has \(m+1\) levels of \(N\) vertices and
\[
N\log_2 N
\]
noninput computation vertices.

For a set \(U\) of noninput vertices, define its input boundary by
\begin{equation}
\operatorname{in}(U)
=
\left\{
v\notin U:
\exists u\in U
\text{ such that }
(v,u)\in E
\right\}.
\label{eq:lower-bound-input-boundary-definition}
\end{equation}
Thus, \(\operatorname{in}(U)\) consists of values computed before the
execution of \(U\) that are required as direct predecessors of
vertices in \(U\).

\begin{theorem}[FFT input-boundary growth]
\label{thm:lower-bound-imported-boundary}
For every integer \(k\geq 1\), every set \(U\) of noninput vertices in
\(G_2(N)\) satisfying
\begin{equation}
|U|
\geq
k\log_2(4k)
\label{eq:lower-bound-boundary-premise}
\end{equation}
also satisfies
\begin{equation}
|\operatorname{in}(U)|
\geq
k.
\label{eq:lower-bound-boundary-conclusion}
\end{equation}
\end{theorem}

Theorem~\ref{thm:lower-bound-imported-boundary} is an imported FFT-DAG
result from the boundary-growth literature \cite{RanjanSaZu11}.  We use
only its graph-theoretic consequence and do not reproduce its
mass-propagation proof here.

\subsection{From Input Boundaries to Block Cache Misses}
\label{app:lower-bounds-boundary-to-miss}

Consider one contiguous segment of a legal topological schedule.  Let
\(U\) be the set of previously uncomputed FFT vertices evaluated during
that segment.

\begin{lemma}[Input boundary to block-cache misses]
\label{lem:lower-bound-boundary-to-miss}
If
\[
|\operatorname{in}(U)|\geq k,
\]
then the execution segment incurs at least
\begin{equation}
\left\lceil
\frac{k-M}{B}
\right\rceil_+
\label{eq:lower-bound-boundary-to-miss}
\end{equation}
cache misses, where
\[
\lceil x\rceil_+
:=
\max\{0,\lceil x\rceil\}.
\]
\end{lemma}

\begin{proof}
At the beginning of the segment, the cache contains at most \(M\)
items.  Thus, even under the most favorable possible cache state, at
most \(M\) distinct input-boundary values can already reside in cache.
At least \(k-M\) required boundary values are therefore absent whenever
\(k>M\).

One cache miss loads at most one cache block containing \(B\) items.
Even if all \(B\) loaded items are distinct missing boundary values,
one miss can supply at most \(B\) such values.  Hence at least
\[
\left\lceil
\frac{k-M}{B}
\right\rceil
\]
misses are required for \(k>M\), while the trivial lower bound is zero
for \(k\leq M\).  This yields \eqref{lower-bound-boundary-to-miss}.
\end{proof}

Lemma~\ref{lem:lower-bound-boundary-to-miss} permits a strong
algorithmic model.  In particular, the algorithm may know \(M\) and
\(B\), choose arbitrary slow-memory layouts, pack \(B\) useful boundary
values into every loaded block, retain the most useful \(M\) values in
cache, and maintain multiple slow-memory copies of previously computed
values.  None of these capabilities eliminates the need to load a
block when a required boundary value is not resident in cache.

\subsection{Finite Block-Aware Lower Bound for One FFT DAG}
\label{app:lower-bounds-finite-single}

\begin{theorem}[Finite block-aware cache lower bound]
\label{thm:lower-bound-finite-single}
For every integer \(k>M\), every no-recomputation topological schedule
of \(G_2(N)\) satisfies
\begin{equation}
Q_{\mathrm{miss}}(N,M,B)
\geq
\left\lfloor
\frac{N\log_2 N}
{k\log_2(4k)}
\right\rfloor
\left\lceil
\frac{k-M}{B}
\right\rceil.
\label{eq:lower-bound-finite-single}
\end{equation}
\end{theorem}

\begin{proof}
Partition the schedule into consecutive complete segments, each of
which computes exactly
\[
k\log_2(4k)
\]
previously uncomputed FFT vertices.  The final incomplete segment, if
any, is ignored.

The radix-2 FFT DAG contains \(N\log_2 N\) noninput computation
vertices, so the number of complete segments is at least
\[
\left\lfloor
\frac{N\log_2 N}
{k\log_2(4k)}
\right\rfloor.
\]
For every complete segment,
Theorem~\ref{thm:lower-bound-imported-boundary} gives an input boundary
of size at least \(k\).  By
Lemma~\ref{lem:lower-bound-boundary-to-miss}, each such segment incurs
at least
\[
\left\lceil
\frac{k-M}{B}
\right\rceil
\]
cache misses.  The segments are disjoint temporal intervals, so their
cache misses add, proving \eqref{lower-bound-finite-single}.
\end{proof}

\subsection{A Stronger Constant-One Asymptotic Bound}
\label{app:lower-bounds-constant-one}

This subsection is not needed for \algName{}'s main upper-bound proof.  It
records the stronger miss-only lower bound noted in the main text,
which applies to arbitrary no-recomputation schedules of the radix-2
FFT DAG.

Let \(\phi(M)\) be a positive function satisfying
\begin{equation}
\phi(M)\longrightarrow\infty,
\qquad
\log\phi(M)=o(\log M),
\label{eq:lower-bound-phi-conditions}
\end{equation}
and choose
\begin{equation}
k
=
\left\lceil
M\phi(M)
\right\rceil.
\label{eq:lower-bound-constant-one-k}
\end{equation}
Assume additionally that
\begin{equation}
M\phi(M)\log_2\!\bigl(4M\phi(M)\bigr)
=
o(N\log_2 N).
\label{eq:lower-bound-constant-one-growth}
\end{equation}

\begin{theorem}[Asymptotic constant-one cache lower bound]
\label{thm:lower-bound-constant-one}
Under \eqref{lower-bound-phi-conditions} and
\eqref{lower-bound-constant-one-growth}, every no-recomputation
topological schedule of \(G_2(N)\) satisfies
\begin{equation}
Q_{\mathrm{miss}}(N,M,B)
\geq
(1-o(1))
\frac{N}{B}\log_M N.
\label{eq:lower-bound-constant-one}
\end{equation}
\end{theorem}

\begin{proof}
Substitute \eqref{lower-bound-constant-one-k} into
Theorem~\ref{thm:lower-bound-finite-single}.  Under
\eqref{lower-bound-constant-one-growth}, the floor, ceiling, and
final incomplete segment contribute only lower-order terms.  Therefore,
\begin{align*}
Q_{\mathrm{miss}}
&\geq
(1-o(1))
\frac{N\log_2 N}
{k\log_2(4k)}
\cdot
\frac{k-M}{B}
\\
&=
(1-o(1))
\frac{N}{B}
\cdot
\frac{\log_2 N}{\log_2 M}
\cdot
\left(1-\frac{M}{k}\right)
\cdot
\frac{\log_2 M}{\log_2(4k)}.
\end{align*}
The choice \(k=\lceil M\phi(M)\rceil\) gives
\[
1-\frac{M}{k}
=
1-\frac{1}{\phi(M)}+o(1)
=
1-o(1).
\]
Moreover, \(\log\phi(M)=o(\log M)\) implies
\[
\frac{\log_2 M}{\log_2(4k)}
=
1-o(1).
\]
Finally,
\[
\frac{\log_2 N}{\log_2 M}
=
\log_M N.
\]
Substituting these estimates proves
\eqref{lower-bound-constant-one}.
\end{proof}

Theorem~\ref{thm:lower-bound-constant-one} shows that the normalized
miss-only cache lower-bound constant is at least \(1-o(1)\).  It does
not prove that \algName{}, or any currently known schedule, attains an upper
bound of the form
\[
(1+o(1))
\frac{N}{B}\log_M N.
\]

\subsection{Disjoint Unions of FFT DAGs}
\label{app:lower-bounds-disjoint-unions}

The local \algName{} phases assign each processor a batch of data-disjoint
FFTs.  Let
\begin{equation}
G_2^{(s)}(n)
=
\bigsqcup_{i=1}^{s} G_2(n)^{(i)}
\label{eq:lower-bound-disjoint-union}
\end{equation}
denote the disjoint union of \(s\) radix-2 FFT DAGs of length \(n\).

\begin{lemma}[Input-boundary growth for disjoint FFT batches]
\label{lem:lower-bound-disjoint-boundary}
For every integer \(k\geq 1\), every set \(U\) of noninput vertices in
\(G_2^{(s)}(n)\) satisfying
\begin{equation}
|U|
\geq
4k\log_2(4k)
\label{eq:lower-bound-disjoint-premise}
\end{equation}
also satisfies
\begin{equation}
|\operatorname{in}(U)|
\geq
k.
\label{eq:lower-bound-disjoint-conclusion}
\end{equation}
\end{lemma}

\begin{proof}
We prove the contrapositive.  Let
\[
U_i
=
U\cap V\bigl(G_2(n)^{(i)}\bigr),
\qquad
r_i
=
|\operatorname{in}(U_i)|.
\]
Since the component DAGs are disjoint,
\[
|\operatorname{in}(U)|
=
\sum_i r_i.
\]
Suppose
\[
|\operatorname{in}(U)|<k.
\]
For every nonempty \(U_i\), we have \(r_i\geq 1\).  Applying the
contrapositive of Theorem~\ref{thm:lower-bound-imported-boundary} with
parameter \(r_i+1\) gives
\[
|U_i|
<
(r_i+1)\log_2\!\bigl(4(r_i+1)\bigr).
\]
Because \(r_i\geq 1\),
\[
r_i+1\leq 2r_i,
\qquad
4(r_i+1)\leq 8r_i.
\]
Thus,
\[
|U_i|
<
2r_i\log_2(8r_i)
\leq
2r_i\log_2(8k).
\]
Summing over all nonempty components yields
\[
|U|
<
2\left(\sum_i r_i\right)\log_2(8k)
<
2k\log_2(8k).
\]
For every \(k\geq 1\),
\[
2\log_2(8k)
\leq
4\log_2(4k).
\]
Hence,
\[
|U|
<
4k\log_2(4k),
\]
which proves the contrapositive of
\eqref{lower-bound-disjoint-premise}--\eqref{lower-bound-disjoint-conclusion}.
\end{proof}

The constant \(4\) in Lemma~\ref{lem:lower-bound-disjoint-boundary} is
chosen for a simple uniform disjoint-union statement and is not claimed
to be tight.

\subsection{Cache Lower Bound for a Batch of Local FFTs}
\label{app:lower-bounds-batches}

\begin{theorem}[Finite lower bound for a batch of FFTs]
\label{thm:lower-bound-finite-batch}
Let one processor evaluate \(s\) data-disjoint length-\(n\) radix-2 FFT
DAGs without arithmetic recomputation.  For every integer \(k>M\),
\begin{equation}
Q_{\mathrm{batch}}(n,s)
\geq
\left\lfloor
\frac{sn\log_2 n}
{4k\log_2(4k)}
\right\rfloor
\left\lceil
\frac{k-M}{B}
\right\rceil.
\label{eq:lower-bound-finite-batch}
\end{equation}
\end{theorem}

\begin{proof}
The disjoint union contains \(sn\log_2 n\) noninput computation
vertices.  Partition the execution into complete segments, each
computing
\[
4k\log_2(4k)
\]
previously uncomputed vertices.
Lemma~\ref{lem:lower-bound-disjoint-boundary} gives an input boundary
of size at least \(k\) for every complete segment.
Lemma~\ref{lem:lower-bound-boundary-to-miss} then gives at least
\[
\left\lceil
\frac{k-M}{B}
\right\rceil
\]
cache misses per complete segment.  Summing the segment lower bounds
proves \eqref{lower-bound-finite-batch}.
\end{proof}

\begin{corollary}[Asymptotic lower bound for batched local FFTs]
\label{cor:lower-bound-batch-asymptotic}
There exist universal constants \(c,c'>0\) such that
\begin{equation}
Q_{\mathrm{batch}}(n,s)
\geq
c\frac{ns}{B}
\bigl(1+\log_M n\bigr)
-
c'\frac{M}{B}.
\label{eq:lower-bound-batch-asymptotic}
\end{equation}
If the cache is initially empty, or if
\begin{equation}
ns\geq \kappa M
\label{eq:lower-bound-batch-large-instance}
\end{equation}
for a sufficiently large universal constant \(\kappa\), then
\begin{equation}
Q_{\mathrm{batch}}(n,s)
=
\Omega\!\left(
\frac{ns}{B}
\bigl(1+\log_M n\bigr)
\right).
\label{eq:lower-bound-batch-clean}
\end{equation}
\end{corollary}

\begin{proof}
Choose \(k=2M\) in Theorem~\ref{thm:lower-bound-finite-batch}.  This
gives constants \(c_1,c_2>0\) such that
\begin{equation}
Q_{\mathrm{batch}}(n,s)
\geq
c_1\frac{ns}{B}\log_M n
-
c_2\frac{M}{B}.
\label{eq:lower-bound-batch-logarithmic}
\end{equation}

Independently, the \(sn\) inputs of the batch are distinct.  Therefore,
if at most \(M\) useful input values are initially resident in cache,
the execution must incur at least
\begin{equation}
Q_{\mathrm{batch}}(n,s)
\geq
\frac{(ns-M)_+}{B}
\label{eq:lower-bound-batch-compulsory}
\end{equation}
cache misses to access the remaining input values.  Combining
\eqref{lower-bound-batch-logarithmic} and
\eqref{lower-bound-batch-compulsory}, and using
\[
\max\{x,y\}\geq\frac{x+y}{2},
\]
gives \eqref{lower-bound-batch-asymptotic} after absorbing constants.

If the cache is initially empty, the compulsory input scan contributes
\(\Omega(ns/B)\) directly.  If instead
\eqref{lower-bound-batch-large-instance} holds for sufficiently
large \(\kappa\), then the linear term \(ns/B\) absorbs the negative
\(M/B\) term in \eqref{lower-bound-batch-asymptotic}.  In either
case, \eqref{lower-bound-batch-clean} follows.
\end{proof}

\subsection{Application to \algName{} Local FFT Phases}
\label{app:lower-bounds-\algName{}-local}

We now complete the deferred proof of the local cache lower bound
stated as Proposition~6 in the main text.

In Phase~I, every processor evaluates
\[
s_1
=
\frac{N_2}{p}
\]
data-disjoint transforms of length
\[
n_1=N_1.
\]
In Phase~III, every processor evaluates
\[
s_2
=
\frac{N_1}{p}
\]
data-disjoint transforms of length
\[
n_2=N_2.
\]
Both local batches have size
\begin{equation}
n_1s_1
=
n_2s_2
=
\frac{N}{p}.
\label{eq:lower-bound-\algName{}-local-batch-size}
\end{equation}

\begin{theorem}[\algName{} local cache lower bound]
\label{thm:lower-bound-\algName{}-local-cache}
Under the exact slab decomposition and no-arithmetic-recomputation
assumptions, \algName{}'s two local FFT phases satisfy
\begin{equation}
Q_{\mathrm{crit}}^{\mathrm{local}}
=
\Omega\!\left(
\frac{N}{pB}
\bigl(1+\log_M N\bigr)
\right).
\label{eq:lower-bound-\algName{}-local-cache}
\end{equation}
Consequently, the complete \algName{} execution satisfies the same cache
lower bound.
\end{theorem}

\begin{proof}
Fix a sufficiently large universal constant \(\kappa\).

First suppose that
\begin{equation}
\frac{N}{p}
\geq
\kappa M.
\label{eq:lower-bound-\algName{}-large-local-batch}
\end{equation}
Applying Corollary~\ref{cor:lower-bound-batch-asymptotic} separately
to the two local phases gives
\begin{align}
Q_{q,\mathrm{I}}
&=
\Omega\!\left(
\frac{N}{pB}
\bigl(1+\log_M N_1\bigr)
\right),
\label{eq:lower-bound-\algName{}-phase-I}
\\
Q_{q,\mathrm{III}}
&=
\Omega\!\left(
\frac{N}{pB}
\bigl(1+\log_M N_2\bigr)
\right).
\label{eq:lower-bound-\algName{}-phase-III}
\end{align}
The phases execute sequentially.  Since
\[
\log_M N_1+\log_M N_2
=
\log_M N,
\]
their cache lower bounds sum to
\[
\Omega\!\left(
\frac{N}{pB}
\bigl(1+\log_M N\bigr)
\right).
\]

Now suppose that
\begin{equation}
\frac{N}{p}
<
\kappa M.
\label{eq:lower-bound-\algName{}-small-local-batch}
\end{equation}
The exact divisibility assumptions imply
\[
p\mid N_1,
\qquad
p\mid N_2,
\]
and hence
\[
p\leq N_1,
\qquad
p\leq N_2.
\]
Therefore,
\[
N_1
\leq
\frac{N}{p},
\qquad
N_2
\leq
\frac{N}{p},
\]
so
\[
N
=
N_1N_2
\leq
\left(\frac{N}{p}\right)^2
<
\kappa^2M^2.
\]
It follows that
\[
1+\log_M N
=
O(1).
\]

Phase~I begins with an empty cache and accesses \(\Theta(N/p)\)
distinct local input items.  It consequently incurs
\[
Q_{q,\mathrm{I}}
=
\Omega\!\left(
\frac{N}{pB}
\right)
=
\Omega\!\left(
\frac{N}{pB}
\bigl(1+\log_M N\bigr)
\right).
\]
This proves \eqref{lower-bound-\algName{}-local-cache} in both cases.

Finally, the exact slab decomposition assigns the same local FFT batch
sizes to every processor.  Hence, the lower bound applies to the
critical path.  The fused middle stage can only add cache misses, so
the complete \algName{} execution satisfies the same lower bound.
\end{proof}

\subsection{Fixed-Radix Extension}
\label{app:lower-bounds-fixed-radix}

\begin{corollary}[Extension to fixed radix]
\label{cor:lower-bound-fixed-radix}
Let
\[
b=2^t,
\qquad
t=O(1).
\]
If every radix-\(b\) butterfly is expanded into its \(t\) constituent
radix-2 butterfly layers, then the resulting dependency graph is the
standard radix-2 FFT DAG.  Therefore, the single-transform lower bounds
of Sections~\ref{app:lower-bounds-finite-single} and
\ref{app:lower-bounds-constant-one}, the batched lower bounds of
Section~\ref{app:lower-bounds-batches}, and the \algName{} local cache lower
bound of Theorem~\ref{thm:lower-bound-\algName{}-local-cache} all apply to
the expanded fixed-radix Cooley--Tukey DAG.
\end{corollary}

If a complete radix-\(b\) DFT is instead treated as one indivisible
atomic vertex, rather than as a bounded collection of radix-2
arithmetic layers, then the graph structure and associated leading
constants differ.  Such an atomic model is outside the scope of the
lower-bound claims in this appendix.
}

\punt{
\section{Lower Bounds}
\label{app:lower-bounds}

This appendix establishes the lower bounds used by \algName{} and makes
their scopes explicit.  The claims concern distinct resources and
different classes of algorithms.

First, under the owner-computes and no-replication assumptions, every
correct parallel DFT requires at least one global communication round.
Second, \algName{} is migration-optimal for the prescribed fused
digit-reversal permutation between the prescribed source and
factor-swapped target slab distributions.  This does not assert
communication-volume optimality among all distributed FFT algorithms
with different intermediate ownership schedules.  Third, the work and
cache lower bounds concern the fixed-radix Cooley--Tukey computation
DAG under the no-arithmetic-recomputation assumption.

Throughout this appendix, input values are initially unreplicated, no
processor initially owns the complete input, and a processor may obtain
a remotely owned value only through an explicit BSP communication
superstep.  Unless stated otherwise, arithmetic vertices of the
fixed-radix Cooley--Tukey DAG are not recomputed.

\subsection{One Global Redistribution Is Necessary}
\label{app:lower-bounds-one-round}

\begin{proposition}[At least one redistribution is necessary]
\label{prop:lower-bound-one-redistribution}
Every correct parallel algorithm for the complete length-\(N\) DFT
satisfying the owner-computes and no-replication assumptions requires
at least one global redistribution round:
\begin{equation}
R\geq 1.
\label{eq:lower-bound-one-round}
\end{equation}
\end{proposition}

\begin{proof}
Every DFT coefficient depends on every input item, because every entry
of the DFT matrix \(\Omega_N\) is nonzero.  Since no processor initially
owns the complete input, every processor lacks at least one input item
initially owned by another processor.

Suppose, for contradiction, that \(R=0\).  Then no processor can obtain
a remotely owned input item.  Consequently, the local state of a
processor, and every output coefficient returned by that processor,
is independent of at least one remotely owned input value.  Changing
only that missing input value changes every correct DFT coefficient,
which contradicts correctness.  Therefore \(R\geq 1\).
\end{proof}

Proposition~\ref{prop:lower-bound-one-redistribution} is a lower bound
on the number of communication rounds required for the complete DFT
under the stated ownership model.  It is not a lower bound specific to
the fused digit-reversal permutation.

\subsection{Migration Lower Bound for the Prescribed Fused Permutation}
\label{app:lower-bounds-migration}

We next consider only the permutation materialized by \algName{}'s middle
stage:
\begin{equation}
\pi_m(a\pplus c)
=
\pi_{m_2}(c)\pplus \pi_{m_1}(a),
\label{eq:lower-bound-fused-permutation}
\end{equation}
where \(a\in[N_1]\) and \(c\in[N_2]\).  The source ownership is
\begin{equation}
\operatorname{owner}_0(a,c)
=
\left\lfloor
\frac{c}{N_2/p}
\right\rfloor,
\label{eq:lower-bound-source-owner}
\end{equation}
and the target ownership is
\begin{equation}
\operatorname{owner}_1(u,v)
=
\left\lfloor
\frac{v}{N_1/p}
\right\rfloor.
\label{eq:lower-bound-target-owner}
\end{equation}

\begin{proposition}[Exact migration lower bound for the fused permutation]
\label{prop:lower-bound-migration}
Any algorithm that materializes
\eqref{lower-bound-fused-permutation} between the ownerships
\eqref{lower-bound-source-owner} and
\eqref{lower-bound-target-owner} must communicate at least
\begin{equation}
\mu
\geq
N\left(1-\frac{1}{p}\right)
\label{eq:lower-bound-migration-volume}
\end{equation}
remote items.  Moreover,
\begin{equation}
\max_{q\in[p]}
\left\{
h_q^{\mathrm{send}},
h_q^{\mathrm{recv}}
\right\}
\geq
\frac{N}{p}\left(1-\frac{1}{p}\right).
\label{eq:lower-bound-migration-h}
\end{equation}
\end{proposition}

\begin{proof}
Fix one source column coordinate \(c\).  Its source owner is
\[
q
=
\left\lfloor
\frac{c}{N_2/p}
\right\rfloor.
\]
After the fused permutation, the target owner is determined by
\[
v=\pi_{m_1}(a).
\]
Since \(\pi_{m_1}\) is a bijection on \([N_1]\), exactly \(N_1/p\)
row coordinates \(a\) satisfy
\[
\left\lfloor
\frac{\pi_{m_1}(a)}{N_1/p}
\right\rfloor
=
q.
\]
Thus, for each fixed source column \(c\), exactly \(N_1/p\) items
retain their original processor ownership.

There are \(N_2\) possible source columns.  Hence, exactly
\[
N_2\cdot\frac{N_1}{p}
=
\frac{N}{p}
\]
items can remain on their original processor.  Every other item must
cross a processor boundary at least once.  Therefore,
\[
\mu
\geq
N-\frac{N}{p}
=
N\left(1-\frac{1}{p}\right).
\]

The aggregate remote send volume and aggregate remote receive volume
are both at least the quantity above.  Averaging either quantity over
the \(p\) processors yields
\[
\max_q h_q^{\mathrm{send}}
\geq
\frac{\mu}{p},
\qquad
\max_q h_q^{\mathrm{recv}}
\geq
\frac{\mu}{p},
\]
which proves \eqref{lower-bound-migration-h}.
\end{proof}

\begin{corollary}[Optimality of \algName{}'s middle-stage migration]
\label{cor:lower-bound-\algName{}-migration-optimal}
The fused redistribution of Appendix~\ref{app:routing} is
migration-optimal for the prescribed permutation and slab ownership
contracts.  In particular, \algName{} attains
\begin{equation}
\mu
=
N\left(1-\frac{1}{p}\right),
\qquad
h
=
\frac{N}{p}\left(1-\frac{1}{p}\right).
\label{eq:lower-bound-\algName{}-migration-optimal}
\end{equation}
\end{corollary}

The scope of Corollary~\ref{cor:lower-bound-\algName{}-migration-optimal} is
deliberately limited.  It does \emph{not} claim that \algName{} minimizes
communication volume among all distributed FFT algorithms.  Rather,
\algName{} is migration-optimal for the prescribed fused permutation between
the prescribed source and factor-swapped target slab distributions.

\subsection{Critical-Path Work Lower Bound}
\label{app:lower-bounds-work}

\begin{proposition}[Critical-path work lower bound]
\label{prop:lower-bound-work}
For the fixed-base Cooley--Tukey DAG used by \algName{}, under the
no-arithmetic-recomputation assumption,
\begin{equation}
W_{\mathrm{crit}}
=
\Omega\!\left(
\frac{N}{p}\log N
\right).
\label{eq:lower-bound-work}
\end{equation}
\end{proposition}

\begin{proof}
For a fixed radix \(b=2^t=O(1)\), the expanded Cooley--Tukey DAG
contains \(\Theta(N\log N)\) arithmetic vertices.  Since arithmetic
recomputation is disallowed, every such vertex must be evaluated at
least once.  Therefore, by averaging over the \(p\) processors, at
least one processor executes
\[
\Omega\!\left(
\frac{N\log N}{p}
\right)
\]
arithmetic operations.  This processor lies on the critical path,
which proves \eqref{lower-bound-work}.
\end{proof}

This is a lower bound for the specified fixed-radix Cooley--Tukey
computation DAG without arithmetic recomputation.  It is not an
unrestricted arithmetic lower bound for every possible DFT algorithm.

\subsection{Imported FFT-DAG Input-Boundary Theorem}
\label{app:lower-bounds-input-boundary}

The cache lower bound uses a graph-theoretic boundary theorem for the
standard radix-2 FFT DAG.  Let \(G_2(N)\) denote the radix-2 FFT DAG on
\(N=2^m\) inputs.  It has \(m+1\) levels of \(N\) vertices and
\[
N\log_2 N
\]
noninput computation vertices.

For a set \(U\) of noninput vertices, define its input boundary by
\begin{equation}
\operatorname{in}(U)
=
\left\{
v\notin U:
\exists u\in U
\text{ such that }
(v,u)\in E
\right\}.
\label{eq:lower-bound-input-boundary-definition}
\end{equation}
Thus, \(\operatorname{in}(U)\) consists of values computed before the
execution of \(U\) that are required as direct predecessors of
vertices in \(U\).

\begin{theorem}[FFT input-boundary growth]
\label{thm:lower-bound-imported-boundary}
For every integer \(k\geq 1\), every set \(U\) of noninput vertices in
\(G_2(N)\) satisfying
\begin{equation}
|U|
\geq
k\log_2(4k)
\label{eq:lower-bound-boundary-premise}
\end{equation}
also satisfies
\begin{equation}
|\operatorname{in}(U)|
\geq
k.
\label{eq:lower-bound-boundary-conclusion}
\end{equation}
\end{theorem}

Theorem~\ref{thm:lower-bound-imported-boundary} is an imported
FFT-DAG result from the boundary-growth literature \cite{RanjanSaZu11}.  We use only its
graph-theoretic consequence and do not reproduce its
mass-propagation proof here.

\subsection{From Input Boundaries to Block Cache Misses}
\label{app:lower-bounds-boundary-to-miss}

Consider one contiguous segment of a legal topological schedule.  Let
\(U\) be the set of previously uncomputed FFT vertices evaluated during
that segment.

\begin{lemma}[Input boundary to block-cache misses]
\label{lem:lower-bound-boundary-to-miss}
If
\[
|\operatorname{in}(U)|\geq k,
\]
then the execution segment incurs at least
\begin{equation}
\left\lceil
\frac{k-M}{B}
\right\rceil_+
\label{eq:lower-bound-boundary-to-miss}
\end{equation}
cache misses, where
\[
\lceil x\rceil_+
:=
\max\{0,\lceil x\rceil\}.
\]
\end{lemma}

\begin{proof}
At the beginning of the segment, the cache contains at most \(M\)
items.  Thus, even under the most favorable possible cache state, at
most \(M\) distinct input-boundary values can already reside in cache.
At least \(k-M\) required boundary values are therefore absent whenever
\(k>M\).

One cache miss loads at most one cache block containing \(B\) items.
Even if all \(B\) loaded items are distinct missing boundary values,
one miss can supply at most \(B\) such values.  Hence at least
\[
\left\lceil
\frac{k-M}{B}
\right\rceil
\]
misses are required for \(k>M\), while the trivial lower bound is zero
for \(k\leq M\).  This yields
\eqref{lower-bound-boundary-to-miss}.
\end{proof}

Lemma~\ref{lem:lower-bound-boundary-to-miss} permits a strong
algorithmic model.  In particular, the algorithm may know \(M\) and
\(B\), choose arbitrary slow-memory layouts, pack \(B\) useful boundary
values into every loaded block, retain the most useful \(M\) values in
cache, and maintain multiple slow-memory copies of previously computed
values.  None of these capabilities eliminates the need to load a
block when a required boundary value is not resident in cache.

\subsection{Finite Block-Aware Lower Bound for One FFT DAG}
\label{app:lower-bounds-finite-single}

\begin{theorem}[Finite block-aware cache lower bound]
\label{thm:lower-bound-finite-single}
For every integer \(k>M\), every no-recomputation topological schedule
of \(G_2(N)\) satisfies
\begin{equation}
Q_{\mathrm{miss}}(N,M,B)
\geq
\left\lfloor
\frac{N\log_2 N}
{k\log_2(4k)}
\right\rfloor
\left\lceil
\frac{k-M}{B}
\right\rceil.
\label{eq:lower-bound-finite-single}
\end{equation}
\end{theorem}

\begin{proof}
Partition the schedule into consecutive complete segments, each of
which computes exactly
\[
k\log_2(4k)
\]
previously uncomputed FFT vertices.  The final incomplete segment, if
any, is ignored.

The radix-2 FFT DAG contains \(N\log_2 N\) noninput computation
vertices, so the number of complete segments is at least
\[
\left\lfloor
\frac{N\log_2 N}
{k\log_2(4k)}
\right\rfloor.
\]
For every complete segment, Theorem~\ref{thm:lower-bound-imported-boundary}
gives an input boundary of size at least \(k\).  By
Lemma~\ref{lem:lower-bound-boundary-to-miss}, each such segment incurs
at least
\[
\left\lceil
\frac{k-M}{B}
\right\rceil
\]
cache misses.  The segments are disjoint temporal intervals, so their
cache misses add, proving \eqref{lower-bound-finite-single}.
\end{proof}

\subsection{A Stronger Constant-One Asymptotic Bound}
\label{app:lower-bounds-constant-one}

This subsection is not needed for \algName{}'s main upper-bound proof.  It
records a stronger miss-only lower bound that applies to arbitrary
no-recomputation schedules of the radix-2 FFT DAG.

Let \(\phi(M)\) be a positive function satisfying
\begin{equation}
\phi(M)\longrightarrow\infty,
\qquad
\log\phi(M)=o(\log M),
\label{eq:lower-bound-phi-conditions}
\end{equation}
and choose
\begin{equation}
k
=
\left\lceil
M\phi(M)
\right\rceil.
\label{eq:lower-bound-constant-one-k}
\end{equation}
Assume additionally that
\begin{equation}
M\phi(M)\log_2\!\bigl(4M\phi(M)\bigr)
=
o(N\log_2 N).
\label{eq:lower-bound-constant-one-growth}
\end{equation}

\begin{theorem}[Asymptotic constant-one cache lower bound]
\label{thm:lower-bound-constant-one}
Under \eqref{lower-bound-phi-conditions} and
\eqref{lower-bound-constant-one-growth}, every no-recomputation
topological schedule of \(G_2(N)\) satisfies
\begin{equation}
Q_{\mathrm{miss}}(N,M,B)
\geq
(1-o(1))
\frac{N}{B}\log_M N.
\label{eq:lower-bound-constant-one}
\end{equation}
\end{theorem}

\begin{proof}
Substitute \eqref{lower-bound-constant-one-k} into
Theorem~\ref{thm:lower-bound-finite-single}.  Under
\eqref{lower-bound-constant-one-growth}, the floor, ceiling, and
final incomplete segment contribute only lower-order terms.  Therefore,
\begin{align*}
Q_{\mathrm{miss}}
&\geq
(1-o(1))
\frac{N\log_2 N}
{k\log_2(4k)}
\cdot
\frac{k-M}{B}
\\
&=
(1-o(1))
\frac{N}{B}
\cdot
\frac{\log_2 N}{\log_2 M}
\cdot
\left(1-\frac{M}{k}\right)
\cdot
\frac{\log_2 M}{\log_2(4k)}.
\end{align*}
The choice \(k=\lceil M\phi(M)\rceil\) gives
\[
1-\frac{M}{k}
=
1-\frac{1}{\phi(M)}+o(1)
=
1-o(1).
\]
Moreover, \(\log\phi(M)=o(\log M)\) implies
\[
\frac{\log_2 M}{\log_2(4k)}
=
1-o(1).
\]
Finally,
\[
\frac{\log_2 N}{\log_2 M}
=
\log_M N.
\]
Substituting these estimates proves
\eqref{lower-bound-constant-one}.
\end{proof}

Theorem~\ref{thm:lower-bound-constant-one} shows that the normalized
miss-only cache lower-bound constant is at least \(1-o(1)\).  It does
not prove that \algName{}, or any currently known schedule, attains an upper
bound of the form
\[
(1+o(1))
\frac{N}{B}\log_M N.
\]

\subsection{Disjoint Unions of FFT DAGs}
\label{app:lower-bounds-disjoint-unions}

The local \algName{} phases assign each processor a batch of data-disjoint
FFTs.  Let
\begin{equation}
G_2^{(s)}(n)
=
\bigsqcup_{i=1}^{s} G_2(n)^{(i)}
\label{eq:lower-bound-disjoint-union}
\end{equation}
denote the disjoint union of \(s\) radix-2 FFT DAGs of length \(n\).

\begin{lemma}[Input-boundary growth for disjoint FFT batches]
\label{lem:lower-bound-disjoint-boundary}
For every integer \(k\geq 1\), every set \(U\) of noninput vertices in
\(G_2^{(s)}(n)\) satisfying
\begin{equation}
|U|
\geq
4k\log_2(4k)
\label{eq:lower-bound-disjoint-premise}
\end{equation}
also satisfies
\begin{equation}
|\operatorname{in}(U)|
\geq
k.
\label{eq:lower-bound-disjoint-conclusion}
\end{equation}
\end{lemma}

\begin{proof}
We prove the contrapositive.  Let
\[
U_i
=
U\cap V\bigl(G_2(n)^{(i)}\bigr),
\qquad
r_i
=
|\operatorname{in}(U_i)|.
\]
Since the component DAGs are disjoint,
\[
|\operatorname{in}(U)|
=
\sum_i r_i.
\]
Suppose
\[
|\operatorname{in}(U)|<k.
\]
For every nonempty \(U_i\), we have \(r_i\geq 1\).  Applying the
contrapositive of Theorem~\ref{thm:lower-bound-imported-boundary} with
parameter \(r_i+1\) gives
\[
|U_i|
<
(r_i+1)\log_2\!\bigl(4(r_i+1)\bigr).
\]
Because \(r_i\geq 1\),
\[
r_i+1\leq 2r_i,
\qquad
4(r_i+1)\leq 8r_i.
\]
Thus,
\[
|U_i|
<
2r_i\log_2(8r_i)
\leq
2r_i\log_2(8k).
\]
Summing over all nonempty components yields
\[
|U|
<
2\left(\sum_i r_i\right)\log_2(8k)
<
2k\log_2(8k).
\]
For every \(k\geq 1\),
\[
2\log_2(8k)
\leq
4\log_2(4k).
\]
Hence,
\[
|U|
<
4k\log_2(4k),
\]
which proves the contrapositive of
\eqref{lower-bound-disjoint-premise}--\eqref{lower-bound-disjoint-conclusion}.
\end{proof}

The constant \(4\) in Lemma~\ref{lem:lower-bound-disjoint-boundary} is
chosen for a simple uniform disjoint-union statement and is not claimed
to be tight.

\subsection{Cache Lower Bound for a Batch of Local FFTs}
\label{app:lower-bounds-batches}

\begin{theorem}[Finite lower bound for a batch of FFTs]
\label{thm:lower-bound-finite-batch}
Let one processor evaluate \(s\) data-disjoint length-\(n\) radix-2 FFT
DAGs without arithmetic recomputation.  For every integer \(k>M\),
\begin{equation}
Q_{\mathrm{batch}}(n,s)
\geq
\left\lfloor
\frac{sn\log_2 n}
{4k\log_2(4k)}
\right\rfloor
\left\lceil
\frac{k-M}{B}
\right\rceil.
\label{eq:lower-bound-finite-batch}
\end{equation}
\end{theorem}

\begin{proof}
The disjoint union contains \(sn\log_2 n\) noninput computation
vertices.  Partition the execution into complete segments, each
computing
\[
4k\log_2(4k)
\]
previously uncomputed vertices.  Lemma~\ref{lem:lower-bound-disjoint-boundary}
gives an input boundary of size at least \(k\) for every complete
segment.  Lemma~\ref{lem:lower-bound-boundary-to-miss} then gives at
least
\[
\left\lceil
\frac{k-M}{B}
\right\rceil
\]
cache misses per complete segment.  Summing the segment lower bounds
proves \eqref{lower-bound-finite-batch}.
\end{proof}

\begin{corollary}[Asymptotic lower bound for batched local FFTs]
\label{cor:lower-bound-batch-asymptotic}
There exist universal constants \(c,c'>0\) such that
\begin{equation}
Q_{\mathrm{batch}}(n,s)
\geq
c\frac{ns}{B}
\bigl(1+\log_M n\bigr)
-
c'\frac{M}{B}.
\label{eq:lower-bound-batch-asymptotic}
\end{equation}
If the cache is initially empty, or if
\begin{equation}
ns\geq \kappa M
\label{eq:lower-bound-batch-large-instance}
\end{equation}
for a sufficiently large universal constant \(\kappa\), then
\begin{equation}
Q_{\mathrm{batch}}(n,s)
=
\Omega\!\left(
\frac{ns}{B}
\bigl(1+\log_M n\bigr)
\right).
\label{eq:lower-bound-batch-clean}
\end{equation}
\end{corollary}

\begin{proof}
Choose \(k=2M\) in Theorem~\ref{thm:lower-bound-finite-batch}.  This
gives constants \(c_1,c_2>0\) such that
\begin{equation}
Q_{\mathrm{batch}}(n,s)
\geq
c_1\frac{ns}{B}\log_M n
-
c_2\frac{M}{B}.
\label{eq:lower-bound-batch-logarithmic}
\end{equation}

Independently, the \(sn\) inputs of the batch are distinct.  Therefore,
if at most \(M\) useful input values are initially resident in cache,
the execution must incur at least
\begin{equation}
Q_{\mathrm{batch}}(n,s)
\geq
\frac{(ns-M)_+}{B}
\label{eq:lower-bound-batch-compulsory}
\end{equation}
cache misses to access the remaining input values.  Combining
\eqref{lower-bound-batch-logarithmic} and
\eqref{lower-bound-batch-compulsory}, and using
\[
\max\{x,y\}\geq\frac{x+y}{2},
\]
gives \eqref{lower-bound-batch-asymptotic} after absorbing constants.

If the cache is initially empty, the compulsory input scan contributes
\(\Omega(ns/B)\) directly.  If instead
\eqref{lower-bound-batch-large-instance} holds for sufficiently
large \(\kappa\), then the linear term \(ns/B\) absorbs the negative
\(M/B\) term in \eqref{lower-bound-batch-asymptotic}.  In either
case, \eqref{lower-bound-batch-clean} follows.
\end{proof}

\subsection{Application to \algName{} Local FFT Phases}
\label{app:lower-bounds-\algName{}-local}

In Phase~I, every processor evaluates
\[
s_1
=
\frac{N_2}{p}
\]
data-disjoint transforms of length
\[
n_1=N_1.
\]
In Phase~III, every processor evaluates
\[
s_2
=
\frac{N_1}{p}
\]
data-disjoint transforms of length
\[
n_2=N_2.
\]
Both local batches have size
\begin{equation}
n_1s_1
=
n_2s_2
=
\frac{N}{p}.
\label{eq:lower-bound-\algName{}-local-batch-size}
\end{equation}

\begin{theorem}[\algName{} local cache lower bound]
\label{thm:lower-bound-\algName{}-local-cache}
Under the exact slab decomposition and no-arithmetic-recomputation
assumptions, \algName{}'s two local FFT phases satisfy
\begin{equation}
Q_{\mathrm{crit}}^{\mathrm{local}}
=
\Omega\!\left(
\frac{N}{pB}
\bigl(1+\log_M N\bigr)
\right).
\label{eq:lower-bound-\algName{}-local-cache}
\end{equation}
Consequently, the complete \algName{} execution satisfies the same cache
lower bound.
\end{theorem}

\begin{proof}
Fix a sufficiently large universal constant \(\kappa\).

First suppose that
\begin{equation}
\frac{N}{p}
\geq
\kappa M.
\label{eq:lower-bound-\algName{}-large-local-batch}
\end{equation}
Applying Corollary~\ref{cor:lower-bound-batch-asymptotic} separately
to the two local phases gives
\begin{align}
Q_{q,\mathrm{I}}
&=
\Omega\!\left(
\frac{N}{pB}
\bigl(1+\log_M N_1\bigr)
\right),
\label{eq:lower-bound-\algName{}-phase-I}
\\
Q_{q,\mathrm{III}}
&=
\Omega\!\left(
\frac{N}{pB}
\bigl(1+\log_M N_2\bigr)
\right).
\label{eq:lower-bound-\algName{}-phase-III}
\end{align}
The phases execute sequentially.  Since
\[
\log_M N_1+\log_M N_2
=
\log_M N,
\]
their cache lower bounds sum to
\[
\Omega\!\left(
\frac{N}{pB}
\bigl(1+\log_M N\bigr)
\right).
\]

Now suppose that
\begin{equation}
\frac{N}{p}
<
\kappa M.
\label{eq:lower-bound-\algName{}-small-local-batch}
\end{equation}
The exact divisibility assumptions imply
\[
p\mid N_1,
\qquad
p\mid N_2,
\]
and hence
\[
p\leq N_1,
\qquad
p\leq N_2.
\]
Therefore,
\[
N_1
\leq
\frac{N}{p},
\qquad
N_2
\leq
\frac{N}{p},
\]
so
\[
N
=
N_1N_2
\leq
\left(\frac{N}{p}\right)^2
<
\kappa^2M^2.
\]
It follows that
\[
1+\log_M N
=
O(1).
\]

Phase~I begins with an empty cache and accesses
\(\Theta(N/p)\) distinct local input items.  It consequently incurs
\[
Q_{q,\mathrm{I}}
=
\Omega\!\left(
\frac{N}{pB}
\right)
=
\Omega\!\left(
\frac{N}{pB}
\bigl(1+\log_M N\bigr)
\right).
\]
This proves \eqref{lower-bound-\algName{}-local-cache} in both cases.

Finally, the exact slab decomposition assigns the same local FFT batch
sizes to every processor.  Hence, the lower bound applies to the
critical path.  The fused middle stage can only add cache misses, so
the complete \algName{} execution satisfies the same lower bound.
\end{proof}

\subsection{Fixed-Radix Extension}
\label{app:lower-bounds-fixed-radix}

\begin{corollary}[Extension to fixed radix]
\label{cor:lower-bound-fixed-radix}
Let
\[
b=2^t,
\qquad
t=O(1).
\]
If every radix-\(b\) butterfly is expanded into its \(t\) constituent
radix-2 butterfly layers, then the resulting dependency graph is the
standard radix-2 FFT DAG.  Therefore, the single-transform lower bounds
of Sections~\ref{app:lower-bounds-finite-single} and
\ref{app:lower-bounds-constant-one}, the batched lower bounds of
Section~\ref{app:lower-bounds-batches}, and the \algName{} local cache lower
bound of Theorem~\ref{thm:lower-bound-\algName{}-local-cache} all apply to
the expanded fixed-radix Cooley--Tukey DAG.
\end{corollary}

If a complete radix-\(b\) DFT is instead treated as one indivisible
atomic vertex, rather than as a bounded collection of radix-2
arithmetic layers, then the graph structure and associated leading
constants differ.  Such an atomic model is outside the scope of the
lower-bound claims in this appendix.
}
	
\end{document}